\documentclass[11pt]{article}

\usepackage[T1]{fontenc}
\usepackage{lmodern}
\usepackage{hyperref}
\usepackage[letterpaper,margin=1.00in]{geometry}
\usepackage{amsmath, amssymb, amsthm, thmtools, amsfonts}
\usepackage{bbm}
\usepackage{appendix}
\usepackage{graphicx}
\usepackage{color}
\usepackage{algorithm}
\usepackage[noend]{algpseudocode}
\usepackage{xspace}
\usepackage{mathtools}
\usepackage{todonotes}
\usepackage{booktabs}
\usepackage{enumerate}
\usepackage{enumitem}
\newcommand{\ra}[1]{\renewcommand{\arraystretch}{#1}}
\usepackage{tikz}
\usetikzlibrary{math,calc}
\usepackage{caption}
\usepackage{subcaption}
\usepackage{tabularx}
\usepackage{tablefootnote}

\usepackage[
    style=numeric,
    citestyle=numeric,
    giveninits = true,
    url=false,
    doi=true,
]{biblatex}
\usepackage{tikz}
\usetikzlibrary{math}
\tikzmath{
    \abrcversion = 0.01;
    \scale = 1.0;
    \mirror = 0;
    \pins = 0;
    \lr = 0;
    \lg = 0;
    \lb = 0;
    \fr = 255;
    \fg = 255;
    \fb = 255;
    let \line = solid;
    function Amoebot(\x, \y, \z) {
        \t = \scale/2;
        \h = sqrt(\scale^2 - (\scale/2)^2);
        \tx = (\x+\z) * 3*\t;
        \ty = (\y+0.5*\x-0.5*\z) * 2*\h;
        if \mirror > 0 then {
            {
                \filldraw[color=black, fill={rgb,255:red,\fr; green,\fg; blue,\fb}]
                (0 + \ty,\t + \tx)
                --
                (0 + \ty,3*\t + \tx)
                --
                (\h + \ty,4*\t + \tx)
                --
                (2*\h + \ty,3*\t + \tx)
                --
                (2*\h + \ty,\t + \tx)
                --
                (\h + \ty,0 + \tx)
                --
                cycle;
            };
        } else {
            {
                \filldraw[color=black, fill={rgb,255:red,\fr; green,\fg; blue,\fb}]
                (\t + \tx,0 + \ty)
                --
                (3*\t + \tx,0 + \ty)
                --
                (4*\t + \tx,\h + \ty)
                --
                (3*\t + \tx,2*\h + \ty)
                --
                (\t + \tx,2*\h + \ty)
                --
                (0 + \tx,\h + \ty)
                --
                cycle;
            };
        };
        if \pins > 0 then {
            for \p in {1,...,\pins}{
                \q = \p / (\pins+1);
                \xp = (1-\q) * 0 + \q * \t + \tx;
                \ypn = (1-\q) * \h + \q * 2*\h + \ty;
                \yps = (1-\q) * \h + \q * 0 + \ty;
                if \mirror > 0 then {
                    {
                        \filldraw (\ypn, \xp) circle (0.04); % SE
                        \filldraw (\yps, \xp) circle (0.04); % SW
                    };
                } else {
                    {
                        \filldraw (\xp, \ypn) circle (0.04); % NW
                        \filldraw (\xp, \yps) circle (0.04); % SW
                    };
                };
                \xp = (1-\q) * \t + \q * 3*\t + \tx;
                \ypn = 2*\h + \ty;
                \yps = 0 + \ty;
                if \mirror > 0 then {
                    {
                        \filldraw (\ypn, \xp) circle (0.04); % E
                        \filldraw (\yps, \xp) circle (0.04); % W
                    };
                } else {
                    {
                        \filldraw (\xp, \ypn) circle (0.04); % N
                        \filldraw (\xp, \yps) circle (0.04); % S
                    };
                };
                \xp = (1-\q) * 3*\t + \q * 4*\t + \tx;
                \ypn = (1-\q) * 2*\h + \q * \h + \ty;
                \yps = (1-\q) * 0 + \q * \h + \ty;
                if \mirror > 0 then {
                    {
                        \filldraw (\ypn, \xp) circle (0.04); % NE
                        \filldraw (\yps, \xp) circle (0.04); % NW
                    };
                } else {
                    {
                        \filldraw (\xp, \ypn) circle (0.04); % NE
                        \filldraw (\xp, \yps) circle (0.04); % SE
                    };
                };
            };
        };
    };
    function Circuit(\x, \y, \z, \ba, \pa, \bb, \pb) {
        \t = \scale/2;
        \h = sqrt(\scale^2 - (\scale/2)^2);
        \tx = (\x+\z) * 3*\t;
        \ty = (\y+0.5*\x-0.5*\z) * 2*\h;
        if \ba > 0 && \ba < 7 && \bb > 0 && \bb < 7 then {
            if \pa > 0 && \pa <= \pins && \pb > 0 && \pb <= \pins then {
                \q = \pa / (\pins+1);
                if \ba == 1 then {
                    \pax = (1-\q) * \t + \q * 3*\t + \tx;
                    \pay = 0 + \ty;
                };
                if \ba == 2 then {
                    \pax = (1-\q) * 3*\t + \q * 4*\t + \tx;
                    \pay = (1-\q) * 0 + \q * \h + \ty;
                };
                if \ba == 3 then {
                    \pax = (1-\q) * 3*\t + \q * 4*\t + \tx;
                    \pay = (1-\q) * 2*\h + \q * \h + \ty;
                };
                if \ba == 4 then {
                    \pax = (1-\q) * \t + \q * 3*\t + \tx;
                    \pay = 2*\h + \ty;
                };
                if \ba == 5 then {
                    \pax = (1-\q) * 0 + \q * \t + \tx;
                    \pay = (1-\q) * \h + \q * 2*\h + \ty;
                };
                if \ba == 6 then {
                    \pax = (1-\q) * 0 + \q * \t + \tx;
                    \pay = (1-\q) * \h + \q * 0 + \ty;
                };
                \q = \pb / (\pins+1);
                if \bb == 1 then {
                    \pbx = (1-\q) * \t + \q * 3*\t + \tx;
                    \pby = 0 + \ty;
                };
                if \bb == 2 then {
                    \pbx = (1-\q) * 3*\t + \q * 4*\t + \tx;
                    \pby = (1-\q) * 0 + \q * \h + \ty;
                };
                if \bb == 3 then {
                    \pbx = (1-\q) * 3*\t + \q * 4*\t + \tx;
                    \pby = (1-\q) * 2*\h + \q * \h + \ty;
                };
                if \bb == 4 then {
                    \pbx = (1-\q) * \t + \q * 3*\t + \tx;
                    \pby = 2*\h + \ty;
                };
                if \bb == 5 then {
                    \pbx = (1-\q) * 0 + \q * \t + \tx;
                    \pby = (1-\q) * \h + \q * 2*\h + \ty;
                };
                if \bb == 6 then {
                    \pbx = (1-\q) * 0 + \q * \t + \tx;
                    \pby = (1-\q) * \h + \q * 0 + \ty;
                };
                if \mirror > 0 then {
                    {
                        \draw[color={rgb,255:red,\lr; green,\lg; blue,\lb}, \line]
                        (\pay,\pax) -- (\pby,\pbx);
                    };
                } else {
                    {
                        \draw[color={rgb,255:red,\lr; green,\lg; blue,\lb}, \line]
                        (\pax,\pay) -- (\pbx,\pby);
                    };
                };
            };
        };
    };
    function Label(\x, \y, \z) {
        \t = \scale/2;
        \h = sqrt(\scale^2 - (\scale/2)^2);
        \tx = (\x+\z) * 3*\t;
        \ty = (\y+0.5*\x-0.5*\z) * 2*\h;
        if \pins > 0 then {
            for \p in {1,...,\pins}{
                integer \pinnr;
                \pinnr = \p;
                \q = \p / (\pins+1);
                \xp = (1-\q) * 0 + \q * \t + \tx;
                \ypn = (1-\q) * \h + \q * 2*\h + \ty;
                \yps = (1-\q) * \h + \q * 0 + \ty;
                if \mirror > 0 then {
                    {
                        \node[below right] at (\ypn, \xp) {\pinnr}; % SE
                        \node[below left] at (\yps, \xp) {\pinnr}; % SW
                    };
                } else {
                    {
                        \node[above left] at (\xp, \ypn) {\pinnr}; % NW
                        \node[below left] at (\xp, \yps) {\pinnr}; % SW
                    };
                };
                \xp = (1-\q) * \t + \q * 3*\t + \tx;
                \ypn = 2*\h + \ty;
                \yps = 0 + \ty;
                if \mirror > 0 then {
                    {
                        \node[right] at (\ypn + .1, \xp) {\pinnr}; % E
                        \node[left] at (\yps - .1, \xp) {\pinnr}; % W
                    };
                } else {
                    {
                        \node[above] at (\xp, \ypn + .1) {\pinnr}; % N
                        \node[below] at (\xp, \yps - .1) {\pinnr}; % S
                    };
                };
                \xp = (1-\q) * 3*\t + \q * 4*\t + \tx;
                \ypn = (1-\q) * 2*\h + \q * \h + \ty;
                \yps = (1-\q) * 0 + \q * \h + \ty;
                if \mirror > 0 then {
                    {
                        \node[above right] at (\ypn, \xp) {\pinnr}; % NE
                        \node[above left] at (\yps, \xp) {\pinnr}; % NW
                    };
                } else {
                    {
                        \node[above right] at (\xp, \ypn) {\pinnr}; % NE
                        \node[below right] at (\xp, \yps) {\pinnr}; % SE
                    };
                };
            };
        };
    };
}

\tikzmath{
    function heightEqt(\x) { % height of an equilateral triangle
        return sqrt(\x^2 - (\x/2)^2);
    };
}
\newtheorem{theorem}{Theorem}[section]
\newtheorem{lemma}[theorem]{Lemma}

\newtheorem{meta-theorem}[theorem]{Meta-Theorem}

\newtheorem{remark}[theorem]{Remark}
\newtheorem{corollary}[theorem]{Corollary}

\newtheorem{observation}[theorem]{Observation}

\definecolor{darkgreen}{rgb}{0,0.5,0}
\definecolor{darkblue}{rgb}{0,0,0.5}
\usepackage{hyperref}
\hypersetup{
    unicode=false,
    colorlinks=true,
    linkcolor=darkblue,
    citecolor=darkgreen,
    filecolor=magenta,
    urlcolor=cyan
}
\usepackage[capitalize, nameinlink]{cleveref}
\crefname{theorem}{Theorem}{Theorems}
\Crefname{lemma}{Lemma}{Lemmas}
\Crefname{fact}{Fact}{Facts}
\Crefname{observation}{Observation}{Observations}
\Crefname{remark}{Remark}{Remarks}
\Crefname{invariant}{Invariant}{Invariants}
\Crefname{equation}{}{}

\newcommand{\s}{\mathit{State}}
\newcommand{\inactive}{\mathit{Inactive}}
\newcommand{\beep}{\mathit{Beep}}
\newcommand{\pulse}{\mathit{Pulse}}
\newcommand{\sleep}{\mathit{Sleep}}
\newcommand{\listen}{\mathit{Listen}}
\newcommand{\induced}{\mathit{Induced}}
\newcommand{\lock}{\mathit{Lock}}
\renewcommand{\paragraph}[1]{\vspace{0.15cm}\noindent {\bf #1}.}

\newcommand{\N}{\mathbb{N}}
\newcommand{\Z}{\mathbb{Z}}
\newcommand{\Q}{\mathbb{Q}}
\newcommand{\R}{\mathbb{R}}
\newcommand{\E}{\mathbb{E}}
\DeclareMathOperator{\lcm}{lcm}

\newcommand{\mfi}[1]{\todo[inline, color=yellow!50]{Michael: #1}}
\newcommand{\mf}[1]{\todo[color=yellow!50]{Michael: #1}}
\newcommand{\api}[1]{\todo[inline, color=red!35]{Andreas: #1}}
\newcommand{\ap}[1]{\todo[color=red!35]{Andreas: #1}}

\title{Accelerating Amoebots via Reconfigurable Circuits \footnote{This work has been supported by the DFG Project SFB 901 (On-The-Fly Computing) and the DFG Project SCHE 1592/6-1 (PROGMATTER).
}
}

\author{
  Michael Feldmann\\
  \small Paderborn University \\
  \small michael.feldmann@upb.de\\
  \and
  Andreas Padalkin \\
  \small Paderborn University \\
  \small andreas.padalkin@upb.de\\
  \and
  Christian Scheideler\\
  \small Paderborn University \\
  \small scheideler@upb.de\\
  \and
  Shlomi Dolev\\
  \small Ben-Gurion University of the Negev\\
  \small dolev@cs.bgu.ac.il
}

\date{}

\begin{document}

\begin{titlepage}

\maketitle
\thispagestyle{empty}

\begin{abstract}
    We consider an extension to the geometric amoebot model that allows amoebots to form so-called \emph{circuits}.
    Given a connected amoebot structure, a circuit is a subgraph formed by the amoebots that permits the instant transmission of signals.
    We show that such an extension allows for significantly faster solutions to a variety of problems related to programmable matter.
    More specifically, we provide algorithms for leader election, consensus, compass alignment, chirality agreement and shape recognition.
    Leader election can be solved in $\Theta(\log n)$ rounds, w.h.p., consensus in $O(1)$ rounds and both, compass alignment and chirality agreement, can be solved in $O(\log n)$ rounds, w.h.p.
    For shape recognition, the amoebots have to decide whether the amoebot structure forms a particular shape.
    We show how the amoebots can detect a parallelogram with linear and polynomial side ratio within $\Theta(\log{n})$ rounds, w.h.p.
    Finally, we show that the amoebots can detect a shape composed of triangles within $O(1)$ rounds, w.h.p.
\end{abstract}

\end{titlepage}

\maketitle

\section{Introduction}
\label{sec:intro}

Programmable matter is a physical substance consisting of tiny, homogeneous robots (also called \emph{particles}) that is able to dynamically change its physical properties like shape or density.
Such a substance can be deployed, for example, for minimal invasive surgeries through injection into the human body (detecting cancer cells, repairing bones, closing blood vessels, etc.).
Programmable matter has been envisioned for 30 years~\cite{DBLP:journals/ijhsc/ToffoliM93} and is yet still to be realized in practice.
However, theoretical investigation on various models (such as the self-assembly model~\cite{DBLP:conf/stoc/RothemundW00}, the nubot model~\cite{10.1145/2422436.2422476} or the geometric amoebot model~\cite{DBLP:conf/spaa/DerakhshandehDGRSS14}) has already been started and is still continuing in the distributed computing community.

Shape formation algorithms are of particular interest.
Algorithms of polylogarithmic complexity are known for the nubot model \cite{10.1145/2422436.2422476}.
However, these assume particles on the molecular scale since it requires the rotation of entire substructures.
Due to the acting forces, this would not be possible on the micro or macro scale.
In contrast, many problems for the geometric amoebot model come with a natural lower bound of $\Omega(D)$, where $D$ is the diameter of the structure formed by the amoebots.
Problems like leader election require information about the entire structure to be collected in at least one particle so that a particle can declare itself to be a unique leader without running into the situation of having multiple particles declaring leadership.
The lower bound then results from the fact that information only travels particle by particle.
The main goal of our research is to formulate a model that is able to break this lower bound while still being reasonable on the micro or even macro scale.

Many of the various models for programmable matter take their inspiration from nature.
For example, the particles of the nubot model resemble molecules, and the locomotion of the particles of the amoebot model are inspired by amoeba.
However, many more fascinating forms of locomotion can be found in nature.
Our model is motivated by the muscular system.
Muscles are composed of muscle fibers, which can be stimulated to perform coordinated contractions.
These contractions (and their counterpart relaxations) allow for fast locomotion.
The stimuli are inflicted by the nervous system.
The nervous system consists of highly connected nerves.
These are able to rapidly transmit primitive signals (the stimuli) over long distances.
Our aim is to come up with a model for programmable matter incorporating both concepts: the muscular system and the nervous system.

Instead of proposing an entirely new model, we build our model on top of the geometric amoebot model.
This model is predestined for our purpose since it already provides contractions (and expansions) on a small scale of single particles.
Inspired by the nervous system described above, in this paper, we introduce reconfigurable circuits to the geometric amoebot model.
Each particle is allowed to create a constant amount of circuits with a subset of the particle structure.
A circuit formed by particles allows for the instantaneous transmission of primitive signals to all of these.

We start our investigation of the new possibilities on forming circuits in a stationary setting, i.e., for the case that the amoebots do not move.
More precisely, we study algorithms for \emph{leader election}, \emph{consensus}, \emph{compass alignment}, \emph{chirality agreement}, and \emph{shape recognition}.
We show that we are able to achieve a significant improvement on previous results with the help of the circuits.
Leader election has turned out as a crucial primitive to break symmetries in various geometric problems, e.g., shape formation (e.g., see \cite{DBLP:conf/spaa/DerakhshandehGR16,DBLP:journals/dc/LunaFSVY20}) and shape recognition.
Compass alignment and chirality agreement are important prerequisites for the coordinated movements of amoebots.
Therefore, this paper will lay the foundation for \emph{rapid} forms of shape transformation and object coating.

\subsection{Model}
\label{subsec:model}

We introduce an extension of the \emph{geometric amoebot model} from~\cite{DBLP:conf/spaa/DerakhshandehDGRSS14}, which we describe in the following.
In this model, a set of $n$ uniform amoebots is placed on the infinite regular triangular grid graph $G_{eqt} = (V, E)$ (see Figure~\ref{fig:G_eqt}).
The amoebots are anonymous, randomized finite state machines.
Since we are not dealing with movements in this paper, we assume that each amoebot occupies a single node\footnote{Actually, amoebots are allowed to occupy either a single node or a pair of adjacent nodes. W.l.o.g., we can assume the former for all amoebots. Otherwise, we let the latter simulate 2 individual amoebots.} and each node is occupied by at most one amoebot.
Alternatively to the triangular grid, we can consider a hexagonal tiling such that the centers of the hexagons coincide with the nodes of the grid.
We utilize this perspective to show the interior of the amoebots.

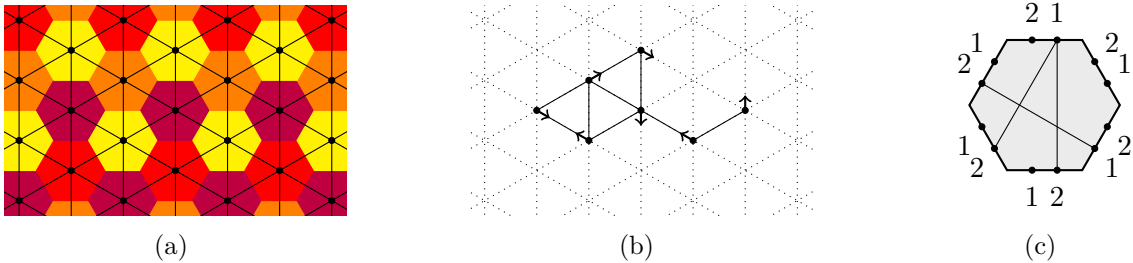
\begin{figure}[htp]
    \centering
    \begin{subfigure}[b]{0.35\textwidth}
        \centering
        \begin{tikzpicture}
            \clip (0.493,1) rectangle (5.050,3.8);
            \tikzmath{
                \scale = .8;
                \t = \scale/2;
                \h = sqrt(\scale^2 - (\scale/2)^2);
                \c = 8;
                % tiles
                for \j in {0,2,...,\c}{
                    for \i in {0,8}{
                        \x = \j*\h;
                        \y = \i*\t;
                        {
                            \filldraw[yellow]
                            (\x-1/3*\h, \y+\t) --
                            (\x+1/3*\h, \y+\t) --
                            (\x+2/3*\h, \y) --
                            (\x+1/3*\h, \y-\t) --
                            (\x-1/3*\h, \y-\t) --
                            (\x-2/3*\h, \y) --
                            cycle;
                        };
                        \x = \j*\h;
                        \y = \i*\t+2*\t;
                        {
                            \filldraw[orange]
                            (\x-1/3*\h, \y+\t) --
                            (\x+1/3*\h, \y+\t) --
                            (\x+2/3*\h, \y) --
                            (\x+1/3*\h, \y-\t) --
                            (\x-1/3*\h, \y-\t) --
                            (\x-2/3*\h, \y) --
                            cycle;
                        };
                        \x = \j*\h;
                        \y = \i*\t+4*\t;
                        {
                            \filldraw[red]
                            (\x-1/3*\h, \y+\t) --
                            (\x+1/3*\h, \y+\t) --
                            (\x+2/3*\h, \y) --
                            (\x+1/3*\h, \y-\t) --
                            (\x-1/3*\h, \y-\t) --
                            (\x-2/3*\h, \y) --
                            cycle;
                        };
                        \x = \j*\h;
                        \y = \i*\t+6*\t;
                        {
                            \filldraw[purple]
                            (\x-1/3*\h, \y+\t) --
                            (\x+1/3*\h, \y+\t) --
                            (\x+2/3*\h, \y) --
                            (\x+1/3*\h, \y-\t) --
                            (\x-1/3*\h, \y-\t) --
                            (\x-2/3*\h, \y) --
                            cycle;
                        };
                        \x = \j*\h+\h;
                        \y = \i*\t+\t;
                        {
                            \filldraw[red]
                            (\x-1/3*\h, \y+\t) --
                            (\x+1/3*\h, \y+\t) --
                            (\x+2/3*\h, \y) --
                            (\x+1/3*\h, \y-\t) --
                            (\x-1/3*\h, \y-\t) --
                            (\x-2/3*\h, \y) --
                            cycle;
                        };
                        \x = \j*\h+\h;
                        \y = \i*\t+3*\t;
                        {
                            \filldraw[purple]
                            (\x-1/3*\h, \y+\t) --
                            (\x+1/3*\h, \y+\t) --
                            (\x+2/3*\h, \y) --
                            (\x+1/3*\h, \y-\t) --
                            (\x-1/3*\h, \y-\t) --
                            (\x-2/3*\h, \y) --
                            cycle;
                        };
                        \x = \j*\h+\h;
                        \y = \i*\t+5*\t;
                        {
                            \filldraw[yellow]
                            (\x-1/3*\h, \y+\t) --
                            (\x+1/3*\h, \y+\t) --
                            (\x+2/3*\h, \y) --
                            (\x+1/3*\h, \y-\t) --
                            (\x-1/3*\h, \y-\t) --
                            (\x-2/3*\h, \y) --
                            cycle;
                        };
                        \x = \j*\h+\h;
                        \y = \i*\t+7*\t;
                        {
                            \filldraw[orange]
                            (\x-1/3*\h, \y+\t) --
                            (\x+1/3*\h, \y+\t) --
                            (\x+2/3*\h, \y) --
                            (\x+1/3*\h, \y-\t) --
                            (\x-1/3*\h, \y-\t) --
                            (\x-2/3*\h, \y) --
                            cycle;
                        };
                    };
                };
                % nodes
                for \j in {0,2,...,\c+1}{
                    for \i in {0,2,...,15}{
                        {
                            \filldraw (\j*\h, \i*\t) circle (0.04);
                            \filldraw (\j*\h+\h, \i*\t+\t) circle (0.04);
                        };
                    };
                };
                % edges
                for \i in {-16,-14,...,32}{
                    {
                        \draw (0*\h, \i*\t) -- (\c*\h+\h, \i*\t+\c*\t+\t);
                        \draw (0*\h, \i*\t) -- (\c*\h+\h, \i*\t-\c*\t-\t);
                    };
                };
                for \i in {0,1,...,\c+1}{
                    {
                        \draw (\i*\h, 0*\t) -- (\i*\h, 16*\t);
                    };
                };
            }
        \end{tikzpicture}
        \caption{\centering}
        \label{fig:G_eqt}
    \end{subfigure}
    \hfill
    \begin{subfigure}[b]{0.35\textwidth}
        \centering
        \begin{tikzpicture}
            \clip (0.493,1) rectangle (5.050,3.8);
            \tikzmath{
                \scale = .8;
                \t = \scale/2;
                \h = sqrt(\scale^2 - (\scale/2)^2);
                \c = 8;
                % grid
                for \i in {-16,-14,...,32}{
                    {
                        \draw[dotted] (0*\h, \i*\t) -- (\c*\h+\h, \i*\t+\c*\t+\t);
                        \draw[dotted] (0*\h, \i*\t) -- (\c*\h+\h, \i*\t-\c*\t-\t);
                    };
                };
                for \i in {0,1,...,\c+1}{
                    {
                        \draw[dotted] (\i*\h, 0*\t) -- (\i*\h, 16*\t);
                    };
                };
                % structure
                {
                    \filldraw (2*\h, 6*\t) circle (0.04);
                    \filldraw (3*\h, 5*\t) circle (0.04);
                    \filldraw (3*\h, 7*\t) circle (0.04);
                    \filldraw (4*\h, 6*\t) circle (0.04);
                    \filldraw (4*\h, 8*\t) circle (0.04);
                    \filldraw (5*\h, 5*\t) circle (0.04);
                    \filldraw (6*\h, 6*\t) circle (0.04);
                    \draw (2*\h, 6*\t) -- (4*\h, 8*\t);
                    \draw (3*\h, 5*\t) -- (4*\h, 6*\t);
                    \draw (5*\h, 5*\t) -- (6*\h, 6*\t);
                    \draw (3*\h, 5*\t) -- (3*\h, 7*\t);
                    \draw (4*\h, 6*\t) -- (4*\h, 8*\t);
                    \draw (2*\h, 6*\t) -- (3*\h, 5*\t);
                    \draw (3*\h, 7*\t) -- (5*\h, 5*\t);
                    \draw[thick,->] (2*\h, 6*\t) -- +(+\h/4,-\t/4);
                    \draw[thick,->] (3*\h, 5*\t) -- +(-\h/4,+\t/4);
                    \draw[thick,->] (3*\h, 7*\t) -- +(+\h/4,+\t/4);
                    \draw[thick,->] (4*\h, 6*\t) -- +(0,-\t/2);
                    \draw[thick,->] (4*\h, 8*\t) -- +(+\h/4,-\t/4);
                    \draw[thick,->] (5*\h, 5*\t) -- +(-\h/4,+\t/4);
                    \draw[thick,->] (6*\h, 6*\t) -- +(0,+\t/2);
                };
            }
        \end{tikzpicture}
        \caption{\centering}
        \label{fig:structure}
    \end{subfigure}
    \hfill
    \begin{subfigure}[b]{0.25\textwidth}
        \centering
        \begin{tikzpicture}
            \tikzmath{
                \scale = 1;
                \pins = 2;
                \x = 0;
                \y = 0;
                \z = 0;
                \fr = 235;
                \fg = 235;
                \fb = 235;
                % auxiliary
                \t = \scale/2;
                \h = sqrt(\scale^2 - (\scale/2)^2);
                % translation
                \tx = (\x+\z) * 3*\t;
                \ty = (\y+0.5*\x-0.5*\z) * 2*\h;
                % amoebot
                {
                    \filldraw[color=black, fill={rgb,255:red,\fr; green,\fg; blue,\fb}, thick]
                    (\t + \tx,0 + \ty)
                    --
                    (3*\t + \tx,0 + \ty)
                    --
                    (4*\t + \tx,\h + \ty)
                    --
                    (3*\t + \tx,2*\h + \ty)
                    --
                    (\t + \tx,2*\h + \ty)
                    --
                    (0 + \tx,\h + \ty)
                    --
                    cycle;
                };
                % pins
                if \pins > 0 then {
                    for \p in {1,...,\pins}{
                        integer \pinnr;
                        \pinnr = \p;
                        \q = \p / (\pins+1);
                        % NW/SW
                        \xp = (1-\q) * 0 + \q * \t + \tx;
                        \ypn = (1-\q) * \h + \q * 2*\h + \ty;
                        \yps = (1-\q) * \h + \q * 0 + \ty;
                        {
                            \filldraw (\xp, \yps) circle (0.04); % SW
                            \node[below left] at (\xp, \yps) {\pinnr}; % SW
                        };
                        % N/S
                        \xp = (1-\q) * \t + \q * 3*\t + \tx;
                        \ypn = 2*\h + \ty;
                        \yps = 0 + \ty;
                        {
                            \filldraw (\xp, \yps) circle (0.04); % S
                            \node[below] at (\xp, \yps - .1) {\pinnr}; % S
                        };
                        % NE/SE
                        \xp = (1-\q) * 3*\t + \q * 4*\t + \tx;
                        \ypn = (1-\q) * 2*\h + \q * \h + \ty;
                        \yps = (1-\q) * 0 + \q * \h + \ty;
                        {
                            \filldraw (\xp, \yps) circle (0.04); % SE
                            \node[below right] at (\xp, \yps) {\pinnr}; % SE
                        };
                    };
                };
                if \pins > 0 then {
                    for \p in {1,...,\pins}{
                        integer \pinnr;
                        \pinnr = \p;
                        \q = (\pins - \p + 1) / (\pins+1);
                        % NW/SW
                        \xp = (1-\q) * 0 + \q * \t + \tx;
                        \ypn = (1-\q) * \h + \q * 2*\h + \ty;
                        \yps = (1-\q) * \h + \q * 0 + \ty;
                        {
                            \filldraw (\xp, \ypn) circle (0.04); % NW
                            \node[above left] at (\xp, \ypn) {\pinnr}; % NW
                        };
                        % N/S
                        \xp = (1-\q) * \t + \q * 3*\t + \tx;
                        \ypn = 2*\h + \ty;
                        \yps = 0 + \ty;
                        {
                            \filldraw (\xp, \ypn) circle (0.04); % N
                            \node[above] at (\xp, \ypn + .1) {\pinnr}; % N
                        };
                        % NE/SE
                        \xp = (1-\q) * 3*\t + \q * 4*\t + \tx;
                        \ypn = (1-\q) * 2*\h + \q * \h + \ty;
                        \yps = (1-\q) * 0 + \q * \h + \ty;
                        {
                            \filldraw (\xp, \ypn) circle (0.04); % NE
                            \node[above right] at (\xp, \ypn) {\pinnr}; % NE
                        };
                    };
                };
                \q = 2 / (\pins+1);
                \pax = (1-\q) * 0 + \q * \t + \tx;
                \pay = (1-\q) * \h + \q * 0 + \ty;
                \q = 2 / (\pins+1);
                \pbx = (1-\q) * \t + \q * 3*\t + \tx;
                \pby = 2*\h + \ty;
                {
                    \draw (\pax,\pay) -- (\pbx,\pby);
                };
                \q = 2 / (\pins+1);
                \pax = (1-\q) * \t + \q * 3*\t + \tx;
                \pay = 2*\h + \ty;
                \q = 2 / (\pins+1);
                \pbx = (1-\q) * \t + \q * 3*\t + \tx;
                \pby = 0 + \ty;
                {
                    \draw (\pax,\pay) -- (\pbx,\pby);
                };
                \q = 1 / (\pins+1);
                \pax = (1-\q) * 3*\t + \q * 4*\t + \tx;
                \pay = (1-\q) * 0 + \q * \h + \ty;
                \q = 1 / (\pins+1);
                \pbx = (1-\q) * 0 + \q * \t + \tx;
                \pby = (1-\q) * \h + \q * 2*\h + \ty;
                {
                    \draw (\pax,\pay) -- (\pbx,\pby);
                };
            }
        \end{tikzpicture}
        \caption{\centering}
        \label{fig:pins}
    \end{subfigure}
    \caption{
        (a) shows a section of $G_{eqt}$ and the corresponding hexagonal tiling.
        (b) shows a connected amoebot structure.
        The nodes are amoebots.
        The solid lines are bonds.
        The arrows indicate the compass orientation.
        We omit the underlying triangular grid in the remaining figures.
        (c) shows the inner perspective of a single amoebot.
        The nodes are pins.
        The pin configuration is given by the edges between the pins.
        Here, the pins are labeled in a counterclockwise fashion.
    }
\end{figure}

A bond is formed between amoebots occupying adjacent nodes.
These amoebots are said to be neighbors.
Neighbors are able to exchange (constant sized) messages\footnote{The communication in previous work is mainly done by shared memories. However, this may lead to various write/read conflicts, which have to be resolved. Our model extension provides a nice solution for this problem, which is discussed in Section~\ref{sec:preliminaries}.}.
We discuss the implementation of these messages in Section~\ref{sec:preliminaries}.
Furthermore, an amoebot assigns a locally unique label to each edge incident to its occupied nodes.
Amoebots do neither agree on a common chirality nor on a common compass orientation.

Let the \emph{amoebot structure}~$S \subseteq V$ be the set of nodes occupied by the amoebots (see Figure~\ref{fig:structure}).
It is connected iff~$G_S$ is connected where $G_S = G_{eqt}|_S$ is the graph induced by~$S$.
We expect that the amoebot structure is connected\footnote{Amoebots have the ability to move by expansions and contractions. Since the presented algorithms work entirely stationary, we omit further details and refer to~\cite{DBLP:conf/spaa/DerakhshandehDGRSS14}. The amoebots have to maintain a connected structure at all times.}.
Furthermore, the amoebot structure is assumed to progress in synchronized rounds.
Each amoebot executes a \emph{look-compute-act} cycle in each round.
We explain the cycles at the end of this section.
The complexity of an algorithm is measured by the number of synchronized rounds required.

In our model extension, \emph{pins} are added to all incident edges of each amoebot (see Figure~\ref{fig:pins}).
The number of pins is equal for each edge and bounded by a constant.
The pins of neighboring amoebots coincide and are therefore viewed as a single set of pins.
The pins of each edge have a fixed order, i.e., each pin has a predecessor and a successor (except for the first and last pin of each edge).
This order is known to both incident amoebots.
However, each amoebot labels the pins according to its chirality.
As a consequence, neighboring amoebots do not agree on the labeling if they have the same chirality.
For this paper, we assume that each edge contains at least two pins (unless stated otherwise) because this allows us to solve the chirality agreement problem.

Amoebots can connect the pins of their incident edges by wires represented as undirected edges.
We call the set of these edges a \emph{pin configuration}.
We do not restrict the set in any way, e.g., by requiring a matching.
Ultimately, we are only interested in the connectivity of pins.
Thus, it does not matter whether a connected component formed by the pins is cycle-free or not.
More specifically, let $P$ be the set of all pins in the system and let $C$ be the set of all edges between pins.
Then, we define a \emph{circuit} as a connected component of graph $G_C = (P, C)$ (see Figure~\ref{fig:circuits}).
An amoebot is part of a circuit iff the circuit contains at least one pin on one of its bonds.

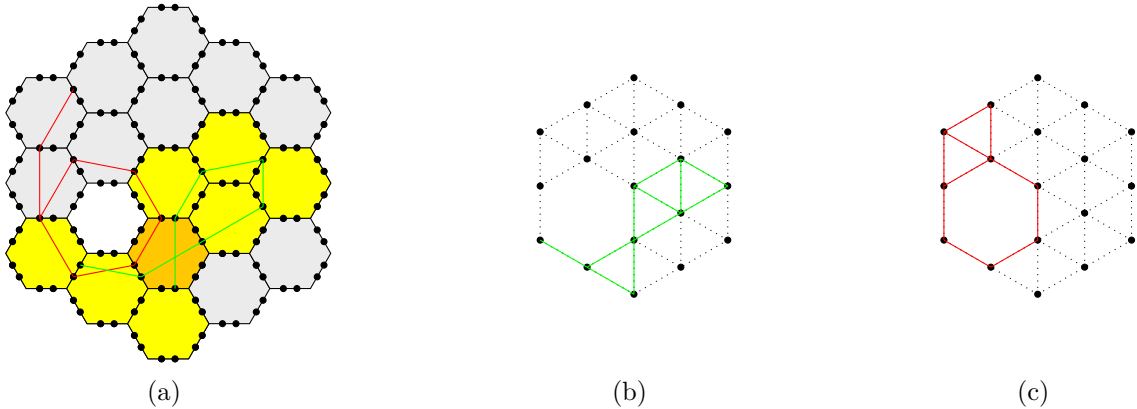
\begin{figure}[htp]
    \centering
    \begin{subfigure}[b]{0.35\textwidth}
        \centering
        \begin{tikzpicture}
            \tikzmath{
                \scale = .6*.9;
                \pins = 2;
                \fr = 235;
                \fg = 235;
                \fb = 235;
                for \i in {2,3,4}{
                    Amoebot(0,\i,0);
                };
                for \i in {1,3,4}{
                    Amoebot(1,\i,0);
                };
                for \i in {0,1,2,3,4}{
                    Amoebot(2,\i,0);
                };
                for \i in {0,1,2,3}{
                    Amoebot(3,\i,0);
                };
                for \i in {0,1,2}{
                    Amoebot(4,\i,0);
                };
                \fr = 255;
                \fg = 255;
                \fb = 0;
                Amoebot(0,2,0);
                Amoebot(1,1,0);
                Amoebot(2,0,0);
                Amoebot(2,2,0);
                Amoebot(3,1,0);
                Amoebot(3,2,0);
                Amoebot(4,1,0);
                \fr = 255;
                \fg = 200;
                \fb = 0;
                Amoebot(2,1,0);
                \lr = 255;
                \lg = 0;
                \lb = 0;
                Circuit(0,2,0,2,1,4,1);
                Circuit(0,3,0,1,1,4,1);
                Circuit(0,3,0,1,1,3,1);
                Circuit(0,4,0,1,1,3,1);
                Circuit(1,1,0,5,1,3,1);
                Circuit(1,3,0,6,1,2,1);
                Circuit(2,1,0,6,1,4,1);
                Circuit(2,2,0,1,1,5,1);
                \lr = 0;
                \lg = 255;
                \lb = 0;
                Circuit(1,1,0,5,2,3,2);
                Circuit(2,1,0,6,2,3,2);
                Circuit(3,1,0,6,2,3,2);
                Circuit(4,1,0,6,2,5,2);
                Circuit(3,2,0,6,2,2,2);
                Circuit(2,2,0,3,2,1,2);
                Circuit(2,1,0,4,2,1,2);
            }
        \end{tikzpicture}
        \caption{\centering}
    \end{subfigure}
    \hfill
    \begin{subfigure}[b]{0.25\textwidth}
        \centering
        \begin{tikzpicture}
            \tikzmath{
                \scale = .8*.9;
                \t = \scale/2;
                \h = sqrt(\scale^2 - (\scale/2)^2);
                % bounding box for better alignment
                {
                    \draw[color=white] (-\t, -2.5*\t) rectangle (4*\h+\t, 9*\t);
                };
                % particles
                for \i in {2,4,6}{
                    {
                        \filldraw (0*\h, \i*\t) circle (0.04);
                    };
                };
                for \i in {1,5,7}{
                    {
                        \filldraw (1*\h, \i*\t) circle (0.04);
                    };
                };
                for \i in {0,2,4,6,8}{
                    {
                        \filldraw (2*\h, \i*\t) circle (0.04);
                    };
                };
                for \i in {1,3,5,7}{
                    {
                        \filldraw (3*\h, \i*\t) circle (0.04);
                    };
                };
                for \i in {2,4,6}{
                    {
                        \filldraw (4*\h, \i*\t) circle (0.04);
                    };
                };
                {
                % bonds
                    % N - S
                    \draw[dotted] (0*\h, 2*\t) -- (0*\h, 6*\t);
                    \draw[dotted] (1*\h, 5*\t) -- (1*\h, 7*\t);
                    \draw[dotted] (2*\h, 0*\t) -- (2*\h, 8*\t);
                    \draw[dotted] (3*\h, 1*\t) -- (3*\h, 7*\t);
                    \draw[dotted] (4*\h, 2*\t) -- (4*\h, 6*\t);
                    % SW - NE
                    \draw[dotted] (0*\h, 6*\t) -- (2*\h, 8*\t);
                    \draw[dotted] (0*\h, 4*\t) -- (3*\h, 7*\t);
                    \draw[dotted] (2*\h, 4*\t) -- (4*\h, 6*\t);
                    \draw[dotted] (1*\h, 1*\t) -- (4*\h, 4*\t);
                    \draw[dotted] (2*\h, 0*\t) -- (4*\h, 2*\t);
                    % NW - SE
                    \draw[dotted] (2*\h, 8*\t) -- (4*\h, 6*\t);
                    \draw[dotted] (1*\h, 7*\t) -- (4*\h, 4*\t);
                    \draw[dotted] (0*\h, 6*\t) -- (4*\h, 2*\t);
                    \draw[dotted] (2*\h, 2*\t) -- (3*\h, 1*\t);
                    \draw[dotted] (0*\h, 2*\t) -- (2*\h, 0*\t);
                    % N - S
                    \draw[color=green] (2*\h, 0*\t) -- (2*\h, 4*\t);
                    \draw[color=green] (3*\h, 3*\t) -- (3*\h, 5*\t);
                    % SW - NE
                    \draw[color=green] (2*\h, 4*\t) -- (3*\h, 5*\t);
                    \draw[color=green] (1*\h, 1*\t) -- (4*\h, 4*\t);
                    % NW - SE
                    \draw[color=green] (3*\h, 5*\t) -- (4*\h, 4*\t);
                    \draw[color=green] (2*\h, 4*\t) -- (3*\h, 3*\t);
                    \draw[color=green] (0*\h, 2*\t) -- (2*\h, 0*\t);
                };
            }
        \end{tikzpicture}
        \caption{\centering}
    \end{subfigure}
    \hfill
    \begin{subfigure}[b]{0.25\textwidth}
        \centering
        \begin{tikzpicture}
            \tikzmath{
                \scale = .8*.9;
                \t = \scale/2;
                \h = sqrt(\scale^2 - (\scale/2)^2);
                % bounding box for better alignment
                {
                    \draw[color=white] (-\t, -2.5*\t) rectangle (4*\h+\t, 9*\t);
                };
                % particles
                for \i in {2,4,6}{
                    {
                        \filldraw (0*\h, \i*\t) circle (0.04);
                    };
                };
                for \i in {1,5,7}{
                    {
                        \filldraw (1*\h, \i*\t) circle (0.04);
                    };
                };
                for \i in {0,2,4,6,8}{
                    {
                        \filldraw (2*\h, \i*\t) circle (0.04);
                    };
                };
                for \i in {1,3,5,7}{
                    {
                        \filldraw (3*\h, \i*\t) circle (0.04);
                    };
                };
                for \i in {2,4,6}{
                    {
                        \filldraw (4*\h, \i*\t) circle (0.04);
                    };
                };
                {
                % bonds
                    % N - S
                    \draw[dotted] (0*\h, 2*\t) -- (0*\h, 6*\t);
                    \draw[dotted] (1*\h, 5*\t) -- (1*\h, 7*\t);
                    \draw[dotted] (2*\h, 0*\t) -- (2*\h, 8*\t);
                    \draw[dotted] (3*\h, 1*\t) -- (3*\h, 7*\t);
                    \draw[dotted] (4*\h, 2*\t) -- (4*\h, 6*\t);
                    % SW - NE
                    \draw[dotted] (0*\h, 6*\t) -- (2*\h, 8*\t);
                    \draw[dotted] (0*\h, 4*\t) -- (3*\h, 7*\t);
                    \draw[dotted] (2*\h, 4*\t) -- (4*\h, 6*\t);
                    \draw[dotted] (1*\h, 1*\t) -- (4*\h, 4*\t);
                    \draw[dotted] (2*\h, 0*\t) -- (4*\h, 2*\t);
                    % NW - SE
                    \draw[dotted] (2*\h, 8*\t) -- (4*\h, 6*\t);
                    \draw[dotted] (1*\h, 7*\t) -- (4*\h, 4*\t);
                    \draw[dotted] (0*\h, 6*\t) -- (4*\h, 2*\t);
                    \draw[dotted] (2*\h, 2*\t) -- (3*\h, 1*\t);
                    \draw[dotted] (0*\h, 2*\t) -- (2*\h, 0*\t);
                    % N - S
                    \draw[color = red] (0*\h, 2*\t) -- (0*\h, 6*\t);
                    \draw[color = red] (1*\h, 5*\t) -- (1*\h, 7*\t);
                    \draw[color = red] (2*\h, 2*\t) -- (2*\h, 4*\t);
                    % SW - NE
                    \draw[color = red] (0*\h, 6*\t) -- (1*\h, 7*\t);
                    \draw[color = red] (0*\h, 4*\t) -- (1*\h, 5*\t);
                    \draw[color = red] (1*\h, 1*\t) -- (2*\h, 2*\t);
                    % NW - SE
                    \draw[color = red] (0*\h, 6*\t) -- (2*\h, 4*\t);
                    \draw[color = red] (0*\h, 2*\t) -- (1*\h, 1*\t);
                };
            }
        \end{tikzpicture}
        \caption{\centering}
    \end{subfigure}
    \caption{
        (a) shows an amoebot structure with two pins per edge.
        The green and red edges are two disconnected circuits.
        Note that each unconnected pin resembles another circuit.
        The yellow and orange amoebots are connected to the green circuit.
        Note that the orange particle does not know that its connections belong to the same circuit.
        (b) and (c) show a simplified view on the green and red circuit, respectively.
    }
    \label{fig:circuits}
\end{figure}

Each amoebot can send a primitive signal (a \emph{beep}) through its circuits that is received by all amoebots of the same circuit in the next round.
%This includes the sending amoebot.
The amoebots receive a beep if at least one amoebot sends a beep on the circuit
but the amoebots neither know the origin of the signal nor the number of origins.
We have chosen a primitive signal instead of more complex messages to keep our extension as simple as possible.
However, these are enough to send whole messages if there is only one sender (see Section~\ref{sec:preliminaries}).

It remains to explain the \emph{look-compute-act} cycles.
In the look phase the amoebot gathers information.
It reads its local memory and it may receive beeps on its circuits.
In the compute phase it performs some calculations and may update its local memory.
In the act phase it may reconfigure the connections of its pins and send a signal through an arbitrary number of its circuits.
The signals are received during the next look phase.

\subsection{Why care about Circuits?}

In order to provide a concrete motivation why our proposed extension of the amoebot model via reconfigurable circuits has some benefits over the standard model, we first present a fast algorithm for the leader election problem (which is an extensively researched problem within the geometric amoebot model~\cite{DBLP:conf/sss/BazziB19, DBLP:conf/algosensors/DaymudeGRSS17, DBLP:conf/dna/DerakhshandehGS15, DBLP:conf/icalp/EmekKLM19, DBLP:conf/algosensors/GastineauAMT18, DBLP:journals/dc/LunaFSVY20}) and then describe how using circuits can be helpful in shape transformation, which may potentially lead to faster algorithms for the (extensively researched) shape formation problem~\cite{DBLP:conf/icdcn/DaymudeGHKSR20, DBLP:conf/nanocom/DerakhshandehGR15, DBLP:conf/spaa/DerakhshandehGR16, DBLP:conf/dna/DerakhshandehGS15, DBLP:journals/dc/LunaFSVY20}.

\subsubsection{Leader Election}
\label{sec:leader}

In this section we give an efficient solution for the leader election problem, where the amoebots have to agree on exactly one single amoebot, which becomes the leader.
The algorithm has already been proposed in a very similar manner for a clique of $n$ nodes in the beeping model (see the related work in Section~\ref{subsec:related_work}) by Gilbert and Newport~\cite{DBLP:conf/wdag/GilbertN15} and elects a leader in $O(\log^2 n)$ rounds w.h.p.\footnote{An event holds with high probability (w.h.p.) if it holds with probability at least $1 - 1/n^c$ where the constant $c$ can be made arbitrarily large.}
However, by carefully adjusting their termination subroutine, making it run in parallel to the election routine and executing it only a constant amount of times, we are able to improve their algorithm such that the leader is now elected after only $\Theta(\log n)$ rounds w.h.p.

Before starting the actual algorithm, the amoebots connect their all of their local pins to a clique (the algorithm only requires one pin per bond, but also works for multiple pins).
By doing so we establish a single circuit that contains all amoebots.
We call this circuit the \emph{global circuit}.
The algorithm only makes use of the global circuit and works in two phases (see Algorithm~\ref{alg:leader_election} in Appendix~\ref{app:leader} for the pseudocode and a detailed analysis).
Let $C_1 \subseteq S$ be the set of potential \emph{candidate} to be the leader.
Initially, $C_1 = S$.
In the first phase, we reduce the number of candidates from $n$ to $O(\log n)$.
The second phase then elects a leader among the remaining $O(\log n)$ candidates, w.h.p.

In the first phase we perform a \emph{tournament} on the candidates.
A single \emph{iteration} of the tournament works as follows.
Each candidate $u \in C_1$ tosses a coin that is either $\mathit{HEADS}$ or $\mathit{TAILS}$ and stores the result in a variable $u.c_1$.
Now consider two subsequent rounds $r_1, r_2$.
In round $r_1$ all candidates $u \in C_1$ with $u.c_1 = \mathit{HEADS}$ send a beep through the global circuit and in $r_2$ all candidates $u \in C_1$ with $u.c_1 = \mathit{TAILS}$ send a signal through the global circuit.
As signals are received in the next round by all particles, each particle is able to check if there is a candidate that beeped in $r_1$ and one that beeped in $r_2$.
If at least one candidate beeped in $r_1$, all candidates that beeped in $r_2$ are out of contention for being the leader (therefore, the set $C_1$ gets updated).
We continue performing iterations of the tournament until we reach an iteration where in either round $r_1$ or in $r_2$ no candidate beeped.
This finishes the first phase of our protocol and takes $\Theta(\log n)$ rounds, w.h.p.

In the second phase we continue the tournament from the first phase on the remaining $O(\log n)$ candidates in $C_1$ until there is only one single candidate left.
Unfortunately, we cannot check efficiently when this is the case because we cannot count the number of particles that have sent a signal in a single round through the global circuit.
Therefore, we aim to continue the tournament from the first phase for another $\Theta(\log n)$ rounds.
As particles cannot count to $\Theta(\log n)$ due to their constant-sized storage, we just perform a second tournament on the set $C_2$, initially set to $S$, in parallel to the tournament on the set $C_1$.
One iteration of the second phase therefore consists of $4$ rounds ($2$ rounds for the tournament on $C_1$ and $2$ rounds for the tournament on $C_2$).
The second tournament terminates once the tournament on $C_2$ has finished.
By repeating the second tournament $\kappa$ times for a constant $\kappa \geq 3$, we can guarantee that w.h.p., only one single node remains as a candidate.

\begin{theorem}
\label{th:leader_election}
	There exists a protocol using the global circuit that lets the amoebots elect a leader within $\Theta(\log n)$ rounds, w.h.p.
\end{theorem}

\subsubsection{Rapid Shape Transformation}

We strongly believe that through the usage of circuits along with additional extensions based on the muscular system, one is able to come up with faster algorithms for the shape transformation problem.
In the shape transformation problem, the amoebots have the task to arrange themselves (via movements) in a specific (2-dimensional) shape, like a parallelogram, a triangle, or any shape imaginable.
For the standard geometric amoebot model, some lower bounds are known for the shape formation problem.
Di Luna et al.~\cite{DBLP:journals/dc/LunaFSVY20} show that any universal shape formation algorithm needs $\Omega(n^2)$ moves and $\Omega(n)$ rounds.
Given that the desired shape is only a collection of triangles, the lower bound on the number of rounds reduces to $\Omega(\sqrt{n})$, as it has been shown by Derakhshandeh et al.~\cite{DBLP:conf/spaa/DerakhshandehGR16}.

Both lower bounds can be broken with the help of circuits.
To give an (informal) description of a concrete example, consider the shape formation process depicted in Figure~\ref{fig:shape_formation}, where $n$ amoebots initially form a simple line and aim to transform into a parallelogram of side length $\sqrt n$.

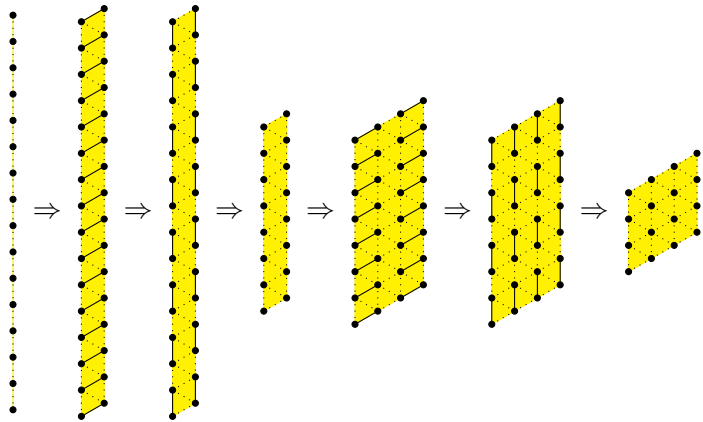
\begin{figure}[htp]
    \centering
    \begin{tikzpicture}
        \tikzmath{
            \scale = .35;
            \t = \scale/2;
            \h = sqrt(\scale^2 - (\scale/2)^2);
            % 1
            \xo = 0*\h;
            \yo = 0.5*\t;
            {
                \filldraw[yellow]
                (0*\h+\xo, 0*\t+\yo) --
                (0*\h+\xo, 30*\t+\yo) --
                (0*\h+\xo, 30*\t+\yo) --
                (0*\h+\xo, 0*\t+\yo) --
                cycle;
            };
            for \j in {0}{
                for \i in {0,4,...,28}{
                    {
                        \filldraw (0*\h+\j*\h+\xo, \i*\t+\j*\t+0*\t+\yo) circle (0.04);
                        \filldraw (0*\h+\j*\h+\xo, \i*\t+\j*\t+2*\t+\yo) circle (0.04);
                        \draw[dotted] (0*\h+\j*\h+\xo, \i*\t+\j*\t+0*\t+\yo) -- (0*\h+\j*\h+\xo, \i*\t+\j*\t+2*\t+\yo);
                    };
                };
                for \i in {2,6,...,26}{
                    {
                        \draw[dotted] (0*\h+\j*\h+\xo, \i*\t+\j*\t+0*\t+\yo) -- (0*\h+\j*\h+\xo, \i*\t+\j*\t+2*\t+\yo);
                    };
                };
            };
            % 2
            \xo = 3*\h;
            \yo = 0*\t;
            {
                \filldraw[yellow]
                (0*\h+\xo, 0*\t+\yo) --
                (0*\h+\xo, 30*\t+\yo) --
                (1*\h+\xo, 31*\t+\yo) --
                (1*\h+\xo, 1*\t+\yo) --
                cycle;
            };
            for \j in {0}{
                for \i in {0,4,...,28}{
                    {
                        \filldraw (0*\h+\j*\h+\xo, \i*\t+\j*\t+0*\t+\yo) circle (0.04);
                        \filldraw (0*\h+\j*\h+\xo, \i*\t+\j*\t+2*\t+\yo) circle (0.04);
                        \filldraw (1*\h+\j*\h+\xo, \i*\t+\j*\t+1*\t+\yo) circle (0.04);
                        \filldraw (1*\h+\j*\h+\xo, \i*\t+\j*\t+3*\t+\yo) circle (0.04);
                        \draw (0*\h+\j*\h+\xo, \i*\t+\j*\t+0*\t+\yo) -- (1*\h+\j*\h+\xo, \i*\t+\j*\t+1*\t+\yo);
                        \draw (0*\h+\j*\h+\xo, \i*\t+\j*\t+2*\t+\yo) -- (1*\h+\j*\h+\xo, \i*\t+\j*\t+3*\t+\yo);
                        \draw[dotted] (0*\h+\j*\h+\xo, \i*\t+\j*\t+0*\t+\yo) -- (0*\h+\j*\h+\xo, \i*\t+\j*\t+2*\t+\yo);
                        \draw[dotted] (1*\h+\j*\h+\xo, \i*\t+\j*\t+1*\t+\yo) -- (1*\h+\j*\h+\xo, \i*\t+\j*\t+3*\t+\yo);
                        \draw[dotted] (0*\h+\j*\h+\xo, \i*\t+\j*\t+2*\t+\yo) -- (1*\h+\j*\h+\xo, \i*\t+\j*\t+1*\t+\yo);
                    };
                };
                for \i in {2,6,...,26}{
                    {
                        \draw[dotted] (0*\h+\j*\h+\xo, \i*\t+\j*\t+0*\t+\yo) -- (0*\h+\j*\h+\xo, \i*\t+\j*\t+2*\t+\yo);
                        \draw[dotted] (1*\h+\j*\h+\xo, \i*\t+\j*\t+1*\t+\yo) -- (1*\h+\j*\h+\xo, \i*\t+\j*\t+3*\t+\yo);
                        \draw[dotted] (0*\h+\j*\h+\xo, \i*\t+\j*\t+2*\t+\yo) -- (1*\h+\j*\h+\xo, \i*\t+\j*\t+1*\t+\yo);
                    };
                };
            };
            % 3
            \xo = 7*\h;
            \yo = 0*\t;
            {
                \filldraw[yellow]
                (0*\h+\xo, 0*\t+\yo) --
                (0*\h+\xo, 30*\t+\yo) --
                (1*\h+\xo, 31*\t+\yo) --
                (1*\h+\xo, 1*\t+\yo) --
                cycle;
            };
            for \j in {0}{
                for \i in {0,4,...,28}{
                    {
                        \filldraw (0*\h+\j*\h+\xo, \i*\t+\j*\t+0*\t+\yo) circle (0.04);
                        \filldraw (0*\h+\j*\h+\xo, \i*\t+\j*\t+2*\t+\yo) circle (0.04);
                        \filldraw (1*\h+\j*\h+\xo, \i*\t+\j*\t+1*\t+\yo) circle (0.04);
                        \filldraw (1*\h+\j*\h+\xo, \i*\t+\j*\t+3*\t+\yo) circle (0.04);
                        \draw[dotted] (0*\h+\j*\h+\xo, \i*\t+\j*\t+0*\t+\yo) -- (1*\h+\j*\h+\xo, \i*\t+\j*\t+1*\t+\yo);
                        \draw[dotted] (0*\h+\j*\h+\xo, \i*\t+\j*\t+2*\t+\yo) -- (1*\h+\j*\h+\xo, \i*\t+\j*\t+3*\t+\yo);
                        \draw (0*\h+\j*\h+\xo, \i*\t+\j*\t+0*\t+\yo) -- (0*\h+\j*\h+\xo, \i*\t+\j*\t+2*\t+\yo);
                        \draw (1*\h+\j*\h+\xo, \i*\t+\j*\t+1*\t+\yo) -- (1*\h+\j*\h+\xo, \i*\t+\j*\t+3*\t+\yo);
                        \draw[dotted] (0*\h+\j*\h+\xo, \i*\t+\j*\t+2*\t+\yo) -- (1*\h+\j*\h+\xo, \i*\t+\j*\t+1*\t+\yo);
                    };
                };
                for \i in {2,6,...,26}{
                    {
                        \draw[dotted] (0*\h+\j*\h+\xo, \i*\t+\j*\t+0*\t+\yo) -- (0*\h+\j*\h+\xo, \i*\t+\j*\t+2*\t+\yo);
                        \draw[dotted] (1*\h+\j*\h+\xo, \i*\t+\j*\t+1*\t+\yo) -- (1*\h+\j*\h+\xo, \i*\t+\j*\t+3*\t+\yo);
                        \draw[dotted] (0*\h+\j*\h+\xo, \i*\t+\j*\t+2*\t+\yo) -- (1*\h+\j*\h+\xo, \i*\t+\j*\t+1*\t+\yo);
                    };
                };
            };
            % 4
            \xo = 11*\h;
            \yo = 8*\t;
            {
                \filldraw[yellow]
                    (0*\h+\xo, 0*\t+\yo) --
                    (0*\h+\xo, 14*\t+\yo) --
                    (1*\h+\xo, 15*\t+\yo) --
                    (1*\h+\xo, 1*\t+\yo) --
                    cycle;
            };
            for \j in {0,1}{
                for \i in {0,4,...,14}{
                    {
                        \filldraw (0*\h+\j*\h+\xo, \i*\t+\j*\t+0*\t+\yo) circle (0.04);
                        \filldraw (0*\h+\j*\h+\xo, \i*\t+\j*\t+2*\t+\yo) circle (0.04);
                        \draw[dotted] (0*\h+\j*\h+\xo, \i*\t+\j*\t+0*\t+\yo) -- (0*\h+\j*\h+\xo, \i*\t+\j*\t+2*\t+\yo);
                    };
                };
                for \i in {2,6,...,12}{
                    {
                        \draw[dotted] (0*\h+\j*\h+\xo, \i*\t+\j*\t+0*\t+\yo) -- (0*\h+\j*\h+\xo, \i*\t+\j*\t+2*\t+\yo);
                    };
                };
            };
            for \j in {0}{
                for \i in {0,4,...,14}{
                    {
                        \draw[dotted] (0*\h+\j*\h+\xo, \i*\t+\j*\t+0*\t+\yo) -- (1*\h+\j*\h+\xo, \i*\t+\j*\t+1*\t+\yo);
                        \draw[dotted] (0*\h+\j*\h+\xo, \i*\t+\j*\t+2*\t+\yo) -- (1*\h+\j*\h+\xo, \i*\t+\j*\t+3*\t+\yo);
                        \draw[dotted] (0*\h+\j*\h+\xo, \i*\t+\j*\t+2*\t+\yo) -- (1*\h+\j*\h+\xo, \i*\t+\j*\t+1*\t+\yo);
                    };
                };
                for \i in {2,6,...,12}{
                    {
                        \draw[dotted] (0*\h+\j*\h+\xo, \i*\t+\j*\t+2*\t+\yo) -- (1*\h+\j*\h+\xo, \i*\t+\j*\t+1*\t+\yo);
                    };
                };
            };
            % 5
            \xo = 15*\h;
            \yo = 7*\t;
            {
                \filldraw[yellow]
                (0*\h+\xo, 0*\t+\yo) --
                (0*\h+\xo, 14*\t+\yo) --
                (3*\h+\xo, 17*\t+\yo) --
                (3*\h+\xo, 3*\t+\yo) --
                cycle;
            };
            for \j in {0,2}{
                for \i in {0,4,...,14}{
                    {
                        \filldraw (0*\h+\j*\h+\xo, \i*\t+\j*\t+0*\t+\yo) circle (0.04);
                        \filldraw (0*\h+\j*\h+\xo, \i*\t+\j*\t+2*\t+\yo) circle (0.04);
                        \filldraw (1*\h+\j*\h+\xo, \i*\t+\j*\t+1*\t+\yo) circle (0.04);
                        \filldraw (1*\h+\j*\h+\xo, \i*\t+\j*\t+3*\t+\yo) circle (0.04);
                        \draw (0*\h+\j*\h+\xo, \i*\t+\j*\t+0*\t+\yo) -- (1*\h+\j*\h+\xo, \i*\t+\j*\t+1*\t+\yo);
                        \draw (0*\h+\j*\h+\xo, \i*\t+\j*\t+2*\t+\yo) -- (1*\h+\j*\h+\xo, \i*\t+\j*\t+3*\t+\yo);
                        \draw[dotted] (0*\h+\j*\h+\xo, \i*\t+\j*\t+0*\t+\yo) -- (0*\h+\j*\h+\xo, \i*\t+\j*\t+2*\t+\yo);
                        \draw[dotted] (1*\h+\j*\h+\xo, \i*\t+\j*\t+1*\t+\yo) -- (1*\h+\j*\h+\xo, \i*\t+\j*\t+3*\t+\yo);
                        \draw[dotted] (0*\h+\j*\h+\xo, \i*\t+\j*\t+2*\t+\yo) -- (1*\h+\j*\h+\xo, \i*\t+\j*\t+1*\t+\yo);
                    };
                };
                for \i in {2,6,...,12}{
                    {
                        \draw[dotted] (0*\h+\j*\h+\xo, \i*\t+\j*\t+0*\t+\yo) -- (0*\h+\j*\h+\xo, \i*\t+\j*\t+2*\t+\yo);
                        \draw[dotted] (1*\h+\j*\h+\xo, \i*\t+\j*\t+1*\t+\yo) -- (1*\h+\j*\h+\xo, \i*\t+\j*\t+3*\t+\yo);
                        \draw[dotted] (0*\h+\j*\h+\xo, \i*\t+\j*\t+2*\t+\yo) -- (1*\h+\j*\h+\xo, \i*\t+\j*\t+1*\t+\yo);
                    };
                };
            };
            for \j in {1}{
                for \i in {0,4,...,14}{
                    {
                        \draw[dotted] (0*\h+\j*\h+\xo, \i*\t+\j*\t+0*\t+\yo) -- (1*\h+\j*\h+\xo, \i*\t+\j*\t+1*\t+\yo);
                        \draw[dotted] (0*\h+\j*\h+\xo, \i*\t+\j*\t+2*\t+\yo) -- (1*\h+\j*\h+\xo, \i*\t+\j*\t+3*\t+\yo);
                        \draw[dotted] (0*\h+\j*\h+\xo, \i*\t+\j*\t+2*\t+\yo) -- (1*\h+\j*\h+\xo, \i*\t+\j*\t+1*\t+\yo);
                    };
                };
                for \i in {2,6,...,12}{
                    {
                        \draw[dotted] (0*\h+\j*\h+\xo, \i*\t+\j*\t+2*\t+\yo) -- (1*\h+\j*\h+\xo, \i*\t+\j*\t+1*\t+\yo);
                    };
                };
            };
            % 6
            \xo = 21*\h;
            \yo = 7*\t;
            {
                \filldraw[yellow]
                (0*\h+\xo, 0*\t+\yo) --
                (0*\h+\xo, 14*\t+\yo) --
                (3*\h+\xo, 17*\t+\yo) --
                (3*\h+\xo, 3*\t+\yo) --
                cycle;
            };
            for \j in {0,2}{
                for \i in {0,4,...,14}{
                    {
                        \filldraw (0*\h+\j*\h+\xo, \i*\t+\j*\t+0*\t+\yo) circle (0.04);
                        \filldraw (0*\h+\j*\h+\xo, \i*\t+\j*\t+2*\t+\yo) circle (0.04);
                        \filldraw (1*\h+\j*\h+\xo, \i*\t+\j*\t+1*\t+\yo) circle (0.04);
                        \filldraw (1*\h+\j*\h+\xo, \i*\t+\j*\t+3*\t+\yo) circle (0.04);
                        \draw[dotted] (0*\h+\j*\h+\xo, \i*\t+\j*\t+0*\t+\yo) -- (1*\h+\j*\h+\xo, \i*\t+\j*\t+1*\t+\yo);
                        \draw[dotted] (0*\h+\j*\h+\xo, \i*\t+\j*\t+2*\t+\yo) -- (1*\h+\j*\h+\xo, \i*\t+\j*\t+3*\t+\yo);
                        \draw (0*\h+\j*\h+\xo, \i*\t+\j*\t+0*\t+\yo) -- (0*\h+\j*\h+\xo, \i*\t+\j*\t+2*\t+\yo);
                        \draw (1*\h+\j*\h+\xo, \i*\t+\j*\t+1*\t+\yo) -- (1*\h+\j*\h+\xo, \i*\t+\j*\t+3*\t+\yo);
                        \draw[dotted] (0*\h+\j*\h+\xo, \i*\t+\j*\t+2*\t+\yo) -- (1*\h+\j*\h+\xo, \i*\t+\j*\t+1*\t+\yo);
                    };
                };
                for \i in {2,6,...,12}{
                    {
                        \draw[dotted] (0*\h+\j*\h+\xo, \i*\t+\j*\t+0*\t+\yo) -- (0*\h+\j*\h+\xo, \i*\t+\j*\t+2*\t+\yo);
                        \draw[dotted] (1*\h+\j*\h+\xo, \i*\t+\j*\t+1*\t+\yo) -- (1*\h+\j*\h+\xo, \i*\t+\j*\t+3*\t+\yo);
                        \draw[dotted] (0*\h+\j*\h+\xo, \i*\t+\j*\t+2*\t+\yo) -- (1*\h+\j*\h+\xo, \i*\t+\j*\t+1*\t+\yo);
                    };
                };
            };
            for \j in {1}{
                for \i in {0,4,...,14}{
                    {
                        \draw[dotted] (0*\h+\j*\h+\xo, \i*\t+\j*\t+0*\t+\yo) -- (1*\h+\j*\h+\xo, \i*\t+\j*\t+1*\t+\yo);
                        \draw[dotted] (0*\h+\j*\h+\xo, \i*\t+\j*\t+2*\t+\yo) -- (1*\h+\j*\h+\xo, \i*\t+\j*\t+3*\t+\yo);
                        \draw[dotted] (0*\h+\j*\h+\xo, \i*\t+\j*\t+2*\t+\yo) -- (1*\h+\j*\h+\xo, \i*\t+\j*\t+1*\t+\yo);
                    };
                };
                for \i in {2,6,...,12}{
                    {
                        \draw[dotted] (0*\h+\j*\h+\xo, \i*\t+\j*\t+2*\t+\yo) -- (1*\h+\j*\h+\xo, \i*\t+\j*\t+1*\t+\yo);
                    };
                };
            };
            % 7
            \xo = 27*\h;
            \yo = 11*\t;
            {
                \filldraw[yellow]
                (0*\h+\xo, 0*\t+\yo) --
                (0*\h+\xo, 6*\t+\yo) --
                (3*\h+\xo, 9*\t+\yo) --
                (3*\h+\xo, 3*\t+\yo) --
                cycle;
            };
            for \j in {0,1,2,3}{
                for \i in {0,4}{
                    {
                        \filldraw (0*\h+\j*\h+\xo, \i*\t+\j*\t+0*\t+\yo) circle (0.04);
                        \filldraw (0*\h+\j*\h+\xo, \i*\t+\j*\t+2*\t+\yo) circle (0.04);
                        \draw[dotted] (0*\h+\j*\h+\xo, \i*\t+\j*\t+0*\t+\yo) -- (0*\h+\j*\h+\xo, \i*\t+\j*\t+2*\t+\yo);
                    };
                };
                for \i in {2}{
                    {
                        \draw[dotted] (0*\h+\j*\h+\xo, \i*\t+\j*\t+0*\t+\yo) -- (0*\h+\j*\h+\xo, \i*\t+\j*\t+2*\t+\yo);
                    };
                };
            };
            for \j in {0,1,2}{
                for \i in {0,4}{
                    {
                        \draw[dotted] (0*\h+\j*\h+\xo, \i*\t+\j*\t+0*\t+\yo) -- (1*\h+\j*\h+\xo, \i*\t+\j*\t+1*\t+\yo);
                        \draw[dotted] (0*\h+\j*\h+\xo, \i*\t+\j*\t+2*\t+\yo) -- (1*\h+\j*\h+\xo, \i*\t+\j*\t+3*\t+\yo);
                        \draw[dotted] (0*\h+\j*\h+\xo, \i*\t+\j*\t+2*\t+\yo) -- (1*\h+\j*\h+\xo, \i*\t+\j*\t+1*\t+\yo);
                    };
                };
                for \i in {2}{
                    {
                        \draw[dotted] (0*\h+\j*\h+\xo, \i*\t+\j*\t+2*\t+\yo) -- (1*\h+\j*\h+\xo, \i*\t+\j*\t+1*\t+\yo);
                    };
                };
            };
            % Arrows
            {
                \draw (1.5*\h, 15.5*\t) node {$\Rightarrow$};
                \draw (5.5*\h, 15.5*\t) node {$\Rightarrow$};
                \draw (9.5*\h, 15.5*\t) node {$\Rightarrow$};
                \draw (13.5*\h, 15.5*\t) node {$\Rightarrow$};
                \draw (19.5*\h, 15.5*\t) node {$\Rightarrow$};
                \draw (25.5*\h, 15.5*\t) node {$\Rightarrow$};
            };
        }
    \end{tikzpicture}
    \caption{Sequence of steps needed in order to transform a line of $n=16$ amoebots into a parallelogram of side length $\sqrt{n} = 4$.
    }
    \label{fig:shape_formation}
\end{figure}

The amoebots first compute a leader with the help of the global circuit.
The leader is then able to use the global circuit in order to coordinate joint movements.
For example, the leader can instruct all amoebots on the line to expand in a certain direction and perform a rotation\footnote{Note that the rotation primitive can be emulated via expansions and contractions. For more details on these movement primitives, see for example~\cite{DBLP:conf/spaa/DerakhshandehDGRSS14}.} as seen in the first two steps in Figure~\ref{fig:shape_formation}.
By performing further joint movement operations (that only need a constant amount of rounds), one is therefore able to cut the line in half and move the lower half on top of the upper half.
Applying this procedure recursively on the lines for $O(\log n)$ iterations results in the desired parallelogram.

Note that in order to realize these synchronous movements, one has to solve additional problems like compass alignment and chirality agreement (otherwise the amoebots would move or rotate in different directions), which we describe later in this paper.

\subsection{Problem Statement and Our Contribution}
\label{subsec:problem}

In this paper we establish the new circuit model for amoebots and highlight its power by considering the following fundamental problems on a stationary connected structure $S$ of amoebots.
An overview over our results is given by Table~\ref{tab:results}.
As mentioned in the previous section, algorithms for these problems are important building blocks on the way to faster shape transformation algorithms.

First, we show that the amoebots can solve the \emph{consensus problem}, i.e., how they can agree on one out of at most $k$ input values, for some constant $k$.

Then, we study the \emph{compass alignment problem}.
Initially, the amoebots may not agree on a common orientation.
Thus, the goal of the problem is to align the compasses of all amoebots globally.
A compass alignment is essential to coordinate synchronized movements.
Our algorithm requires $O(\log n)$ rounds, w.h.p.
The same approach can also be used to solve the \emph{chirality agreement problem}.
Note that compass alignment and chirality agreement are harder problems than consensus, since one is not able to just use a global circuit in order change the compasses/chiralities of all amoebots at once.

Finally, we look into the \emph{shape recognition problem}.
Amoebots are exposed to environmental influences, which may damage their structure.
In order to detect and repair these structural flaws, the amoebot structure has to check whether its shape matches the desired one.
Having access to simple shape recognition algorithms may be beneficial when checking whether a shape transformation algorithm has reached its desired shape.
We propose algorithms for various classes of shapes.
In particular, we present an $O(1)$-round algorithm for parallelograms with linear side ratio, an $O(\log n)$-round algorithm for parallelograms with polynomial side ratio, and an $O(1)$-round algorithm for shapes composed of triangles.

\begin{table}[ht]
    \caption{
    An overview over our algorithmic results.
    }
    \scriptsize
    \begin{tabularx}{0.99\textwidth}{lXXlX}
        \toprule
        \scriptsize Problem & \scriptsize Minimum Required Pins & \scriptsize Common Chirality & \scriptsize Runtime & \scriptsize Section\\
        \midrule
        \scriptsize Leader election & \scriptsize 1 & \scriptsize No & \scriptsize $\Theta(\log n)$ w.h.p. & \scriptsize Section~\ref{sec:leader} \\
        \scriptsize Consensus & \scriptsize 1 & \scriptsize No & \scriptsize $O(1)$ & \scriptsize Section~\ref{subsec:consensus} \\
        \scriptsize Compass alignment & \scriptsize 2 & \scriptsize Yes & \scriptsize $O(\log n)$ w.h.p.\footnotemark[6] & \scriptsize Section~\ref{subsec:compass} \\
        \scriptsize Chirality agreement & \scriptsize 2 & \scriptsize No & \scriptsize $O(\log n)$ w.h.p. & \scriptsize Section~\ref{subsec:chirality} \\
        \scriptsize Shape recognition & \scriptsize & \scriptsize & \scriptsize & \scriptsize  \\
        \scriptsize ~~Parallelograms & \scriptsize 1 & \scriptsize No & \scriptsize $O(1)$ & \scriptsize Appendix~\ref{app:parallelogram} \\
        \scriptsize ~~Parallelograms with linear side ratio & \scriptsize 1 & \scriptsize No & \scriptsize $\Theta(\log n)$ w.h.p.\footnotemark[7] & \scriptsize Appendix~\ref{app:parallelogram} \\
        \scriptsize ~~Parallelograms with polynomial side ratio & \scriptsize 2 & \scriptsize No & \scriptsize $\Theta(\log n)$ w.h.p. & \scriptsize Appendix~\ref{app:parallelogram} \\
        \scriptsize ~~Universal shape recognition & \scriptsize 2 & \scriptsize Yes & \scriptsize $O(1)$\footnotemark[6] & \scriptsize Section~\ref{subsec:recognition} \\
        \bottomrule
    \end{tabularx}
    \label{tab:results}
\end{table}
\footnotetext[6]{
The required number of pins per bond can be reduced to 1.
However, this requires a local leader election (see Section~\ref{sec:preliminaries}) in order to be able to apply the message transmission primitive (see Appendix~\ref{app:preliminaries}).
Consequently, the runtime changes to $\Theta(\log n)$, w.h.p.
}
\footnotetext[7]{
The runtime is a result of the necessary execution of the leader election algorithm,
and reduces to $O(1)$ under the assumption that a leader is initially given.
}

\subsection{Related Work}
\label{subsec:related_work}

\paragraph{Standard Amoebot Model}
The leader election problem is an extensively researched problem within the geometric amoebot model.
Table~\ref{tab:leader_election} compares various publications.
To our knowledge there are no publications regarding consensus, compass alignment and shape recognition within the standard amoebot model.
The chirality agreement problem was solved by Di Luna et al.~\cite{DBLP:journals/dc/LunaFSVY20}.
Their algorithm requires $O(n)$ rounds.
Further publications include shape formation~\cite{DBLP:conf/dna/DerakhshandehGS15, DBLP:conf/nanocom/DerakhshandehGR15, DBLP:conf/spaa/DerakhshandehGR16, DBLP:journals/dc/LunaFSVY20, DBLP:conf/icdcn/DaymudeGHKSR20}, gathering~\cite{DBLP:conf/podc/CannonDRR16}, and object coating~\cite{DBLP:journals/tcs/DerakhshandehGR17, DBLP:journals/nc/DaymudeDGPRSS18}.

\begin{table}[ht]
    \caption{
    Related work for leader election (updated from~\cite{DBLP:conf/sss/BazziB19, DBLP:conf/icalp/EmekKLM19}).
    $L_{\max}$ denotes the length of the longest boundary.
    $L$ denotes the length of the outer boundary.
    $r(G)$ denotes the radius of $G$.
    $mtree(G)$ denotes the maximum height among all induced subtrees of $G$.
    }
    \scriptsize
    \begin{tabularx}{0.99\textwidth}{llllllllX}
        \toprule
        \scriptsize Paper & \scriptsize Circuts & \scriptsize Deter- & \scriptsize Holes & \scriptsize Weak & \scriptsize Station- & \scriptsize Common & \scriptsize Multiple & \scriptsize Time \\
        & & \scriptsize ministic & & \scriptsize Scheduler & \scriptsize ary & \scriptsize Chirality & \scriptsize Leaders & \\
        \midrule
        \scriptsize \cite{DBLP:conf/algosensors/DaymudeGRSS17} & \scriptsize No & \scriptsize No & \scriptsize Yes & \scriptsize Yes & \scriptsize Yes & \scriptsize Yes & \scriptsize No & \scriptsize $O(L_{\max})$ expected \\
        \scriptsize \cite{DBLP:conf/dna/DerakhshandehGS15} & \scriptsize No & \scriptsize No & \scriptsize Yes & \scriptsize Yes & \scriptsize Yes & \scriptsize Yes & \scriptsize No & \scriptsize $O(L)$ w.h.p. \\
        \scriptsize \cite{DBLP:journals/dc/LunaFSVY20} & \scriptsize No & \scriptsize Yes & \scriptsize No & \scriptsize Yes & \scriptsize No & \scriptsize No & \scriptsize Yes & \scriptsize $O(n^2)$ \\
        \scriptsize \cite{DBLP:conf/algosensors/GastineauAMT18} & \scriptsize No & \scriptsize Yes & \scriptsize No & \scriptsize No & \scriptsize Yes & \scriptsize Yes & \scriptsize No & \scriptsize $r(G) + mtree(G) + 1$ \\
        \scriptsize \cite{DBLP:conf/icalp/EmekKLM19} & \scriptsize No & \scriptsize Yes & \scriptsize Yes & \scriptsize No & \scriptsize No & \scriptsize No & \scriptsize No & \scriptsize $O(Ln^2)$ \\
        \scriptsize \cite{DBLP:conf/sss/BazziB19} & \scriptsize No & \scriptsize Yes & \scriptsize Yes & \scriptsize Yes & \scriptsize Yes & \scriptsize Yes & \scriptsize Yes & \scriptsize $O(n^2)$ \\
        \scriptsize This & \scriptsize Yes & \scriptsize No & \scriptsize Yes & \scriptsize Sync & \scriptsize Yes & \scriptsize No & \scriptsize No & \scriptsize $\Theta(\log n)$ w.h.p. \\
        %This paper & \scriptsize Yes & \scriptsize No & \scriptsize Yes & Sync & \scriptsize Yes & \scriptsize No & \scriptsize No & $\Theta(\log n)$ w.h.p. \\
        \bottomrule
    \end{tabularx}
    \label{tab:leader_election}
\end{table}

\paragraph{Hybrid Programmable Matter}
Hybrid programmable matter is a variant of the standard amoebot model \cite{DBLP:series/lncs/DaymudeHRS19}.
It combines active and passive elements.
A set of active agents are placed on passive hexagonal tiles.
The tiles themselves are immobile and do not have any computational power.
The active agents are able to move the tiles by picking them up.
Moreover, the agents are provided with pebbles, which can be used to mark tiles.

Gmyr et al. \cite{DBLP:conf/mfcs/GmyrHKKRS18} have considered the shape recognition problem with a single agent.
In particular, parallelograms with linear, polynomial, and exponential side ratio are detected.
Unfortunately, the runtime of the algorithms was not analyzed within the paper.
However, the runtime has at least a complexity that is linear to the size of the system.
We are able to adopt and accelerate these results in our model.
Two $\Theta(\log n)$-algorithms for parallelograms with linear and polynomial side ratio are obtained.
These are discussed in Appendix~\ref{app:parallelogram}.
Aside from that, shape formation problems have been considered in \cite{DBLP:journals/nc/GmyrHKKRSS20}.

\paragraph{Reconfigurable Network Models}
Reconfigurable circuits have proven their value in various \emph{reconfigurable network models}, e.g., polymorphic-torus networks \cite{DBLP:journals/tc/LiM89, li1991reconfigurable}, meshes with reconfigurable bus \cite{miller1988meshes, li1991reconfigurable} and bus automata \cite{DBLP:journals/iandc/MoshellR79, DBLP:journals/tsmc/Rothstein88}.
Further examples can be found in \cite{DBLP:journals/iandc/Ben-AsherLPS95, li1991reconfigurable}.
We will discuss the generalized model stated in \cite{DBLP:journals/iandc/Ben-AsherLPS95}.
Usually, a square mesh is utilized as network topology.
A processor with a switch is placed on each node.
For the particular model, the computational power of the processors may vary.
The switches control the connectivity of the incident edges.
The possible connections may be restricted in different variants.
%e.g., (i) general reconfigurable networks (connections are unrestricted), (ii) linear reconfigurable networks (edges can be only connected in pairs or left unconnected), and (iii) directed reconfigurable networks (connections are directed).
A connected component is called a bus.

Each processor connected to a bus may transmit a message on the bus, which is received by all other processors of the same bus.
Contrary to our model, non-primitive messages are possible.
Depending on the specific model, collisons either are detected or go undetected if more than one processor has transmitted a message on a bus.
In the former case the massages are assumed to be destroyed.
In the latter case, for example, a logical OR may be applied to the messages.
Our model corresponds to the latter.

\paragraph{Beeping Model}
Our circuit model shares some similarities with the \emph{beeping model} \cite{DBLP:conf/wdag/CornejoK10}.
In the beeping model, communication between nodes is limited to the following rules: In a single round, a node may choose to either \emph{beep} or \emph{listen}.
A beeping node sends a beep to all of its neighboring nodes and a node that listens can either decide whether at least one of its neighbors beeped in the same round, or if no neighbor beeped.
On a more technical level, the beeping model differs from our circuit model in a sense that nodes are aware of their neighbors in the circuit model, which allows them to exchange (constant-sized) messages with each other, whereas in the beeping model, a listening node that receives a beep is not aware of the particular neighbor that beeped in that specific round.
Several problems like interval coloring~\cite{DBLP:conf/wdag/CornejoK10}, maximal independent set~\cite{DBLP:journals/dc/AfekABCHK13, DBLP:conf/podc/ScottJ013}, consensus~\cite{DBLP:journals/tcs/HounkanliMP20}, rendezvous of two agents~\cite{DBLP:journals/ijfcs/ElouasbiP17}, leader election~\cite{DBLP:conf/wdag/DufoulonBB18,DBLP:conf/wdag/GilbertN15} or clock synchronization~\cite{DBLP:journals/jcss/DolevHJKLRSW16, DBLP:conf/spaa/FKS20, DBLP:journals/ipl/GoudaH90, DBLP:conf/wdag/GuerraouiM15} have already been investigated in the beeping model.

%For the leader election problem, Gilbert and Newport present an algorithm in the beeping model similar to ours that elects a leader among $n$ nodes forming a clique in $O(\log(n + 1/\varepsilon) \log(1/\varepsilon))$ rounds with probability at least $1-1/\varepsilon$.
%However, the runtime of their algorithm increases to $O(\log^2 n)$ if one requires electing a leader with high probability, i.e., with probability at least $1-n^{-c}$ for some constant $c > 0$.
%This is due to their $O\log(1/\varepsilon))$-round termination subroutine that has to be performed in $\Theta(\log n)$ many iterations by the clique.
%In our circuit model, however, we are not limited to one single circuit containing all nodes, but can use a constant number of circuits, which allows for the execution of a termination mechanism that stops the protocol after $\Theta(\log n)$ rounds, even with the nodes not being able to count to $\Theta(\log n)$.

\section{Preliminaries}
\label{sec:preliminaries}

In this section, we discuss simple primitives for local leader election, agreement on the pin labeling, message transmission between neighboring amoebots, and synchronization of instructions.

First, consider the local leader election.
Our goal here is to elect a leader for each pair of neighboring amoebots.
Recall our leader election algorithm in Section~\ref{sec:leader}.
We concurrently perform a tournament on each pair of neighboring amoebots.
The amoebots utilize an \emph{empty pin configuration}, i.e., a pin configuration without any connections.
The coin tosses are sent on all pins simultaneously.
We reduce our leader election algorithm to the second phase due to the small number of candidates.
In order to ensure high probability, we still have to perform the second tournament with all amoebots, which requires a global circuit.
All local leaders are elected after $\Theta(\log n)$ rounds, w.h.p.

Second, consider the agreement on the pin labeling.
%Let $c$ be the number of pins per bond.
%Suppose that $c > 1$.
%Otherwise, the pin labeling would be trivial.
Note that initially, neighboring amoebots only agree on the pin labeling iff they do not share a common chirality.
Certainly, the pin labeling problem can be reduced to the local leader election problem above because once a local leader has been determined for each pair of neighboring amoebots, it can dictate the pin labeling of its neighbor by simply sending a beep on its first pin.
Each pair of neighboring amoebots agrees on the pin labeling after $\Theta(\log n)$ rounds, w.h.p.

\begin{remark}
\label{rem:onepin}
    Observe that a common pin labeling is not necessary if the amoebots have to agree on a single pin.
    We just connect all pins of each incident edge together and consider them as a single one.
\end{remark}

Third, consider the transmission of messages between neighboring amoebots.
Assume that the amoebots agree on a common chirality.
Otherwise, we perform the chirality agreement algorithm (see Section~\ref{subsec:chirality}).
Let each amoebot use the empty pin configuration.
Each amoebot can utilize its first pin to send messages and its last pin to receive messages.
Note that the first pin of each amoebot is connected to the last pin of its neighbor and vice versa.
It has to be clear to both amoebots when the transmission terminates.
This can be done either by fixing the length of the message or by encoding the end within the message.
The amoebots can use a binary encoding for the messages.
We reserve a subsequent round for each bit.
The sending amoebot beeps in the respective round iff the bit is set to 1.
Thus, the transmission of constant-sized messages between two amoebots takes $O(1)$ rounds.
We consider message transmission for a single pin per bond in Appendix~\ref{app:preliminaries}.

Finally, consider the synchronization of instructions.
Different executions of the same instructions may take different amounts of rounds, e.g., due to the randomization.
Hence, it is not obvious for single amoebots whether it is safe to proceed to its next instruction.
To circumvent this problem, we can synchronize the instructions.
In order to do so, we periodically establish the global circuit (see Section~\ref{sec:leader}).
Any amoebot that has not yet finished its current instructions beeps on the circuit.
It is safe to proceed if no amoebot has beeped.

\section{Algorithmic Results}
\label{sec:results}

In this section, we present our algorithmic results.
The definition and motivation of each problem can be found in Section~\ref{subsec:problem}.
More results are given in Appendix~\ref{app:parallelogram}.

\subsection{Consensus}
\label{subsec:consensus}

There is a trivial solution for the consensus problem.
Since the amoebots have only constant memory, we restrict the number of possible input values for the consensus problem to a constant.
W.l.o.g., let $1,\dots,k$ denote these input values.
The amoebot structure establishes a global circuit (see Section~\ref{sec:leader}), which works with an arbitrary number of pins.
Consider $k$ subsequent rounds.
An amoebot with input value $i$ beeps in the $i$-th round.
Thereupon, all amoebots agree on the value of the first round with a beep.
Hence, we obtain the following theorem:

\begin{theorem}
    The consensus problem can be solved within $O(1)$ rounds.
\end{theorem}

\subsection{Compass Alignment}
\label{subsec:compass}

For the compass alignment problem, we assume that the amoebots share a common chirality.
However, in the next section, we show how to resolve this assumption.
Our protocol divides the amoebot structure into regions of amoebots that share the same compass orientation.
%and fuses these iteratively into a single region by adjusting their compasses.
The regions are defined by the connected components of the graph $G_R = (S,A)$ where $A = \{ \{u,v\} \in E \mid u,v \text{ share the same compass orientation} \}$.
Each region $R$ maintains a non-empty set of candidates $C_R$.
Initially, each amoebot is a candidate.
The number of candidates is reduced throughout the algorithm.
Let $\mathcal R$ denote the set of all regions.

We now outline the three phases of a single iteration.
First, each region $R \in \mathcal R$ establishes a regional circuit, which connects all of its amoebots together (see Figures~\ref{fig:regional_circuit1} and~\ref{fig:regional_circuit2}).
This allows for communication within the regions.
The algorithm terminates iff there is only a single regional circuit.
Second, each candidate $u \in C_R$ tosses a coin.
The coin toss of region $R$ is successful iff all coin tosses of its candidates $C_R$ coincide.
Otherwise, the coin toss fails.
Third, each region that has tossed $\mathit{TAILS}$ fuses either with a neighboring region that has tossed $\mathit{HEADS}$ or with a neighboring region that has failed its coin toss.
A region can fuse into another region by adjusting its compass.
Amoebots on the boundary of the region are able to learn about the compass orientations and the coin tosses of neighboring regions.
The amoebots of the region agree on the orientation that requires the least number of clockwise rotations.
Note that further fusions may occur when two neighbor regions adjust their compasses to the same orientation (see Figure~\ref{fig:fusion}).
In the following, we explain the three steps in more detail.

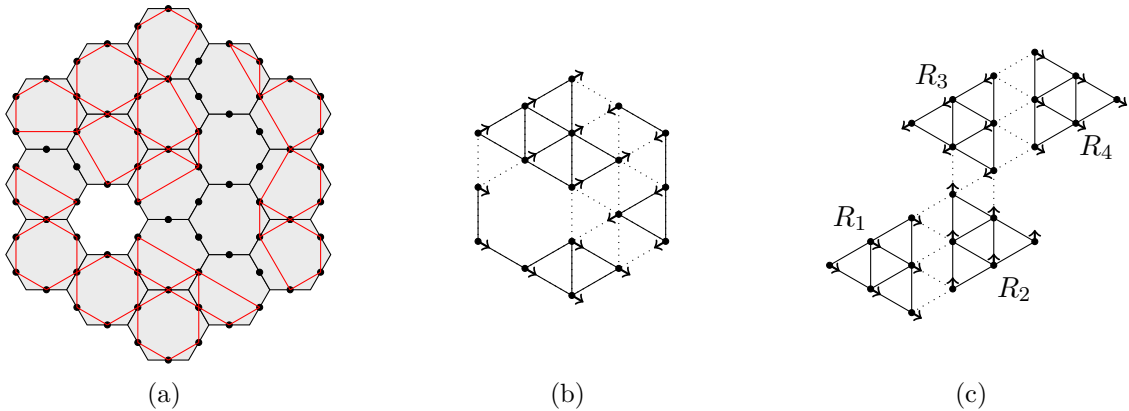
\begin{figure}[htp]
    \centering
    \begin{subfigure}[b]{0.35\textwidth}
        \centering
        \begin{tikzpicture}
        \tikzmath{
            \scale = .6*.9;
            \pins = 1;%
            \fr = 235;
            \fg = 235;
            \fb = 235;
            for \i in {2,3,4}{
                Amoebot(0,\i,0);
            };
            for \i in {1,3,4}{
                Amoebot(1,\i,0);
            };
            for \i in {0,1,2,3,4}{
                Amoebot(2,\i,0);
            };
            for \i in {0,1,2,3}{
                Amoebot(3,\i,0);
            };
            for \i in {0,1,2}{
                Amoebot(4,\i,0);
            };
            \lr = 255;
            \lg = 0;
            \lb = 0;
            Circuit(0,2,0,1,1,2,1);
            Circuit(0,2,0,2,1,3,1);
            Circuit(0,2,0,3,1,4,1);
            Circuit(0,2,0,4,1,5,1);
            Circuit(0,2,0,5,1,6,1);
            Circuit(0,2,0,6,1,1,1);
            Circuit(0,3,0,1,1,2,1);
            Circuit(0,3,0,5,1,6,1);
            Circuit(0,3,0,6,1,1,1);
            Circuit(0,3,0,2,1,5,1);
            Circuit(0,4,0,2,1,3,1);
            Circuit(0,4,0,3,1,4,1);
            Circuit(0,4,0,4,1,5,1);
            Circuit(0,4,0,5,1,6,1);
            Circuit(0,4,0,2,1,6,1);
            Circuit(1,1,0,1,1,2,1);
            Circuit(1,1,0,2,1,3,1);
            Circuit(1,1,0,3,1,4,1);
            Circuit(1,1,0,4,1,5,1);
            Circuit(1,1,0,5,1,6,1);
            Circuit(1,1,0,6,1,1,1);
            Circuit(1,3,0,1,1,2,1);
            Circuit(1,3,0,2,1,3,1);
            Circuit(1,3,0,3,1,4,1);
            Circuit(1,3,0,4,1,5,1);
            Circuit(1,3,0,1,1,5,1);
            Circuit(1,4,0,1,1,2,1);
            Circuit(1,4,0,2,1,3,1);
            Circuit(1,4,0,3,1,4,1);
            Circuit(1,4,0,4,1,5,1);
            Circuit(1,4,0,5,1,6,1);
            Circuit(1,4,0,6,1,1,1);
            Circuit(2,0,0,1,1,2,1);
            Circuit(2,0,0,2,1,3,1);
            Circuit(2,0,0,3,1,4,1);
            Circuit(2,0,0,4,1,5,1);
            Circuit(2,0,0,5,1,6,1);
            Circuit(2,0,0,6,1,1,1);
            Circuit(2,1,0,1,1,2,1);
            Circuit(2,1,0,5,1,6,1);
            Circuit(2,1,0,6,1,1,1);
            Circuit(2,1,0,2,1,5,1);
            Circuit(2,2,0,3,1,4,1);
            Circuit(2,2,0,4,1,5,1);
            Circuit(2,2,0,5,1,6,1);
            Circuit(2,2,0,3,1,6,1);
            Circuit(2,3,0,1,1,2,1);
            Circuit(2,3,0,4,1,5,1);
            Circuit(2,3,0,5,1,6,1);
            Circuit(2,3,0,6,1,1,1);
            Circuit(2,3,0,2,1,4,1);
            Circuit(2,4,0,3,1,4,1);
            Circuit(2,4,0,4,1,5,1);
            Circuit(2,4,0,5,1,6,1);
            Circuit(2,4,0,6,1,1,1);
            Circuit(2,4,0,1,1,3,1);
            Circuit(3,0,0,1,1,2,1);
            Circuit(3,0,0,5,1,6,1);
            Circuit(3,0,0,6,1,1,1);
            Circuit(3,0,0,2,1,5,1);
            Circuit(3,1,0,2,1,3,1);
            Circuit(3,2,0,5,1,6,1);
            Circuit(3,3,0,2,1,3,1);
            Circuit(3,3,0,3,1,4,1);
            Circuit(3,3,0,2,1,4,1);
            Circuit(4,0,0,1,1,2,1);
            Circuit(4,0,0,2,1,3,1);
            Circuit(4,0,0,3,1,4,1);
            Circuit(4,0,0,4,1,5,1);
            Circuit(4,0,0,1,1,5,1);
            Circuit(4,1,0,1,1,2,1);
            Circuit(4,1,0,2,1,3,1);
            Circuit(4,1,0,3,1,4,1);
            Circuit(4,1,0,6,1,1,1);
            Circuit(4,1,0,4,1,6,1);
            Circuit(4,2,0,1,1,2,1);
            Circuit(4,2,0,2,1,3,1);
            Circuit(4,2,0,3,1,4,1);
            Circuit(4,2,0,4,1,5,1);
            Circuit(4,2,0,1,1,5,1);
        }
        \end{tikzpicture}
        \caption{\centering}
        \label{fig:regional_circuit1}
    \end{subfigure}
    \hfill
    \begin{subfigure}[b]{0.25\textwidth}
        \centering
        \begin{tikzpicture}
            \tikzmath{
                \scale = .8*.9;
                \t = \scale/2;
                \h = sqrt(\scale^2 - (\scale/2)^2);
                % bounding box for better alignment
                {
                    \draw[color=white] (-\t, -2.5*\t) rectangle (4*\h+\t, 9*\t);
                };
                % particles
                for \i in {2,4,6}{
                    {
                        \filldraw (0*\h, \i*\t) circle (0.04);
                    };
                };
                for \i in {1,5,7}{
                    {
                        \filldraw (1*\h, \i*\t) circle (0.04);
                    };
                };
                for \i in {0,2,4,6,8}{
                    {
                        \filldraw (2*\h, \i*\t) circle (0.04);
                    };
                };
                for \i in {1,3,5,7}{
                    {
                        \filldraw (3*\h, \i*\t) circle (0.04);
                    };
                };
                for \i in {2,4,6}{
                    {
                        \filldraw (4*\h, \i*\t) circle (0.04);
                    };
                };
                {
                % bonds
                    % N - S
                    \draw[dotted] (0*\h, 2*\t) -- (0*\h, 6*\t);
                    \draw[dotted] (1*\h, 5*\t) -- (1*\h, 7*\t);
                    \draw[dotted] (2*\h, 0*\t) -- (2*\h, 8*\t);
                    \draw[dotted] (3*\h, 1*\t) -- (3*\h, 7*\t);
                    \draw[dotted] (4*\h, 2*\t) -- (4*\h, 6*\t);
                    % SW - NE
                    \draw[dotted] (0*\h, 6*\t) -- (2*\h, 8*\t);
                    \draw[dotted] (0*\h, 4*\t) -- (3*\h, 7*\t);
                    \draw[dotted] (2*\h, 4*\t) -- (4*\h, 6*\t);
                    \draw[dotted] (1*\h, 1*\t) -- (4*\h, 4*\t);
                    \draw[dotted] (2*\h, 0*\t) -- (4*\h, 2*\t);
                    % NW - SE
                    \draw[dotted] (2*\h, 8*\t) -- (4*\h, 6*\t);
                    \draw[dotted] (1*\h, 7*\t) -- (4*\h, 4*\t);
                    \draw[dotted] (0*\h, 6*\t) -- (4*\h, 2*\t);
                    \draw[dotted] (2*\h, 2*\t) -- (3*\h, 1*\t);
                    \draw[dotted] (0*\h, 2*\t) -- (2*\h, 0*\t);
                % orientations
                    % Region NW
                    \draw[thick, ->] (0*\h, 6*\t) -- +(\h/4,\t/4);
                    \draw[thick, ->] (1*\h, 5*\t) -- +(\h/4,\t/4);
                    \draw[thick, ->] (1*\h, 7*\t) -- +(\h/4,\t/4);
                    \draw[thick, ->] (2*\h, 4*\t) -- +(\h/4,\t/4);
                    \draw[thick, ->] (2*\h, 6*\t) -- +(\h/4,\t/4);
                    \draw[thick, ->] (2*\h, 8*\t) -- +(\h/4,\t/4);
                    \draw[thick, ->] (3*\h, 5*\t) -- +(\h/4,\t/4);
                    % Region E
                    \draw[thick, ->] (3*\h, 3*\t) -- +(-\h/4,-\t/4);
                    \draw[thick, ->] (3*\h, 7*\t) -- +(-\h/4,-\t/4);
                    \draw[thick, ->] (4*\h, 2*\t) -- +(-\h/4,-\t/4);
                    \draw[thick, ->] (4*\h, 4*\t) -- +(-\h/4,-\t/4);
                    \draw[thick, ->] (4*\h, 6*\t) -- +(-\h/4,-\t/4);
                    % Region SW
                    \draw[thick, ->] (0*\h, 2*\t) -- +(\h/4,-\t/4);
                    \draw[thick, ->] (0*\h, 4*\t) -- +(\h/4,-\t/4);
                    \draw[thick, ->] (1*\h, 1*\t) -- +(\h/4,-\t/4);
                    \draw[thick, ->] (2*\h, 0*\t) -- +(\h/4,-\t/4);
                    \draw[thick, ->] (2*\h, 2*\t) -- +(\h/4,-\t/4);
                    \draw[thick, ->] (3*\h, 1*\t) -- +(\h/4,-\t/4);
                % circuits
                    % N - S
                    \draw (0*\h, 2*\t) -- (0*\h, 4*\t);
                    \draw (1*\h, 5*\t) -- (1*\h, 7*\t);
                    \draw (2*\h, 0*\t) -- (2*\h, 2*\t);
                    \draw (2*\h, 4*\t) -- (2*\h, 8*\t);
                    \draw (4*\h, 2*\t) -- (4*\h, 6*\t);
                    % SW - NE
                    \draw (0*\h, 6*\t) -- (2*\h, 8*\t);
                    \draw (1*\h, 5*\t) -- (2*\h, 6*\t);
                    \draw (2*\h, 4*\t) -- (3*\h, 5*\t);
                    \draw (1*\h, 1*\t) -- (2*\h, 2*\t);
                    \draw (3*\h, 3*\t) -- (4*\h, 4*\t);
                    \draw (2*\h, 0*\t) -- (3*\h, 1*\t);
                    % NW - SE
                    \draw (3*\h, 7*\t) -- (4*\h, 6*\t);
                    \draw (1*\h, 7*\t) -- (3*\h, 5*\t);
                    \draw (0*\h, 6*\t) -- (2*\h, 4*\t);
                    \draw (3*\h, 3*\t) -- (4*\h, 2*\t);
                    \draw (2*\h, 2*\t) -- (3*\h, 1*\t);
                    \draw (0*\h, 2*\t) -- (2*\h, 0*\t);
                };
            }
        \end{tikzpicture}
        \caption{\centering}
        \label{fig:regional_circuit2}
    \end{subfigure}
    \hfill
    \begin{subfigure}[b]{0.35\textwidth}
        \centering
        \begin{tikzpicture}
            \tikzmath{
                \scale = .7*.9;
                \t = \scale/2;
                \h = sqrt(\scale^2 - (\scale/2)^2);
                {
                    % dummy
                    \filldraw[white] (0*\h0, -2*\t) circle (0.04);
                };
                \x = 0;
                \y = 2;
                {
                    % right
                    \filldraw (\x*\h+0*\h, \y*\t+0*\t) circle (0.04);
                    \filldraw (\x*\h+1*\h, \y*\t-1*\t) circle (0.04);
                    \filldraw (\x*\h+1*\h, \y*\t+1*\t) circle (0.04);
                    \filldraw (\x*\h+2*\h, \y*\t-2*\t) circle (0.04);
                    \filldraw (\x*\h+2*\h, \y*\t+0*\t) circle (0.04);
                    \filldraw (\x*\h+2*\h, \y*\t+2*\t) circle (0.04);
                    \draw (\x*\h+0*\h, \y*\t+0*\t) -- (\x*\h+2*\h, \y*\t-2*\t);
                    \draw (\x*\h+2*\h, \y*\t-2*\t) -- (\x*\h+2*\h, \y*\t+2*\t);
                    \draw (\x*\h+2*\h, \y*\t+2*\t) -- (\x*\h+0*\h, \y*\t+0*\t);
                    \draw (\x*\h+1*\h, \y*\t-1*\t) -- (\x*\h+1*\h, \y*\t+1*\t);
                    \draw (\x*\h+1*\h, \y*\t+1*\t) -- (\x*\h+2*\h, \y*\t+0*\t);
                    \draw (\x*\h+2*\h, \y*\t+0*\t) -- (\x*\h+1*\h, \y*\t-1*\t);
                    \draw[thick,->] (\x*\h+0*\h, \y*\t+0*\t) -- +(\h/4,-\t/4);
                    \draw[thick,->] (\x*\h+1*\h, \y*\t-1*\t) -- +(\h/4,-\t/4);
                    \draw[thick,->] (\x*\h+1*\h, \y*\t+1*\t) -- +(\h/4,-\t/4);
                    \draw[thick,->] (\x*\h+2*\h, \y*\t-2*\t) -- +(\h/4,-\t/4);
                    \draw[thick,->] (\x*\h+2*\h, \y*\t+0*\t) -- +(\h/4,-\t/4);
                    \draw[thick,->] (\x*\h+2*\h, \y*\t+2*\t) -- +(\h/4,-\t/4);
                };
                \x = 2;
                \y = 8;
                {
                    % right
                    \filldraw (\x*\h+0*\h, \y*\t+0*\t) circle (0.04);
                    \filldraw (\x*\h+1*\h, \y*\t-1*\t) circle (0.04);
                    \filldraw (\x*\h+1*\h, \y*\t+1*\t) circle (0.04);
                    \filldraw (\x*\h+2*\h, \y*\t-2*\t) circle (0.04);
                    \filldraw (\x*\h+2*\h, \y*\t+0*\t) circle (0.04);
                    \filldraw (\x*\h+2*\h, \y*\t+2*\t) circle (0.04);
                    \draw (\x*\h+0*\h, \y*\t+0*\t) -- (\x*\h+2*\h, \y*\t-2*\t);
                    \draw (\x*\h+2*\h, \y*\t-2*\t) -- (\x*\h+2*\h, \y*\t+2*\t);
                    \draw (\x*\h+2*\h, \y*\t+2*\t) -- (\x*\h+0*\h, \y*\t+0*\t);
                    \draw (\x*\h+1*\h, \y*\t-1*\t) -- (\x*\h+1*\h, \y*\t+1*\t);
                    \draw (\x*\h+1*\h, \y*\t+1*\t) -- (\x*\h+2*\h, \y*\t+0*\t);
                    \draw (\x*\h+2*\h, \y*\t+0*\t) -- (\x*\h+1*\h, \y*\t-1*\t);
                    \draw[thick,->] (\x*\h+0*\h, \y*\t+0*\t) -- +(-\h/4,-\t/4);
                    \draw[thick,->] (\x*\h+1*\h, \y*\t-1*\t) -- +(-\h/4,-\t/4);
                    \draw[thick,->] (\x*\h+1*\h, \y*\t+1*\t) -- +(-\h/4,-\t/4);
                    \draw[thick,->] (\x*\h+2*\h, \y*\t-2*\t) -- +(-\h/4,-\t/4);
                    \draw[thick,->] (\x*\h+2*\h, \y*\t+0*\t) -- +(-\h/4,-\t/4);
                    \draw[thick,->] (\x*\h+2*\h, \y*\t+2*\t) -- +(-\h/4,-\t/4);
                };
                \x = 5;
                \y = 3;
                {
                    % right
                    \filldraw (\x*\h-0*\h, \y*\t+0*\t) circle (0.04);
                    \filldraw (\x*\h-1*\h, \y*\t-1*\t) circle (0.04);
                    \filldraw (\x*\h-1*\h, \y*\t+1*\t) circle (0.04);
                    \filldraw (\x*\h-2*\h, \y*\t-2*\t) circle (0.04);
                    \filldraw (\x*\h-2*\h, \y*\t+0*\t) circle (0.04);
                    \filldraw (\x*\h-2*\h, \y*\t+2*\t) circle (0.04);
                    \draw (\x*\h-0*\h, \y*\t+0*\t) -- (\x*\h-2*\h, \y*\t-2*\t);
                    \draw (\x*\h-2*\h, \y*\t-2*\t) -- (\x*\h-2*\h, \y*\t+2*\t);
                    \draw (\x*\h-2*\h, \y*\t+2*\t) -- (\x*\h-0*\h, \y*\t+0*\t);
                    \draw (\x*\h-1*\h, \y*\t-1*\t) -- (\x*\h-1*\h, \y*\t+1*\t);
                    \draw (\x*\h-1*\h, \y*\t+1*\t) -- (\x*\h-2*\h, \y*\t+0*\t);
                    \draw (\x*\h-2*\h, \y*\t+0*\t) -- (\x*\h-1*\h, \y*\t-1*\t);
                    \draw[thick,->] (\x*\h-0*\h, \y*\t+0*\t) -- +(0,+\t/2);
                    \draw[thick,->] (\x*\h-1*\h, \y*\t-1*\t) -- +(0,+\t/2);
                    \draw[thick,->] (\x*\h-1*\h, \y*\t+1*\t) -- +(0,+\t/2);
                    \draw[thick,->] (\x*\h-2*\h, \y*\t-2*\t) -- +(0,+\t/2);
                    \draw[thick,->] (\x*\h-2*\h, \y*\t+0*\t) -- +(0,+\t/2);
                    \draw[thick,->] (\x*\h-2*\h, \y*\t+2*\t) -- +(0,+\t/2);
                };
                \x = 7;
                \y = 9;
                {
                    % right
                    \filldraw (\x*\h-0*\h, \y*\t+0*\t) circle (0.04);
                    \filldraw (\x*\h-1*\h, \y*\t-1*\t) circle (0.04);
                    \filldraw (\x*\h-1*\h, \y*\t+1*\t) circle (0.04);
                    \filldraw (\x*\h-2*\h, \y*\t-2*\t) circle (0.04);
                    \filldraw (\x*\h-2*\h, \y*\t+0*\t) circle (0.04);
                    \filldraw (\x*\h-2*\h, \y*\t+2*\t) circle (0.04);
                    \draw (\x*\h-0*\h, \y*\t+0*\t) -- (\x*\h-2*\h, \y*\t-2*\t);
                    \draw (\x*\h-2*\h, \y*\t-2*\t) -- (\x*\h-2*\h, \y*\t+2*\t);
                    \draw (\x*\h-2*\h, \y*\t+2*\t) -- (\x*\h-0*\h, \y*\t+0*\t);
                    \draw (\x*\h-1*\h, \y*\t-1*\t) -- (\x*\h-1*\h, \y*\t+1*\t);
                    \draw (\x*\h-1*\h, \y*\t+1*\t) -- (\x*\h-2*\h, \y*\t+0*\t);
                    \draw (\x*\h-2*\h, \y*\t+0*\t) -- (\x*\h-1*\h, \y*\t-1*\t);
                    \draw[thick,->] (\x*\h-0*\h, \y*\t+0*\t) -- +(\h/4,-\t/4);
                    \draw[thick,->] (\x*\h-1*\h, \y*\t-1*\t) -- +(\h/4,-\t/4);
                    \draw[thick,->] (\x*\h-1*\h, \y*\t+1*\t) -- +(\h/4,-\t/4);
                    \draw[thick,->] (\x*\h-2*\h, \y*\t-2*\t) -- +(\h/4,-\t/4);
                    \draw[thick,->] (\x*\h-2*\h, \y*\t+0*\t) -- +(\h/4,-\t/4);
                    \draw[thick,->] (\x*\h-2*\h, \y*\t+2*\t) -- +(\h/4,-\t/4);
                };
                {
                    % additional bonds
                    \draw[dotted] (2*\h, 0*\t) -- (3*\h, 1*\t);
                    \draw[dotted] (2*\h, 2*\t) -- (3*\h, 3*\t);
                    \draw[dotted] (2*\h, 4*\t) -- (5*\h, 7*\t);
                    \draw[dotted] (4*\h, 8*\t) -- (5*\h, 9*\t);
                    \draw[dotted] (4*\h, 10*\t) -- (5*\h, 11*\t);
                    \draw[dotted] (3*\h, 5*\t) -- (3*\h, 7*\t);
                    \draw[dotted] (4*\h, 4*\t) -- (4*\h, 6*\t);
                    \draw[dotted] (2*\h, 2*\t) -- (3*\h, 1*\t);
                    \draw[dotted] (2*\h, 4*\t) -- (3*\h, 3*\t);
                    \draw[dotted] (4*\h, 8*\t) -- (5*\h, 7*\t);
                    \draw[dotted] (4*\h, 10*\t) -- (5*\h, 9*\t);
                    % labels
                    \draw (0.5*\h,4*\t) node {$R_1$};
                    \draw (4.5*\h,1*\t) node {$R_2$};
                    \draw (2.5*\h,10*\t) node {$R_3$};
                    \draw (6.5*\h,7*\t) node {$R_4$};
                };
            }
        \end{tikzpicture}
        \caption{\centering}
        \label{fig:fusion}
    \end{subfigure}
    \caption{
        (a) and (b) show an example for regional circuits.
        (a) The red edges indicate the regional circuits.
        Each regional circuit connects the amoebots of a region together.
        (b) shows graph $G_R = (S,A)$.
        The solid edges belong to the edge set $A$.
        (c) Region $R_1$ and $R_4$ have the same orientation.
        Region $R_2$ adjusts its compass to region $R_1$.
        Region $R_3$ adjusts its compass to region $R_4$.
        All 4 regions fuse together.
    }
\end{figure}

First, consider the establishment of the regional circuits.
Each pair of neighboring amoebots $u,v \in S$ has to determine whether they share the same compass orientation, i.e., whether $\{ u, v \} \in A$ (see Figure~\ref{fig:regional_circuit2}).
For that purpose, $u$ sends a message to $v$ containing its orientation with respect to the edge $\{ u,v \}$.
Recall that we assume common chirality.
Thus, $v$ is able to interpret the received message correctly.
We define $\operatorname{offset}(u,v)$ as the minimal number of clockwise rotations (by $60^\circ$) necessary to adjust the compass of $u$ to the compass of $v$.
Note that $\operatorname{offset}(u,v) \in [0, 5]$.
Each amoebot is able to compute the offset to all of its neighbors.

Each amoebot $u \in S$ connects the pins on the bonds to $\operatorname A(u) = \{ v \in \operatorname N(u) \mid \{u,v\} \in A\}$ (see Figure~\ref{fig:regional_circuit1}).
It is easy to see that this connects the amoebots of each region.
Note that each amoebot is connected to exactly one regional circuit.
Consider region $R \in \mathcal R$.
Let $B_R = \{ u \in R \mid \operatorname A(u) \neq \operatorname N(u) \}$ denote the boundary amoebots.
Observe that $B_R = \emptyset$ implies that all amoebots in the amoebot structure share a common compass orientation, i.e., $|\mathcal R| = 1$.
We let each amoebot $u \in B_R$ beep on the regional circuit in a predefined round.
The algorithm terminates once a round is reached where the regional circuit has not been activated.

Next, consider the coin tosses of the regions.
We use a similar approach as for the leader election.
Each candidate $u \in C_R$ tosses a coin and stores the result in a variable $u.c$.
Consider two subsequent rounds $r_1, r_2$.
Each candidate $u \in C_R$ sends a beep through its regional circuit in round $r_1$ if $u.c = \mathit{HEADS}$, and in round $r_2$ if $u.c = \mathit{TAILS}$.
Each amoebot $v \in R$ is able to determine if there is a candidate $u \in C_R$ with $u.c = \mathit{HEADS}$ and $u.c = \mathit{TAILS}$, respectively.
It stores the result of the coin toss in a variable $v.c_R$.
If all coin tosses coincide, i.e., if there is only a beep in one of the two rounds, it sets the variable $v.c_R$ to the respective value.
Otherwise, it sets the variable $v.c_R = \mathit{FAILED}$.
Each candidate $u \in C_R$ withdraws its candidacy iff $u.c = \mathit{TAILS}$ and $v.c_R = \mathit{FAILED}$.
Note that there is still a candidate $u \in C_R$ with $u.c = \mathit{HEADS}$.
We let each amoebot $u \in S$ send a message containing $u.c_R$ to all its neighbors.

Finally, consider the fusion of regions.
Suppose that region $R$ has tossed $\mathit{TAILS}$.
The boundary amoebots $B_R$ have obtained information about the neighboring regions in the previous steps.
In particular, they have learned about the offsets and the results of the coin tosses.
Note that the offsets are equal for each amoebot $u \in R$.
We divide the boundary amoebots into 5 subsets according to this information.
Let $\forall i \in [1,5] : B_{R,i} = \{ u \in B_R \mid \exists v \in \operatorname N(u) : \operatorname{offset}(u,v) = i \land v.c_R \neq \mathit{TAILS} \}$.
$B_{R,i} \neq \emptyset$ implies that region $R$ can fuse into another region by rotating its amoebots $i$ times in clockwise direction.
Note that the subsets are not necessarily disjoint.
Consider five subsequent rounds $r_1, \dots, r_5$.
Each boundary amoebot $u \in B_{R,i}$ beeps on the regional circuit in round $r_i$.
Each amoebot $v \in R$ is able to determine if $B_{R,i} \neq \emptyset$ for all $i \in [1,5]$.
It proceeds to the next iteration if $B_{R,i} = \emptyset$ for all $i \in [1,5]$.
Otherwise, it rotates its compass according to the first received beep.
Furthermore, each candidate $c \in C_R$ withdraws each candidacy in this case.

For the purpose of analysis, we add a time stamp to the set of regions. 
Let $\mathcal{R}_0$ denote the set of all regions before the execution of our algorithm and let $\mathcal{R}_t$ denote the set of all regions after the $t$-th iteration and $\mathcal{R}_0$ for $t \geq 1$.

\begin{lemma}
\label{lem:alignment1}
    For all $t \in \N_0$ each region $R \in \mathcal{R}_t$ contains at least one candidate, i.e., $C_R \neq \emptyset$.
\end{lemma}

\begin{proof}
    The statement can be proven by a simple induction.
    The full proof can be found in Appendix~\ref{app:alignment}.
\end{proof}

\begin{lemma}
\label{lem:alignment2}
    After $O(\log n)$ many rounds, there is only a single candidate left in total, i.e., $|\bigcup_{R \in \mathcal R} C_R| = 1$, w.h.p.
\end{lemma}

\begin{proof}[Proof (Sketch)]
    We only outline the proof here.
    The full proof can be found in Appendix~\ref{app:alignment}.
    Let $C_t$ denote the candidates over all regions after the $t$-th iteration. For each candidate $u \in C_t$, we define
    \begin{equation*}
        Y_u =
        \begin{cases}
            1, & \text{if } u \in C_{t+1} \\
            0, & \text{otherwise}.
        \end{cases}
    \end{equation*}
    By case analysis, we can show that $\E[Y_u] \leq \frac{3}{4}$ holds.    
    Let $X_t = |C_t| - 1$ denote the number of candidates after the $t$-th iteration (minus one).
    The value for iteration $t+1$ can be expressed by $X_{t+1} = \sum_{u \in C_t}Y_u - 1$.
    We obtain $\E[X_{t+1} \mid X_t] \leq \frac{3}{4} X_t$.
    This implies for any constant $c > 1$ that
    \begin{equation*}
        \E[X_{t+c \log n} \mid X_t] \leq \left(\frac{3}{4}\right)^{c\log n} X_t \leq \left(\frac{3}{4}\right)^{c\log n} n \leq n^{-c'}
    \end{equation*}
    for some $c' = \Theta(c)$.
    For the second inequality we are utilizing the fact that $X_t \leq n$.
    By applying the Markov inequality, we obtain $\Pr[X_{t + c \log n} \geq 1] \leq n^{-c'}$.
    Note that a single iteration requires $O(1)$ rounds.
    Thus, $|\bigcup_{R \in \mathcal R} C_R| = 1$ holds after $O(\log n)$ rounds, w.h.p.
\end{proof}

Lemma~\ref{lem:alignment2} implies that there is only one region left containing the only remaining candidate.
Combining Lemmas~\ref{lem:alignment1} and \ref{lem:alignment2}, we obtain the following theorem:

\begin{theorem}
    There is a protocol that aligns the compasses of all amoebots after $O(\log n)$ rounds, w.h.p.
\end{theorem}

\begin{remark}
    Note that the termination of our protocol does not imply a leader election because the remaining region may still contain several candidates.
\end{remark}
    
\subsection{Chirality Agreement}
\label{subsec:chirality}

In general, it is impossible for amoebots with only a single pin per bond to arrive at an agreement on the chirality without movements.
Already an amoebot structure consisting of two (adjacent) amoebots is not able to break the symmetry, if they cannot change their shape.
Therefore, we consider amoebots with at least two pins per bond.
We can apply the same approach as for the compass alignment.
The regions are redefined by the connected components of graph $G_R = (S,A)$ where $A = \{ \{u,v\} \in E \mid u,v \text{ agree on the chirality} \}$.
Each amoebot activates its first pin according to its chirality.
%Note that this pin is not necessarily the pin corresponding to the agreed on labeling.
Neighboring amoebots beep on two distinct pins iff they share a common chirality.
We obtain the following corollary:

\begin{corollary}
    There is a protocol that lets the amoebot structure agree on the chirality after $O(\log n)$ rounds, w.h.p.
\end{corollary}

\subsection{Universal Shape Recognition}
\label{subsec:recognition}

In order to keep the shapes as generic as possible, we consider shapes composed of triangles.
We adopt the definition for shapes by \cite{DBLP:conf/spaa/DerakhshandehGR16}:
A shape $\mathcal S$ is a finite set of faces on the infinite regular triangular grid graph~$G_{eqt}$ (see Figure~\ref{fig:shape}).
The number of faces is assumed to be constant.
A shape $\mathcal S$ is connected iff the corresponding nodes in the induced subgraph in the dual graph of $G_{eqt}$ are connected.
Let $T(\mathcal S)$ be a transformation of a shape $\mathcal S$ defined by a translation, a rotation, and an isotropic scaling such that the vertices of the triangles forming $\mathcal S$ coincide with vertices in $G_{eqt}$.
Let $\sigma$ denote the scaling factor of $T(\mathcal S)$.
Let $V(T(\mathcal S))$ be the set of all nodes that lie on a vertex of $T(\mathcal S)$, on a edge of $T(\mathcal S)$ or inside of $T(\mathcal S)$.
We call $V(T(\mathcal S))$ a \emph{representation} of shape $\mathcal S$.
In addition to \cite{DBLP:conf/spaa/DerakhshandehGR16}, $\mathcal S$ is minimal iff $\forall \mathcal S': \exists T, T': V(T(\mathcal S))=V(T'(\mathcal S')) \Rightarrow |\mathcal S|\leq|\mathcal S'|$ (see Figures~\ref{fig:shape3} and~\ref{fig:shape4}).

\begin{figure}[htb]
    \centering
    \begin{subfigure}[b]{0.19\textwidth}
        \centering
        \begin{tikzpicture}
            \tikzmath{
                \scale = .3;
                \t = \scale/2;
                \h = sqrt(\scale^2 - (\scale/2)^2);
                % shape
                {
                    \filldraw[yellow]
                    (0*\h, 6*\t) --
                    (0*\h, 10*\t) --
                    (2*\h, 8*\t) --
                    (6*\h, 12*\t) --
                    (6*\h, 4*\t) --
                    (8*\h, 2*\t) --
                    (6*\h, 0*\t) --
                    cycle;
                };
                % particles
                for \i in {6,10}{
                    {
                        \filldraw (0*\h, \i*\t) circle (0.04);
                    };
                };
                for \i in {4,8}{
                    {
                        \filldraw (2*\h, \i*\t) circle (0.04);
                    };
                };
                for \i in {2,6,10}{
                    {
                        \filldraw (4*\h, \i*\t) circle (0.04);
                    };
                };
                for \i in {0,4,8,12}{
                    {
                        \filldraw (6*\h, \i*\t) circle (0.04);
                    };
                };
                for \i in {2}{
                    {
                        \filldraw (8*\h, \i*\t) circle (0.04);
                    };
                };
                {
                % bonds
                    % N - S
                    \draw (0*\h, 6*\t) -- (0*\h, 10*\t);
                    \draw (2*\h, 4*\t) -- (2*\h, 8*\t);
                    \draw (4*\h, 2*\t) -- (4*\h, 10*\t);
                    \draw (6*\h, 0*\t) -- (6*\h, 12*\t);
                    % SW - NE
                    \draw (0*\h, 6*\t) -- (6*\h, 12*\t);
                    \draw (2*\h, 4*\t) -- (6*\h, 8*\t);
                    \draw (4*\h, 2*\t) -- (6*\h, 4*\t);
                    \draw (6*\h, 0*\t) -- (8*\h, 2*\t);
                    % NW - SE
                    \draw (4*\h, 10*\t) -- (6*\h, 8*\t);
                    \draw (0*\h, 10*\t) -- (8*\h, 2*\t);
                    \draw (0*\h, 6*\t) -- (6*\h, 0*\t);
                };
            }
        \end{tikzpicture}
        \caption{\centering}
        \label{fig:shape1}
    \end{subfigure}
    \hfill
    \begin{subfigure}[b]{0.19\textwidth}
        \centering
        \begin{tikzpicture}
            \tikzmath{
                \scale = .3;
                \t = \scale/2;
                \h = sqrt(\scale^2 - (\scale/2)^2);
                % shape
                {
                    \filldraw[yellow]
                    (0*\h, 6*\t) --
                    (0*\h, 10*\t) --
                    (2*\h, 8*\t) --
                    (6*\h, 12*\t) --
                    (6*\h, 4*\t) --
                    (8*\h, 2*\t) --
                    (6*\h, 0*\t) --
                    cycle;
                };
                % particles
                for \i in {6,8,10}{
                    {
                        \filldraw (0*\h, \i*\t) circle (0.04);
                    };
                };
                for \i in {5,7,9}{
                    {
                        \filldraw (1*\h, \i*\t) circle (0.04);
                    };
                };
                for \i in {4,6,8}{
                    {
                        \filldraw (2*\h, \i*\t) circle (0.04);
                    };
                };
                for \i in {3,5,7,9}{
                    {
                        \filldraw (3*\h, \i*\t) circle (0.04);
                    };
                };
                for \i in {2,4,6,8,10}{
                    {
                        \filldraw (4*\h, \i*\t) circle (0.04);
                    };
                };
                for \i in {1,3,5,7,9,11}{
                    {
                        \filldraw (5*\h, \i*\t) circle (0.04);
                    };
                };
                for \i in {0,2,4,6,8,10,12}{
                    {
                        \filldraw (6*\h, \i*\t) circle (0.04);
                    };
                };
                for \i in {1,3}{
                    {
                        \filldraw (7*\h, \i*\t) circle (0.04);
                    };
                };
                for \i in {2}{
                    {
                        \filldraw (8*\h, \i*\t) circle (0.04);
                    };
                };
                {
                % bonds
                    % N - S
                    \draw (0*\h, 6*\t) -- (0*\h, 10*\t);
                    \draw[dotted] (1*\h, 5*\t) -- (1*\h, 9*\t);
                    \draw (2*\h, 4*\t) -- (2*\h, 8*\t);
                    \draw[dotted] (3*\h, 3*\t) -- (3*\h, 9*\t);
                    \draw (4*\h, 2*\t) -- (4*\h, 10*\t);
                    \draw[dotted] (5*\h, 1*\t) -- (5*\h, 11*\t);
                    \draw (6*\h, 0*\t) -- (6*\h, 12*\t);
                    \draw[dotted] (7*\h, 1*\t) -- (7*\h, 3*\t);
                    % SW - NE
                    \draw[dotted] (0*\h, 8*\t) -- (1*\h, 9*\t);
                    \draw (0*\h, 6*\t) -- (6*\h, 12*\t);
                    \draw[dotted] (1*\h, 5*\t) -- (6*\h, 10*\t);
                    \draw (2*\h, 4*\t) -- (6*\h, 8*\t);
                    \draw[dotted] (3*\h, 3*\t) -- (6*\h, 6*\t);
                    \draw (4*\h, 2*\t) -- (6*\h, 4*\t);
                    \draw[dotted] (5*\h, 1*\t) -- (7*\h, 3*\t);
                    \draw (6*\h, 0*\t) -- (8*\h, 2*\t);
                    % NW - SE
                    \draw[dotted] (5*\h, 11*\t) -- (6*\h, 10*\t);
                    \draw (4*\h, 10*\t) -- (6*\h, 8*\t);
                    \draw[dotted] (3*\h, 9*\t) -- (6*\h, 6*\t);
                    \draw (0*\h, 10*\t) -- (8*\h, 2*\t);
                    \draw[dotted] (0*\h, 8*\t) -- (7*\h, 1*\t);
                    \draw (0*\h, 6*\t) -- (6*\h, 0*\t);
                };
            }
        \end{tikzpicture}
        \caption{\centering}
        \label{fig:shape2}
    \end{subfigure}
    \hfill
    \begin{subfigure}[b]{0.19\textwidth}
        \centering
        \begin{tikzpicture}
            \tikzmath{
                \scale = .3*3/2;
                \t = \scale/2;
                \h = sqrt(\scale^2 - (\scale/2)^2);
                % shape
                {
                    \filldraw[yellow]
                    (0*\h, 2*\t) --
                    (0*\h, 10*\t) --
                    (4*\h, 6*\t) --
                    cycle;
                };
                % particles
                for \i in {2,6,10}{
                    {
                        \filldraw (0*\h, \i*\t) circle (0.04);
                    };
                };
                for \i in {4,8}{
                    {
                        \filldraw (2*\h, \i*\t) circle (0.04);
                    };
                };
                for \i in {6}{
                    {
                        \filldraw (4*\h, \i*\t) circle (0.04);
                    };
                };
                {
                % bonds
                    \draw (0*\h, 2*\t) -- (0*\h, 10*\t);
                    \draw (0*\h, 2*\t) -- (4*\h, 6*\t);
                    \draw (0*\h, 10*\t) -- (4*\h, 6*\t);
                    \draw (2*\h, 4*\t) -- (2*\h, 8*\t);
                    \draw (0*\h, 6*\t) -- (2*\h, 8*\t);
                    \draw (0*\h, 6*\t) -- (2*\h, 4*\t);
                };
            }
        \end{tikzpicture}
        \caption{\centering}
        \label{fig:shape3}
    \end{subfigure}
    \hfill
    \begin{subfigure}[b]{0.19\textwidth}
        \centering
        \begin{tikzpicture}
            \tikzmath{
                \scale = .3*3/2;
                \t = \scale/2;
                \h = sqrt(\scale^2 - (\scale/2)^2);
                % shape
                {
                    \filldraw[yellow]
                    (0*\h, 2*\t) --
                    (0*\h, 10*\t) --
                    (4*\h, 6*\t) --
                    cycle;
                };
                % particles
                for \i in {2,6,10}{
                    {
                        \filldraw (0*\h, \i*\t) circle (0.04);
                    };
                };
                for \i in {4,8}{
                    {
                        \filldraw (2*\h, \i*\t) circle (0.04);
                    };
                };
                for \i in {6}{
                    {
                        \filldraw (4*\h, \i*\t) circle (0.04);
                    };
                };
                {
                % bonds
                    \draw (0*\h, 2*\t) -- (0*\h, 10*\t);
                    \draw (0*\h, 2*\t) -- (4*\h, 6*\t);
                    \draw (0*\h, 10*\t) -- (4*\h, 6*\t);
                    \draw[dotted] (2*\h, 4*\t) -- (2*\h, 8*\t);
                    \draw[dotted] (0*\h, 6*\t) -- (2*\h, 8*\t);
                    \draw[dotted] (0*\h, 6*\t) -- (2*\h, 4*\t);
                };
            }
        \end{tikzpicture}
        \caption{\centering}
        \label{fig:shape4}
    \end{subfigure}
    \hfill
    \begin{subfigure}[b]{0.19\textwidth}
        \centering
        \begin{tikzpicture}
            \tikzmath{
                \scale = .3/7*6;
                \t = \scale/2;
                \h = sqrt(\scale^2 - (\scale/2)^2);
                \j = 0;
                \i = 0;
                {
                    % right
                    \filldraw[yellow] (\j*\h+0*\h, \i*\t+0*\t) -- (\j*\h+4*\h, \i*\t+4*\t) -- (\j*\h+4*\h, \i*\t-4*\t) -- cycle;
                    \filldraw (\j*\h+0*\h, \i*\t+0*\t) circle (0.04);
                    \filldraw (\j*\h+1*\h, \i*\t+1*\t) circle (0.04);
                    \filldraw (\j*\h+1*\h, \i*\t-1*\t) circle (0.04);
                    \filldraw (\j*\h+2*\h, \i*\t+2*\t) circle (0.04);
                    \filldraw (\j*\h+2*\h, \i*\t+0*\t) circle (0.04);
                    \filldraw (\j*\h+2*\h, \i*\t-2*\t) circle (0.04);
                    \filldraw (\j*\h+3*\h, \i*\t+3*\t) circle (0.04);
                    \filldraw (\j*\h+3*\h, \i*\t+1*\t) circle (0.04);
                    \filldraw (\j*\h+3*\h, \i*\t-1*\t) circle (0.04);
                    \filldraw (\j*\h+3*\h, \i*\t-3*\t) circle (0.04);
                    \filldraw (\j*\h+4*\h, \i*\t+4*\t) circle (0.04);
                    \filldraw (\j*\h+4*\h, \i*\t+2*\t) circle (0.04);
                    \filldraw (\j*\h+4*\h, \i*\t+0*\t) circle (0.04);
                    \filldraw (\j*\h+4*\h, \i*\t-2*\t) circle (0.04);
                    \filldraw (\j*\h+4*\h, \i*\t-4*\t) circle (0.04);
                    \draw (\j*\h+0*\h, \i*\t+0*\t) -- (\j*\h+4*\h, \i*\t+4*\t) -- (\j*\h+4*\h, \i*\t-4*\t) -- cycle;
                    \draw[dotted] (\j*\h+1*\h, \i*\t-1*\t) -- (\j*\h+1*\h, \i*\t+1*\t);
                    \draw[dotted] (\j*\h+2*\h, \i*\t-2*\t) -- (\j*\h+2*\h, \i*\t+2*\t);
                    \draw[dotted] (\j*\h+3*\h, \i*\t-3*\t) -- (\j*\h+3*\h, \i*\t+3*\t);
                    \draw[dotted] (\j*\h+4*\h, \i*\t-4*\t) -- (\j*\h+4*\h, \i*\t+4*\t);
                    \draw[dotted] (\j*\h+1*\h, \i*\t-1*\t) -- (\j*\h+4*\h, \i*\t+2*\t);
                    \draw[dotted] (\j*\h+2*\h, \i*\t-2*\t) -- (\j*\h+4*\h, \i*\t+0*\t);
                    \draw[dotted] (\j*\h+3*\h, \i*\t-3*\t) -- (\j*\h+4*\h, \i*\t-2*\t);
                    \draw[dotted] (\j*\h+1*\h, \i*\t+1*\t) -- (\j*\h+4*\h, \i*\t-2*\t);
                    \draw[dotted] (\j*\h+2*\h, \i*\t+2*\t) -- (\j*\h+4*\h, \i*\t-0*\t);
                    \draw[dotted] (\j*\h+3*\h, \i*\t+3*\t) -- (\j*\h+4*\h, \i*\t+2*\t);
                };
                \j = 7;
                \i = 7;
                {
                    % left
                    \filldraw[yellow] (\j*\h-0*\h, \i*\t+0*\t) -- (\j*\h-3*\h, \i*\t+3*\t) -- (\j*\h-3*\h, \i*\t-3*\t) -- cycle;
                    \filldraw (\j*\h-0*\h, \i*\t+0*\t) circle (0.04);
                    \filldraw (\j*\h-1*\h, \i*\t+1*\t) circle (0.04);
                    \filldraw (\j*\h-1*\h, \i*\t-1*\t) circle (0.04);
                    \filldraw (\j*\h-2*\h, \i*\t+2*\t) circle (0.04);
                    \filldraw (\j*\h-2*\h, \i*\t+0*\t) circle (0.04);
                    \filldraw (\j*\h-2*\h, \i*\t-2*\t) circle (0.04);
                    \filldraw (\j*\h-3*\h, \i*\t+3*\t) circle (0.04);
                    \filldraw (\j*\h-3*\h, \i*\t+1*\t) circle (0.04);
                    \filldraw (\j*\h-3*\h, \i*\t-1*\t) circle (0.04);
                    \filldraw (\j*\h-3*\h, \i*\t-3*\t) circle (0.04);
                    \draw (\j*\h-0*\h, \i*\t+0*\t) -- (\j*\h-3*\h, \i*\t+3*\t) -- (\j*\h-3*\h, \i*\t-3*\t) -- cycle;
                    \draw[dotted] (\j*\h-1*\h, \i*\t-1*\t) -- (\j*\h-1*\h, \i*\t+1*\t);
                    \draw[dotted] (\j*\h-2*\h, \i*\t-2*\t) -- (\j*\h-2*\h, \i*\t+2*\t);
                    \draw[dotted] (\j*\h-3*\h, \i*\t-3*\t) -- (\j*\h-3*\h, \i*\t+3*\t);
                    \draw[dotted] (\j*\h-1*\h, \i*\t-1*\t) -- (\j*\h-3*\h, \i*\t+1*\t);
                    \draw[dotted] (\j*\h-2*\h, \i*\t-2*\t) -- (\j*\h-3*\h, \i*\t-1*\t);
                    \draw[dotted] (\j*\h-1*\h, \i*\t+1*\t) -- (\j*\h-3*\h, \i*\t-1*\t);
                    \draw[dotted] (\j*\h-2*\h, \i*\t+2*\t) -- (\j*\h-3*\h, \i*\t+1*\t);
                };
            }
        \end{tikzpicture}
        \caption{\centering}
        \label{fig:shape5}
    \end{subfigure}
    \caption{
        Shapes and their representations.
        The triangles of the shapes are depicted in yellow.
        The solid edges indicate the edges of the triangles.
        The nodes indicate amoebots.
        The solid and dotted edges indicate bonds between amoebots.
        (a) and (b) show different representations of the same shape for scale $\sigma = 1$ and $\sigma = 2$, respectively.
        (c) and (d) show two shapes with identical representations.
        By definition, the shape in (d)~is minimal.
        (e) does not show a shape since the triangles are neither connected nor do they coincide with $G_{eqt}$.
        However, the second phase of our algorithm detects the boundary of both triangles,
        and the third phase computes a triangulation with different sized triangles for each face.
    }
    \label{fig:shape}
\end{figure}
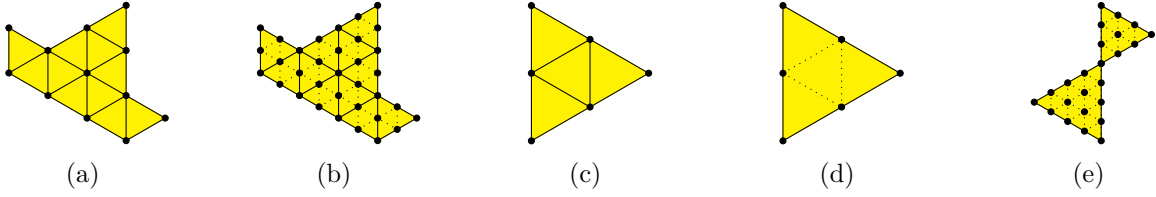

We assume that all amoebots share a common chirality.
Otherwise, we perform the chirality agreement algorithm from Section~\ref{subsec:chirality}.
The basic idea of our universal shape recognition algorithm is to partition the amoebot structure into triangles of the same size, and to compare the resulting triangulation to the shape $\mathcal S$.
The algorithm consists of four phases:
(i)~the identification of representations for $\sigma < 4$,
(ii)~the detection of the boundary,
(iii)~the triangulation of the shape, and
(iv)~the identification of representations for $\sigma \geq 4$.
We first give a short description of each phase.
The details are explained afterwards.
We defer the analysis for the correctness and runtime of our algorithm to Appendix~\ref{app:recognition}.

In the first phase, we directly compare the amoebot structure to the shape.
The representations of the shape with $\sigma < 4$ can be encoded within the finite state machines of the amoebots.
Each amoebot gathers information about the positions of all amoebots up to distance $3 \cdot |\mathcal S| + 1$.
This is enough to decide whether the amoebot structure is a representation of the shape with $\sigma < 4$.
If this is the case, the algorithm terminates with a positive answer.
Otherwise, we proceed to the next phase.

In the second phase, we detect the boundary of the shape.
An edge of the shape is a boundary edge iff exactly 1 of its 2 incident triangles belongs to the shape.
A corner of the shape is a boundary corner iff either 1, 2, 4, or 5 contiguous triangles of its 6 incident triangles belong to the shape.
On any representation of the shape with $\sigma \geq 4$, the phase detects the boundary of the shape.
On any representation of the shape with $\sigma < 4$, the phase would behave arbitrarily.
However, keep in mind that we already have eliminated those in the first phase.
On any other amoebot structure, the phase either fails to detect a boundary or detects a boundary unequal to the shape.
The algorithm terminates immediately with a negative answer if the phase fails to detect a boundary.
The boundaries that are unequal to the shape will finally result in a triangulation unequal to the shape.

In the third phase, we compute the triangles of the shape.
For this purpose, we apply a triangulation algorithm (see Algorithm~\ref{alg:triangulation} for pseudocode).
In order to avoid a linear runtime, we fix the number of iterations of our triangulation algorithm.
This can be done because the runtime of the triangulation algorithm does only depend on the shape and not on its representation.
The overall recognition algorithm terminates with the negative answer if the triangulation algorithm does not terminate on time or terminates prematurely.

In the fourth phase, we compare the resulting triangulation to the shape.
Note that each amoebot knows the shape.
Each triangle corner gathers information about the positions of all triangles up to distance $|\mathcal S| + 1$.
This is enough to decide whether the amoebot structure is a representation of the shape (with $\sigma \geq 4$).
If this is the case, the algorithm terminates with a positive answer.
Otherwise, it terminates with a negative answer.

We now discuss the phases in detail.
First, consider the boundary detection (second phase).
Consider the 2-neigh\-bor\-hoods within a representation of a shape with $\sigma \geq 4$.
We can categorize the occurring 2-neigh\-bor\-hoods into three disjoint classes:
(i)~boundary corners, (ii)~boundary edges, and (iii)~interior.
Figure~\ref{fig:neighborhoods} shows the 2-neigh\-bor\-hoods for the first class.
We refrain from stating the 2-neigh\-bor\-hoods of the remaining classes.
Note that each amoebot knows its 2-neigh\-bor\-hood from the first phase such that it is able to categorize itself.
The phase fails to detect a boundary if any amoebot cannot categorize itself into one of these classes, e.g., the 2-neigh\-bor\-hood contains an isolated boundary edge.
Otherwise, each amoebot can decide which of its incident edges belong to the boundary.
Note that $\sigma \geq 4$ is crucial since this guarantees that each triangle contains at least three inner nodes.
This allows us to safely identify boundary corners and boundary edges (see Figures~\ref{fig:sigma0} and~\ref{fig:sigma1}).

For arbitrary amoebot structures, the boundary may enclose several faces, which are connected by boundary corners (see Figure~\ref{fig:shape5}).
The third phase computes a triangulation for each face.
The scale of the triangles can fall below 4 and may even differ for each single face (see Figure~\ref{fig:shape5}).
However, we detect the disconnection in the fourth phase.

\begin{figure}[htb]
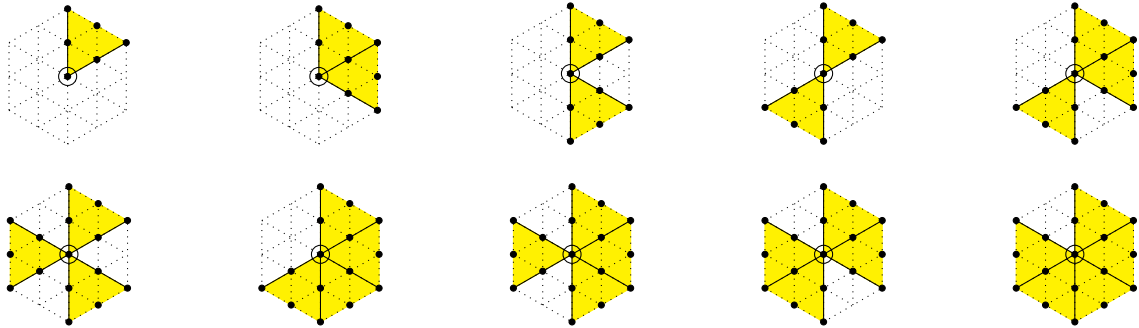

    \centering
    \begin{subfigure}[b]{0.19\textwidth}
        \centering

    \end{subfigure}
    \caption{
        All possible neighborhoods of amoebots representing a boundary corner except for symmetries.
        A specific shape may exclude some of these 2-neighborhoods.
        The dotted lines indicate the 2-neighborhood.
        Note that the triangles are only shown partially.
    }
    \label{fig:neighborhoods}
\end{figure}

\begin{figure}[htb]
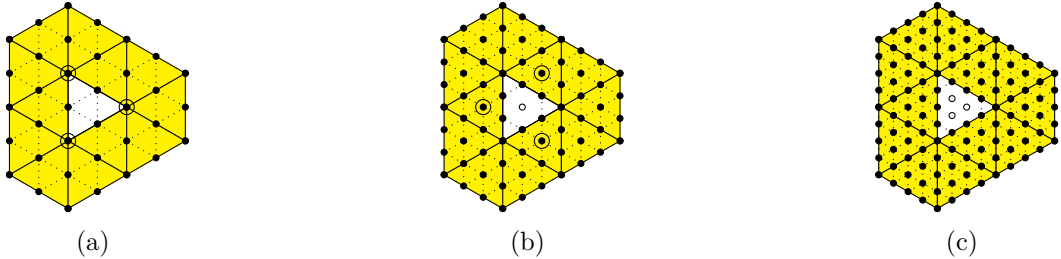

    \centering
    \begin{subfigure}[b]{0.30\columnwidth}
        \centering
        %
        \caption{\centering}
        \label{fig:sigma0}
    \end{subfigure}
    \hfill
    \begin{subfigure}[b]{0.30\columnwidth}
        \centering
        %
        \caption{\centering}
        \label{fig:sigma1}
    \end{subfigure}
    \hfill
    \begin{subfigure}[b]{0.30\columnwidth}
        \centering
        %
        \caption{\centering}
        \label{fig:sigma2}
    \end{subfigure}
    \caption{
        (a), (b) and (c) show three representations of the yellow shape with scale $\sigma = 2$, $\sigma = 3$ and $\sigma = 4$, respectively.
        The empty nodes indicate unoccupied nodes in the hole.
        Note that the additional bonds in the hole do not contradict the definition of a representation.
        The circled amoebots in (a) cannot detect themselves as boundary corners.
        The circled amoebots in (b) cannot detect themselves as part of the interior.
    }
\end{figure}

Next, consider the triangulation of the shape (see Algorithm~\ref{alg:triangulation} for pseudocode).
Figure~\ref{fig:triangulation_example} shows an exemplary execution.
The triangulation is based on the following observation:

\begin{observation}
\label{obs:incident_edges}
    Let a shape $\mathcal S$ and a representation $V(T(\mathcal S))$ be given.
    Consider any node $u \in V(T(\mathcal S))$ representing a corner of $\mathcal S$.
    For each node $v \in V(T(\mathcal S))$ adjacent to $u$, there is a line segment from $u$ over $v$ till the boundary consisting of the edges of the shape $\mathcal S$.
    This follows directly from the topology of $G_{eqt}$.
\end{observation}

\begin{algorithm}[ht]
\caption{Triangulation Protocol from a global perspective}
\label{alg:triangulation}
\begin{algorithmic}[1]
    \State Establish circuits along the axes
    \State $V \gets \emptyset$ \Comment{Set of identified corners}
    \State $N \gets \{ u \in S \mid u \text{ is a boundary corner} \}$ \Comment{Set of newly identified corners}
    \While{$N \neq \emptyset$}
        \State $V \gets V \cup N$
        %\State $N \gets \emptyset$
        \State Each amoebot $u \in V$ sends a beep on all circuits along the axes
        \State $N \gets \{ u \in S \setminus V \mid u \text{ has received a beep on at least 2 circuits along the axes} \}$
    \EndWhile
\end{algorithmic}
\end{algorithm}

\begin{figure}[htb]
    \centering
    \begin{subfigure}[b]{0.24\textwidth}
        \centering
        \begin{tikzpicture}
            \tikzmath{
                \scale = .3;
                \t = \scale/2;
                \h = sqrt(\scale^2 - (\scale/2)^2);
                % shape
                {
                    \filldraw[yellow]
                    (2*\h, 10*\t) --
                    (8*\h, 4*\t) --
                    (4*\h, 0*\t) --
                    (0*\h, 4*\t) --
                    (2*\h, 6*\t) --
                    (0*\h, 8*\t) --
                    cycle;
                };
                % particles
                for \i in {4,8}{
                    {
                        \filldraw (0*\h, \i*\t) circle (0.04);
                    };
                };
                for \i in {3,5,7,9}{
                    {
                        \filldraw (1*\h, \i*\t) circle (0.04);
                    };
                };
                for \i in {2,4,6,8,10}{
                    {
                        \filldraw (2*\h, \i*\t) circle (0.04);
                    };
                };
                for \i in {1,3,5,7,9}{
                    {
                        \filldraw (3*\h, \i*\t) circle (0.04);
                    };
                };
                for \i in {0,2,4,6,8}{
                    {
                        \filldraw (4*\h, \i*\t) circle (0.04);
                    };
                };
                for \i in {1,3,5,7}{
                    {
                        \filldraw (5*\h, \i*\t) circle (0.04);
                    };
                };
                for \i in {2,4,6}{
                    {
                        \filldraw (6*\h, \i*\t) circle (0.04);
                    };
                };
                for \i in {3,5}{
                    {
                        \filldraw (7*\h, \i*\t) circle (0.04);
                    };
                };
                for \i in {4}{
                    {
                        \filldraw (8*\h, \i*\t) circle (0.04);
                    };
                };
                {
                % bonds
                    % N - S
                    \draw[dotted] (1*\h, 3*\t) -- (1*\h, 5*\t);
                    \draw[dotted] (1*\h, 7*\t) -- (1*\h, 9*\t);
                    \draw[dotted] (2*\h, 2*\t) -- (2*\h, 10*\t);
                    \draw[dotted] (3*\h, 1*\t) -- (3*\h, 9*\t);
                    \draw[dotted] (4*\h, 0*\t) -- (4*\h, 8*\t);
                    \draw[dotted] (5*\h, 1*\t) -- (5*\h, 7*\t);
                    \draw[dotted] (6*\h, 2*\t) -- (6*\h, 6*\t);
                    \draw[dotted] (7*\h, 3*\t) -- (7*\h, 5*\t);
                    % SW - NE
                    \draw[dotted] (0*\h, 8*\t) -- (2*\h, 10*\t);
                    \draw[dotted] (1*\h, 7*\t) -- (3*\h, 9*\t);
                    \draw[dotted] (0*\h, 4*\t) -- (4*\h, 8*\t);
                    \draw[dotted] (1*\h, 3*\t) -- (5*\h, 7*\t);
                    \draw[dotted] (2*\h, 2*\t) -- (6*\h, 6*\t);
                    \draw[dotted] (3*\h, 1*\t) -- (7*\h, 5*\t);
                    \draw[dotted] (4*\h, 0*\t) -- (8*\h, 4*\t);
                    % NW - SE
                    \draw[dotted] (2*\h, 10*\t) -- (8*\h, 4*\t);
                    \draw[dotted] (1*\h, 9*\t) -- (7*\h, 3*\t);
                    \draw[dotted] (0*\h, 8*\t) -- (6*\h, 2*\t);
                    \draw[dotted] (1*\h, 5*\t) -- (5*\h, 1*\t);
                    \draw[dotted] (0*\h, 4*\t) -- (4*\h, 0*\t);
                % corners
                    \draw (2*\h, 10*\t) circle (0.1);
                    \draw (8*\h, 4*\t) circle (0.1);
                    \draw (4*\h, 0*\t) circle (0.1);
                    \draw (0*\h, 4*\t) circle (0.1);
                    \draw (2*\h, 6*\t) circle (0.1);
                    \draw (0*\h, 8*\t) circle (0.1);
                };
            }
        \end{tikzpicture}
        \caption*{\centering$i = 1$}
    \end{subfigure}
    \hfill
    \begin{subfigure}[b]{0.24\textwidth}
        \centering
        \begin{tikzpicture}
            \tikzmath{
                \scale = .3;
                \t = \scale/2;
                \h = sqrt(\scale^2 - (\scale/2)^2);
                % shape
                {
                    \filldraw[yellow]
                    (2*\h, 10*\t) --
                    (8*\h, 4*\t) --
                    (4*\h, 0*\t) --
                    (0*\h, 4*\t) --
                    (2*\h, 6*\t) --
                    (0*\h, 8*\t) --
                    cycle;
                };
                % particles
                for \i in {4,8}{
                    {
                        \filldraw (0*\h, \i*\t) circle (0.04);
                    };
                };
                for \i in {3,5,7,9}{
                    {
                        \filldraw (1*\h, \i*\t) circle (0.04);
                    };
                };
                for \i in {2,4,6,8,10}{
                    {
                        \filldraw (2*\h, \i*\t) circle (0.04);
                    };
                };
                for \i in {1,3,5,7,9}{
                    {
                        \filldraw (3*\h, \i*\t) circle (0.04);
                    };
                };
                for \i in {0,2,4,6,8}{
                    {
                        \filldraw (4*\h, \i*\t) circle (0.04);
                    };
                };
                for \i in {1,3,5,7}{
                    {
                        \filldraw (5*\h, \i*\t) circle (0.04);
                    };
                };
                for \i in {2,4,6}{
                    {
                        \filldraw (6*\h, \i*\t) circle (0.04);
                    };
                };
                for \i in {3,5}{
                    {
                        \filldraw (7*\h, \i*\t) circle (0.04);
                    };
                };
                for \i in {4}{
                    {
                        \filldraw (8*\h, \i*\t) circle (0.04);
                    };
                };
                {
                % bonds
                    % N - S
                    \draw[dotted] (1*\h, 3*\t) -- (1*\h, 5*\t);
                    \draw[dotted] (1*\h, 7*\t) -- (1*\h, 9*\t);
                    \draw (2*\h, 2*\t) -- (2*\h, 10*\t);
                    \draw[dotted] (3*\h, 1*\t) -- (3*\h, 9*\t);
                    \draw (4*\h, 0*\t) -- (4*\h, 8*\t);
                    \draw[dotted] (5*\h, 1*\t) -- (5*\h, 7*\t);
                    \draw[dotted] (6*\h, 2*\t) -- (6*\h, 6*\t);
                    \draw[dotted] (7*\h, 3*\t) -- (7*\h, 5*\t);
                    % SW - NE
                    \draw (0*\h, 8*\t) -- (2*\h, 10*\t);
                    \draw[dotted] (1*\h, 7*\t) -- (3*\h, 9*\t);
                    \draw (0*\h, 4*\t) -- (4*\h, 8*\t);
                    \draw[dotted] (1*\h, 3*\t) -- (5*\h, 7*\t);
                    \draw[dotted] (2*\h, 2*\t) -- (6*\h, 6*\t);
                    \draw[dotted] (3*\h, 1*\t) -- (7*\h, 5*\t);
                    \draw (4*\h, 0*\t) -- (8*\h, 4*\t);
                    % NW - SE
                    \draw (2*\h, 10*\t) -- (8*\h, 4*\t);
                    \draw[dotted] (1*\h, 9*\t) -- (7*\h, 3*\t);
                    \draw (0*\h, 8*\t) -- (6*\h, 2*\t);
                    \draw[dotted] (1*\h, 5*\t) -- (5*\h, 1*\t);
                    \draw (0*\h, 4*\t) -- (4*\h, 0*\t);
                % corners
                    \draw (2*\h, 2*\t) circle (0.1);
                    \draw (4*\h, 4*\t) circle (0.1);
                    \draw (4*\h, 8*\t) circle (0.1);
                };
            }
        \end{tikzpicture}
        \caption*{\centering$i = 2$}
    \end{subfigure}
    \hfill
    \begin{subfigure}[b]{0.24\textwidth}
        \centering
        \begin{tikzpicture}
            \tikzmath{
                \scale = .3;
                \t = \scale/2;
                \h = sqrt(\scale^2 - (\scale/2)^2);
                % shape
                {
                    \filldraw[yellow]
                    (2*\h, 10*\t) --
                    (8*\h, 4*\t) --
                    (4*\h, 0*\t) --
                    (0*\h, 4*\t) --
                    (2*\h, 6*\t) --
                    (0*\h, 8*\t) --
                    cycle;
                };
                % particles
                for \i in {4,8}{
                    {
                        \filldraw (0*\h, \i*\t) circle (0.04);
                    };
                };
                for \i in {3,5,7,9}{
                    {
                        \filldraw (1*\h, \i*\t) circle (0.04);
                    };
                };
                for \i in {2,4,6,8,10}{
                    {
                        \filldraw (2*\h, \i*\t) circle (0.04);
                    };
                };
                for \i in {1,3,5,7,9}{
                    {
                        \filldraw (3*\h, \i*\t) circle (0.04);
                    };
                };
                for \i in {0,2,4,6,8}{
                    {
                        \filldraw (4*\h, \i*\t) circle (0.04);
                    };
                };
                for \i in {1,3,5,7}{
                    {
                        \filldraw (5*\h, \i*\t) circle (0.04);
                    };
                };
                for \i in {2,4,6}{
                    {
                        \filldraw (6*\h, \i*\t) circle (0.04);
                    };
                };
                for \i in {3,5}{
                    {
                        \filldraw (7*\h, \i*\t) circle (0.04);
                    };
                };
                for \i in {4}{
                    {
                        \filldraw (8*\h, \i*\t) circle (0.04);
                    };
                };
                {
                % bonds
                    % N - S
                    \draw[dotted] (1*\h, 3*\t) -- (1*\h, 5*\t);
                    \draw[dotted] (1*\h, 7*\t) -- (1*\h, 9*\t);
                    \draw (2*\h, 2*\t) -- (2*\h, 10*\t);
                    \draw[dotted] (3*\h, 1*\t) -- (3*\h, 9*\t);
                    \draw (4*\h, 0*\t) -- (4*\h, 8*\t);
                    \draw[dotted] (5*\h, 1*\t) -- (5*\h, 7*\t);
                    \draw[dotted] (6*\h, 2*\t) -- (6*\h, 6*\t);
                    \draw[dotted] (7*\h, 3*\t) -- (7*\h, 5*\t);
                    % SW - NE
                    \draw (0*\h, 8*\t) -- (2*\h, 10*\t);
                    \draw[dotted] (1*\h, 7*\t) -- (3*\h, 9*\t);
                    \draw (0*\h, 4*\t) -- (4*\h, 8*\t);
                    \draw[dotted] (1*\h, 3*\t) -- (5*\h, 7*\t);
                    \draw (2*\h, 2*\t) -- (6*\h, 6*\t);
                    \draw[dotted] (3*\h, 1*\t) -- (7*\h, 5*\t);
                    \draw (4*\h, 0*\t) -- (8*\h, 4*\t);
                    % NW - SE
                    \draw (2*\h, 10*\t) -- (8*\h, 4*\t);
                    \draw[dotted] (1*\h, 9*\t) -- (7*\h, 3*\t);
                    \draw (0*\h, 8*\t) -- (6*\h, 2*\t);
                    \draw[dotted] (1*\h, 5*\t) -- (5*\h, 1*\t);
                    \draw (0*\h, 4*\t) -- (4*\h, 0*\t);
                % corners
                    \draw (6*\h, 6*\t) circle (0.1);
                };
            }
        \end{tikzpicture}
        \caption*{\centering$i = 3$}
    \end{subfigure}
    \hfill
    \begin{subfigure}[b]{0.24\textwidth}
        \centering
        \begin{tikzpicture}
            \tikzmath{
                \scale = .3;
                \t = \scale/2;
                \h = sqrt(\scale^2 - (\scale/2)^2);
                % shape
                {
                    \filldraw[yellow]
                    (2*\h, 10*\t) --
                    (8*\h, 4*\t) --
                    (4*\h, 0*\t) --
                    (0*\h, 4*\t) --
                    (2*\h, 6*\t) --
                    (0*\h, 8*\t) --
                    cycle;
                };
                % particles
                for \i in {4,8}{
                    {
                        \filldraw (0*\h, \i*\t) circle (0.04);
                    };
                };
                for \i in {3,5,7,9}{
                    {
                        \filldraw (1*\h, \i*\t) circle (0.04);
                    };
                };
                for \i in {2,4,6,8,10}{
                    {
                        \filldraw (2*\h, \i*\t) circle (0.04);
                    };
                };
                for \i in {1,3,5,7,9}{
                    {
                        \filldraw (3*\h, \i*\t) circle (0.04);
                    };
                };
                for \i in {0,2,4,6,8}{
                    {
                        \filldraw (4*\h, \i*\t) circle (0.04);
                    };
                };
                for \i in {1,3,5,7}{
                    {
                        \filldraw (5*\h, \i*\t) circle (0.04);
                    };
                };
                for \i in {2,4,6}{
                    {
                        \filldraw (6*\h, \i*\t) circle (0.04);
                    };
                };
                for \i in {3,5}{
                    {
                        \filldraw (7*\h, \i*\t) circle (0.04);
                    };
                };
                for \i in {4}{
                    {
                        \filldraw (8*\h, \i*\t) circle (0.04);
                    };
                };
                {
                % bonds
                    % N - S
                    \draw[dotted] (1*\h, 3*\t) -- (1*\h, 5*\t);
                    \draw[dotted] (1*\h, 7*\t) -- (1*\h, 9*\t);
                    \draw (2*\h, 2*\t) -- (2*\h, 10*\t);
                    \draw[dotted] (3*\h, 1*\t) -- (3*\h, 9*\t);
                    \draw (4*\h, 0*\t) -- (4*\h, 8*\t);
                    \draw[dotted] (5*\h, 1*\t) -- (5*\h, 7*\t);
                    \draw (6*\h, 2*\t) -- (6*\h, 6*\t);
                    \draw[dotted] (7*\h, 3*\t) -- (7*\h, 5*\t);
                    % SW - NE
                    \draw (0*\h, 8*\t) -- (2*\h, 10*\t);
                    \draw[dotted] (1*\h, 7*\t) -- (3*\h, 9*\t);
                    \draw (0*\h, 4*\t) -- (4*\h, 8*\t);
                    \draw[dotted] (1*\h, 3*\t) -- (5*\h, 7*\t);
                    \draw (2*\h, 2*\t) -- (6*\h, 6*\t);
                    \draw[dotted] (3*\h, 1*\t) -- (7*\h, 5*\t);
                    \draw (4*\h, 0*\t) -- (8*\h, 4*\t);
                    % NW - SE
                    \draw (2*\h, 10*\t) -- (8*\h, 4*\t);
                    \draw[dotted] (1*\h, 9*\t) -- (7*\h, 3*\t);
                    \draw (0*\h, 8*\t) -- (6*\h, 2*\t);
                    \draw[dotted] (1*\h, 5*\t) -- (5*\h, 1*\t);
                    \draw (0*\h, 4*\t) -- (4*\h, 0*\t);
                };
            }
        \end{tikzpicture}
        \caption*{\centering$i = 4$}
    \end{subfigure}
    \caption{
        An exemplary execution of Algorithm \ref{alg:triangulation} is shown.
        For the sake of simplicity, we use an example with $\sigma = 2$.
        The second phase fails to detect the boundary.
        So, assume that the boundary is given.
        Each figure shows the situation at the beginning of the $i$-th iteration.
        %The shape is depicted in yellow.
        %The representation of the shape are given by the nodes.
        The set of newly identified corners $N$ is circled.
        The edges indicate the activated circuits.
    }
    \label{fig:triangulation_example}
\end{figure}
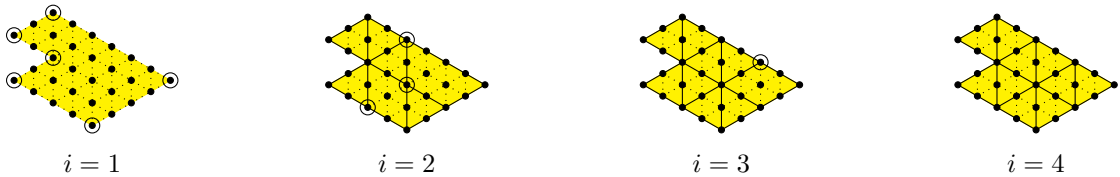

Algorithm~\ref{alg:triangulation} progresses iteratively.
In each iteration, we add edges originating from newly identified corners according to Observation~\ref{obs:incident_edges}.
New corners are determined by the intersections of these edges.
The first set of newly identified corners is given by the first class of the second phase. In order to compute the edges of the triangles, we establish circuits along the axes (see Figure~\ref{fig:triangulation:circuits1}).
However, we remove connections if they reach the boundary.
The algorithm results in a partition of the polygon into triangles.
The vertices of the triangles are given by $V$.
The edges of the triangles are given by the activated circuits.

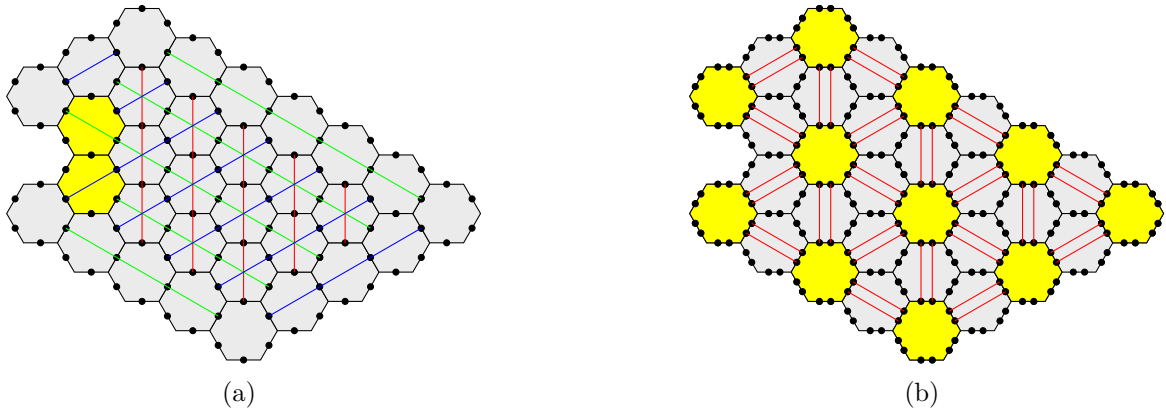
\begin{figure}[htb]
    \centering
    \begin{subfigure}[b]{0.45\textwidth}
        \centering
        \begin{tikzpicture}
        \tikzmath{
            \scale = .45;
            \pins = 1;
            % amoebots
            \fr = 235;
            \fg = 235;
            \fb = 235;
            for \j in {0,1}{
                for \i in {-1,-2,-3,-4,-5}{
                    Amoebot(\j,0,\i);
                };
            };
            for \j in {2,3,4}{
                for \i in {-1,-2,-3,-4,-5,-6,-7}{
                    Amoebot(\j,0,\i);
                };
            };
            \fr = 255;
            \fg = 255;
            \fb = 0;
            Amoebot(1,0,-5);
            Amoebot(2,0,-6);
            % connections
            for \i in {-2,-3,-4}{
                \lr = 0;
                \lg = 255;
                \lb = 0;
                Circuit(0,0,\i,2,1,5,1);
            };
            for \i in {-2,-3,-4}{
                \lr = 255;
                \lg = 0;
                \lb = 0;
                Circuit(1,0,\i,1,1,4,1);
                \lr = 0;
                \lg = 255;
                \lb = 0;
                Circuit(1,0,\i,2,1,5,1);
                \lr = 0;
                \lg = 0;
                \lb = 255;
                Circuit(1,0,\i,3,1,6,1);
            };
            for \i in {-2,-3,-4,-5}{
                \lr = 255;
                \lg = 0;
                \lb = 0;
                Circuit(2,0,\i,1,1,4,1);
                \lr = 0;
                \lg = 255;
                \lb = 0;
                Circuit(2,0,\i,2,1,5,1);
                \lr = 0;
                \lg = 0;
                \lb = 255;
                Circuit(2,0,\i,3,1,6,1);
            };
            for \i in {-2,-3,-4,-5,-6}{
                \lr = 255;
                \lg = 0;
                \lb = 0;
                Circuit(3,0,\i,1,1,4,1);
                \lr = 0;
                \lg = 255;
                \lb = 0;
                Circuit(3,0,\i,2,1,5,1);
                \lr = 0;
                \lg = 0;
                \lb = 255;
                Circuit(3,0,\i,3,1,6,1);
            };
            for \i in {-2,-3,-4,-5,-6}{
                \lr = 0;
                \lg = 255;
                \lb = 0;
                Circuit(4,0,\i,2,1,5,1);
            };
            for \i in {1,2,3}{
                \lr = 0;
                \lg = 0;
                \lb = 255;
                Circuit(\i,0,-1,3,1,6,1);
            };
            \lr = 0;
            \lg = 0;
            \lb = 255;
            Circuit(1,0,-5,3,1,6,1);
            Circuit(3,0,-7,3,1,6,1);
            \lr = 0;
            \lg = 255;
            \lb = 0;
            Circuit(2,0,-6,2,1,5,1);
        }
        \end{tikzpicture}
        \caption{\centering}
        \label{fig:triangulation:circuits1}
    \end{subfigure}
    \hfill
    \begin{subfigure}[b]{0.45\textwidth}
        \centering
        \begin{tikzpicture}
        \tikzmath{
            \scale = .45;
            \pins = 2;
            % amoebots
            \fr = 235;
            \fg = 235;
            \fb = 235;
            for \j in {0,1}{
                for \i in {-1,-2,-3,-4,-5}{
                    Amoebot(\j,0,\i);
                };
            };
            for \j in {2,3,4}{
                for \i in {-1,-2,-3,-4,-5,-6,-7}{
                    Amoebot(\j,0,\i);
                };
            };
            \fr = 255;
            \fg = 255;
            \fb = 0;
            for \j in {0}{
                for \i in {-1,-3,-5}{
                    Amoebot(\j,0,\i);
                };
            };
            for \j in {2,4}{
                for \i in {-1,-3,-5,-7}{
                    Amoebot(\j,0,\i);
                };
            };
            % connections
            \lr = 255;
            \lg = 0;
            \lb = 0;
            for \j in {0}{
                for \i in {-2,-4}{
                    Circuit(\j,0,\i,2,1,5,1);
                    Circuit(\j,0,\i,2,2,5,2);
                };
            };
            for \j in {2,4}{
                for \i in {-2,-4,-6}{
                    Circuit(\j,0,\i,2,1,5,1);
                    Circuit(\j,0,\i,2,2,5,2);
                };
            };
            for \j in {1}{
                for \i in {-2,-4}{
                    Circuit(\j,0,\i,1,1,4,1);
                    Circuit(\j,0,\i,1,2,4,2);
                };
            };
            for \j in {3}{
                for \i in {-2,-4,-6}{
                    Circuit(\j,0,\i,1,1,4,1);
                    Circuit(\j,0,\i,1,2,4,2);
                };
            };
            for \j in {1}{
                for \i in {-1,-3,-5}{
                    Circuit(\j,0,\i,3,1,6,1);
                    Circuit(\j,0,\i,3,2,6,2);
                };
            };
            for \j in {3}{
                for \i in {-1,-3,-5,-7}{
                    Circuit(\j,0,\i,3,1,6,1);
                    Circuit(\j,0,\i,3,2,6,2);
                };
            };
        }
        \end{tikzpicture}
        \caption{\centering}
        \label{fig:triangulation:circuits2}
    \end{subfigure}
    \caption{
        %Utilized circuits.
        (a) shows the circuits along the axes.
        Thereby, each amoebot connects (i)~its northern pin to its southern pin, (ii)~its north-eastern pin to its south-western pin, and (iii)~its north-western pin to its south-eastern pin.
        No common orientation is required due to the symmetry of the connections.
        The amoebots on the boundary have removed all connections that reach the boundary.
        In particular, each yellow amoebot has removed the connection from its northern pin to its southern pin.
        Note that this works with an arbitrary number of pins (see Remark~\ref{rem:onepin}).
        (b) shows the circuits for each side of a triangle.
        The corners are depicted in yellow.
    }
\end{figure}

It remains to show how the information gathering in the first and fourth phase works.
We utilize the primitive for message transmission between neighboring amoebots in the first phase and between the endpoints of each triangle side in the fourth phase (see Section~\ref{sec:preliminaries}).
In order to apply the primitive for the latter, we have to connect the first pin of one endpoint to the last pin of the other endpoint and vice versa (see Figure~\ref{fig:triangulation:circuits2}).
The information is gathered iteratively.
The messages of the $i$-th iteration contain all known positions of amoebots or known triangles with respect to the transmission direction and up to distance $i$ (see Figure~\ref{fig:message}).
The positions up to distance $i + 1$ are known afterwards.
Recall that the required distance is constant for both phases.
This bounds the number of iterations and the size of the messages.
Finally, note that we cannot assume that the triangles are connected in the fourth phase.
We therefore gather the information for each face separately.
In this way, boundary corners are able to detect whether there is more than one face (see Appendix~\ref{app:recognition}).

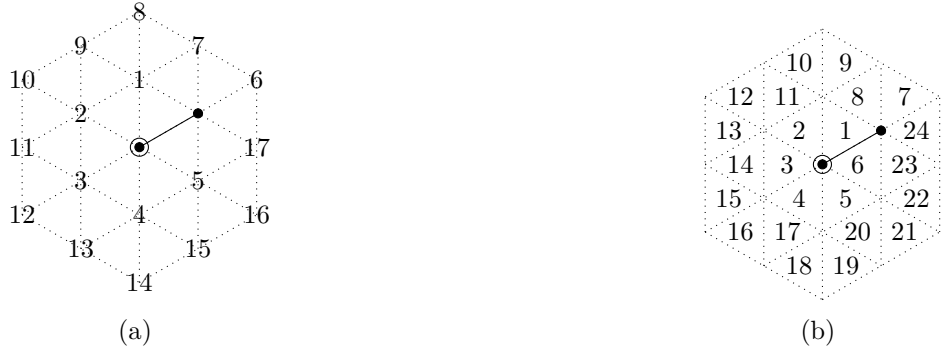
\begin{figure}[htb]
    \centering
    \begin{subfigure}[b]{0.45\columnwidth}
        \centering
        \begin{tikzpicture}
            \tikzmath{
                \scale = .9;
                \t = \scale/2;
                \h = sqrt(\scale^2 - (\scale/2)^2);
                {
                % amoebots
                    \filldraw (2*\h, 4*\t) circle (0.06);
                    \draw (2*\h, 4*\t) circle (0.12);
                    \filldraw (3*\h, 5*\t) circle (0.06);
                % grid
                    % N - S
                    \draw[dotted] (0*\h, 2*\t) -- (0*\h, 6*\t);
                    \draw[dotted] (1*\h, 1*\t) -- (1*\h, 7*\t);
                    \draw[dotted] (2*\h, 0*\t) -- (2*\h, 8*\t);
                    \draw[dotted] (3*\h, 1*\t) -- (3*\h, 7*\t);
                    \draw[dotted] (4*\h, 2*\t) -- (4*\h, 6*\t);
                    % SW - NE
                    \draw[dotted] (0*\h, 6*\t) -- (2*\h, 8*\t);
                    \draw[dotted] (0*\h, 4*\t) -- (3*\h, 7*\t);
                    \draw[dotted] (0*\h, 2*\t) -- (4*\h, 6*\t);
                    \draw[dotted] (1*\h, 1*\t) -- (4*\h, 4*\t);
                    \draw[dotted] (2*\h, 0*\t) -- (4*\h, 2*\t);
                    % NW - SE
                    \draw[dotted] (2*\h, 8*\t) -- (4*\h, 6*\t);
                    \draw[dotted] (1*\h, 7*\t) -- (4*\h, 4*\t);
                    \draw[dotted] (0*\h, 6*\t) -- (4*\h, 2*\t);
                    \draw[dotted] (0*\h, 4*\t) -- (3*\h, 1*\t);
                    \draw[dotted] (0*\h, 2*\t) -- (2*\h, 0*\t);
                % bond
                    \draw (2*\h, 4*\t) -- (3*\h, 5*\t);
                % face ids
                    \draw[black] (0*\h, 2*\t) node {\small $12$};
                    \draw[black] (0*\h, 4*\t) node {\small $11$};
                    \draw[black] (0*\h, 6*\t) node {\small $10$};
                    \draw[black] (1*\h, 1*\t) node {\small $13$};
                    \draw[black] (1*\h, 3*\t) node {\small $3$};
                    \draw[black] (1*\h, 5*\t) node {\small $2$};
                    \draw[black] (1*\h, 7*\t) node {\small $9$};
                    \draw[black] (2*\h, 0*\t) node {\small $14$};
                    \draw[black] (2*\h, 2*\t) node {\small $4$};
                    \draw[black] (2*\h, 6*\t) node {\small $1$};
                    \draw[black] (2*\h, 8*\t) node {\small $8$};
                    \draw[black] (3*\h, 1*\t) node {\small $15$};
                    \draw[black] (3*\h, 3*\t) node {\small $5$};
                    \draw[black] (3*\h, 7*\t) node {\small $7$};
                    \draw[black] (4*\h, 2*\t) node {\small $16$};
                    \draw[black] (4*\h, 4*\t) node {\small $17$};
                    \draw[black] (4*\h, 6*\t) node {\small $6$};
                };
            }
        \end{tikzpicture}
        \caption{\centering}
        \label{fig:message1}
    \end{subfigure}
    \hfill
    \begin{subfigure}[b]{0.45\columnwidth}
        \centering
        \begin{tikzpicture}
            \tikzmath{
                \scale = .9;
                \t = \scale/2;
                \h = sqrt(\scale^2 - (\scale/2)^2);
                {
                % amoebots
                    \filldraw (2*\h, 4*\t) circle (0.06);
                    \draw (2*\h, 4*\t) circle (0.12);
                    \filldraw (3*\h, 5*\t) circle (0.06);
                % grid
                    % N - S
                    \draw[dotted] (0*\h, 2*\t) -- (0*\h, 6*\t);
                    \draw[dotted] (1*\h, 1*\t) -- (1*\h, 7*\t);
                    \draw[dotted] (2*\h, 0*\t) -- (2*\h, 8*\t);
                    \draw[dotted] (3*\h, 1*\t) -- (3*\h, 7*\t);
                    \draw[dotted] (4*\h, 2*\t) -- (4*\h, 6*\t);
                    % SW - NE
                    \draw[dotted] (0*\h, 6*\t) -- (2*\h, 8*\t);
                    \draw[dotted] (0*\h, 4*\t) -- (3*\h, 7*\t);
                    \draw[dotted] (0*\h, 2*\t) -- (4*\h, 6*\t);
                    \draw[dotted] (1*\h, 1*\t) -- (4*\h, 4*\t);
                    \draw[dotted] (2*\h, 0*\t) -- (4*\h, 2*\t);
                    % NW - SE
                    \draw[dotted] (2*\h, 8*\t) -- (4*\h, 6*\t);
                    \draw[dotted] (1*\h, 7*\t) -- (4*\h, 4*\t);
                    \draw[dotted] (0*\h, 6*\t) -- (4*\h, 2*\t);
                    \draw[dotted] (0*\h, 4*\t) -- (3*\h, 1*\t);
                    \draw[dotted] (0*\h, 2*\t) -- (2*\h, 0*\t);
                % bond
                    \draw (2*\h, 4*\t) -- (3*\h, 5*\t);
                % face ids
                    \draw[black] (0.6*\h, 2*\t) node {\small $16$};
                    \draw[black] (0.4*\h, 3*\t) node {\small $15$};
                    \draw[black] (0.6*\h, 4*\t) node {\small $14$};
                    \draw[black] (0.4*\h, 5*\t) node {\small $13$};
                    \draw[black] (0.6*\h, 6*\t) node {\small $12$};
                    \draw[black] (1.6*\h, 1*\t) node {\small $18$};
                    \draw[black] (1.4*\h, 2*\t) node {\small $17$};
                    \draw[black] (1.6*\h, 3*\t) node {\small $4$};
                    \draw[black] (1.4*\h, 4*\t) node {\small $3$};
                    \draw[black] (1.6*\h, 5*\t) node {\small $2$};
                    \draw[black] (1.4*\h, 6*\t) node {\small $11$};
                    \draw[black] (1.6*\h, 7*\t) node {\small $10$};
                    \draw[black] (2.4*\h, 1*\t) node {\small $19$};
                    \draw[black] (2.6*\h, 2*\t) node {\small $20$};
                    \draw[black] (2.4*\h, 3*\t) node {\small $5$};
                    \draw[black] (2.6*\h, 4*\t) node {\small $6$};
                    \draw[black] (2.4*\h, 5*\t) node {\small $1$};
                    \draw[black] (2.6*\h, 6*\t) node {\small $8$};
                    \draw[black] (2.4*\h, 7*\t) node {\small $9$};
                    \draw[black] (3.4*\h, 2*\t) node {\small $21$};
                    \draw[black] (3.6*\h, 3*\t) node {\small $22$};
                    \draw[black] (3.4*\h, 4*\t) node {\small $23$};
                    \draw[black] (3.6*\h, 5*\t) node {\small $24$};
                    \draw[black] (3.4*\h, 6*\t) node {\small $7$};
                };
            }
        \end{tikzpicture}
        \caption{\centering}
        \label{fig:message2}
    \end{subfigure}
    \caption{
        (a) shows an exemplary message for the first phase.
        The nodes indicate two neighboring amoebots $u$ and $v$ where $u$ is the circled one.
        Amoebot $u$ assigns a truth value to each possible amoebot with respect to edge $\{u,v\}$ and sends them to $v$.
        (b) shows an exemplary message for the fourth phase.
        The nodes indicate the the endpoint of a side $u$ and $v$ where the $u$ is the circled one.
        Amoebot $u$ assigns a truth value to each possible triangle with respect to edge $\{u,v\}$ and sends them to $v$.
    }
    \label{fig:message}
\end{figure}

Altogether, we arrive at the following theorem:

\begin{theorem}
\label{th:usr}
    Let shape $\mathcal S$ be connected and minimal.
    The amoebot structure can detect a representation of the shape $V(T(\mathcal S))$.
    It requires $O(1)$ rounds if the amoebots share a common chirality.
    Otherwise, it requires $O(\log n)$ rounds, w.h.p.
\end{theorem}

%%
%% Bibliography
\printbibliography

\appendix

\section{Message Transmission}
\label{app:preliminaries}

We consider the message transmission primitive for the special case of a single pin per bond.
The single pin has to be used for both sending and receiving messages.
However, if both amoebots decided to send a message at the same time, the messages would interfere with each other.
Hence, we first define a priority by a local leader election (see Section~\ref{sec:preliminaries}).
Each message transmission is initiated by a beep.
The first subsequent round is used to determine the sender and receiver.
The prioritized amoebot, i.e., the local leader, beeps if it has initiated the transmission.
Consequently, the other amoebot cancels its transmission if it has initiated one at the same time.
If the prioritized amoebot does not beep, then the other amoebot is free to send its message.
To keep the transmission rates of both amoebots fair, we may transfer the priority to the receiving amoebot after each transmission.
Thus, the transmission of constant-sized messages between two amoebots takes $O(1)$ rounds after initial $\Theta(\log n)$ rounds, w.h.p.

\section{Leader Election}
\label{app:leader}

\algdef{SE}[DOWHILE]{Do}{doWhile}{\algorithmicdo}[1]{\algorithmicwhile\ #1}%
\begin{algorithm}[ht]
\caption{Leader Election Protocol from a global perspective}
\label{alg:leader_election}
\begin{algorithmic}[1]
\State $C_1 \gets S$ \Comment{Initialize set of candidates}
\Do \Comment{Phase $1$ tournament}
	\State Each $u \in C_1$ tosses a coin and either sets $u.c_1 \gets \mathit{HEADS}$ or $u.c_1 \gets \mathit{TAILS}$
	\State $H_1 \gets \{u \in C_1 \mid u.c_1 = \mathit{HEADS}\}$ 
	\State $T_1 \gets C_1 \setminus H_1$
    \If{$H_1 \neq \emptyset$}
    	\State $C_1 \gets C_1 \setminus T_1$
    \EndIf
\doWhile{$H_1 \neq \emptyset \wedge T_1 \neq \emptyset$}
\State
\State $i \gets 0$
\While{$i < \kappa$} \Comment{Phase $2$ tournaments}
	\State $C_2 \gets S$
	\Do
		\State Each $u \in C_1$ tosses a coin and either sets $u.c_1 \gets \mathit{HEADS}$ or $u.c_1 \gets \mathit{TAILS}$
		\State Each $u \in C_2$ tosses a coin and either sets $u.c_2 \gets \mathit{HEADS}$ or $u.c_2 \gets \mathit{TAILS}$
		\State $H_1 \gets \{u \in C_1 \mid u.c_1 = \mathit{HEADS}\}$
		\State $H_2 \gets \{u \in C_2 \mid u.c_2 = \mathit{HEADS}\}$
		\State $T_1 \gets C_1 \setminus H_1$
		\State $T_2 \gets C_2 \setminus H_2$
 	    \If{$H_1 \neq \emptyset$}
 	   		\State $C_1 \gets C_1 \setminus T_1$
 	    \EndIf
  	    \If{$H_2 \neq \emptyset$}
    		\State $C_2 \gets C_2 \setminus T_2$
        \EndIf
	\doWhile{$H_2 \neq \emptyset \wedge T_2 \neq \emptyset$}
	\State $i \gets i + 1$
\EndWhile
\State
\State \Return $C_1$ \Comment{W.h.p., $|C_1| = 1$}
\end{algorithmic}
\end{algorithm}

We give a detailed analysis on the leader election algorithm described in Section~\ref{sec:leader}.
First, we bound the runtime and the number of remaining canidates after the first phase of the protocol has finished.

\begin{lemma} \label{lem:leader_election:t1}
	The first phase of the protocol finishes after $\Theta(\log n)$ rounds, w.h.p., with $|C_1| = O(\log n)$.
\end{lemma}

\begin{proof}
	We first show that in each iteration of the first phase, the number of remaining candidates in $C_1$ gets reduced by some constant factor $k$, w.h.p., $\frac{1}{2} < k < 1$, as long as $|C_1| \geq \frac{6}{(2k-1)^2} c \log n = O(\log n)$, $c > 0$ constant.
	Consider one iteration of the tournament in which $|C_1| \geq \frac{6}{(2k-1)^2} c \log n$.
	For each candidate $u \in C_1$ define 	
	\begin{equation*}
   		X_u =
   		\begin{cases}
     		 1, & \text{if }u.c_1 = \mathit{HEADS} \\
    		 0, & \text{if }u.c_1 = \mathit{TAILS}.
   		\end{cases}
	\end{equation*}
	We have that $\Pr[X_u = 1] = \frac{1}{2}$.
	Define $X = \sum_{u \in C_1} X_u$.
	Then, $\E[X] = \frac{1}{2} |C_1|$.
	Using Chernoff bounds, we obtain
	\begin{eqnarray*}
		\Pr[X \geq k \cdot |C_1|] & = & \Pr \left[X \geq (1+(2k-1))\cdot \frac{1}{2} |C_1| \right]\\
		& \leq & \exp \left(-\frac{(2k-1)^2 \cdot \frac{1}{2}|C_1|}{3} \right)\\
		& = & \exp \left(-\frac{(2k-1)^2 \cdot |C_1|}{6} \right)\\
		& \leq & \exp \left(-c \log n\right)\\
		& < & n^{-c}.
	\end{eqnarray*}
	Therefore, w.h.p., at most $k \cdot |C_1|$ candidates remain after one single iteration.
	
	Analogously, one can show that at least a constant fraction of candidates from $C_1$ remains candidates in the next iteration. 
	
	Now, we can compute the maximum number of iterations $i$ needed until $|C_1|$ gets reduced from $n$ to $\frac{6}{(2k-1)^2} c \log n = O(\log n)$.
	For this, we have to solve $k^i \cdot |C_1| = \frac{6}{(2k-1)^2} c \log n$.
	Trivially, $i = O(\log n)$.
%	We have
%	\begin{align*}
%		&        &  k^i \cdot n &= \frac{6}{(2k-1)^2} \cdot c \log n        		  \\
%		\Leftrightarrow  &       &  i \cdot \log k + \log n&= \log \left(\frac{6}{(2k-1)^2}\right)  + \log c + \log \log n 
%	\end{align*}
%	which leads to $i \in O(\log n)$.
%	and therefore, we obtain $$i = \frac{\log 6 - 2\log(2k-1)  + \log c + \log \log n - \log n}{\log k}.$$
%	Observe that the denominator is negative for $k < 1$ constant.
%	Together with the fact that $c > 0$ is a constant, we obtain $i = O(\log n)$.
	By similar arguments, it follows that the minimum number of rounds $j$ needed to reach $|C_1| = \frac{6}{(2k-1)^2} c \log n$ is $j = \Omega(\log n)$.
	Thus, after $\Theta(\log n)$ rounds, we have reduced $|C_1|$ from $n$ to $O(\log n)$.
	
	We now only need to wait another $O(\log n)$ rounds to reach an iteration in which either $\forall u \in C_1: u.c_1 = \mathit{HEADS}$ or $\forall u \in C_1: u.c_1 = \mathit{TAILS}$, as in each iteration where the above condition does not hold, $|C_1|$ gets reduced by at least $1$ candidate.
	Thus, after $O(\log n)$ of such iterations, only a constant amount of candidates remain in $C_1$ and the first phase terminates with at least a constant probability from this point forward.
	This concludes the proof.
\end{proof}

Next, we also bound the runtime and the number of remaining candidates after the second phase of the protocol has finished.

\begin{lemma} \label{lem:leader_election:t2}
	The second phase of the protocol finishes after $\Theta(\log n)$ rounds, w.h.p., with $|C_1| = 1$, w.h.p.
\end{lemma}

\begin{proof}
	Since the tournament on $C_2$ starts with $n$ amoebots, we immediately get that the second phase finishes after $\Theta(\log n)$ rounds, w.h.p. due to Lemma~\ref{lem:leader_election:t1}. 
	
	It therefore remains to show that the tournament on $C_1$ reduces $|C_1|$ from $O(\log n)$ candidates to one single candidate after $\Theta(\log n)$ rounds, w.h.p.
	From the proof of Lemma~\ref{lem:leader_election:t1} it follows that the first phase runs for at least $\log n - \gamma \log \log n$ rounds, for an appropriately chosen constant $\gamma > 0$.
	Assume that after the first phase there are $c \log n + 1$ candidates left in the set $C_1$.
	The second phase consists of $\kappa \cdot (\log n - \gamma \log \log n)$ iterations, $\kappa \geq 3$ constant.
	For a candidate $u \in V$, let $B(v,i) = 1$ if $u.c_1 = HEADS$ in the $i$-th iteration of the second phase and $B(v,i) = 0$ if $u.c_1 = TAILS$ in the $i$-th iteration of the second phase.
	Fix the amoebot $v \in V$ that is going to be elected as the leader.
	We show for every amoebot $u \neq v$ that the probability that $B(u,i) = B(v,i)$ holds in each of the $\kappa \cdot (\log n - \gamma \log \log n)$ iterations is negligible.
	We get
	\begin{eqnarray*}
		&   & \Pr[\forall i \in \{1,\ldots,\kappa \cdot (\log n - \gamma \log \log n) \}: B(u,i) = B(v,i)]\\
		& = & \left(\frac{1}{2}\right)^{\kappa \cdot (\log n - \gamma \log \log n)} < \frac{1}{n^{\kappa -1}}
%		& = & \frac{2^{\kappa \gamma \log \log n}}{2^{\kappa \log n}}\\
%		& = & \frac{\log^{\kappa \gamma} n}{n^{\kappa}}\\
%		& = & \frac{1}{n^{\kappa -1}} \cdot \frac{\log^{\kappa \gamma} n}{n}\\
%		& < & \frac{1}{n^{\kappa -1}}\\
	\end{eqnarray*}

	By the union bound, we show that the probability that no leader has been elected after $\kappa \cdot (\log n - \gamma \log \log n)$ iterations is negligible, i.e., we compute the probability that $\exists u \neq v: \forall i \in \{1,\ldots, \kappa \cdot (\log n - \gamma \log \log n) \}: B(u,i) = B(v,i)$.
	Denote this event by $NOLEADER$.
	We get $$\Pr[NOLEADER] < c \log n \cdot \frac{1}{n^{\kappa -1}} < \frac{1}{n^{\kappa -2}}.$$
\end{proof}

Combining Lemmas~\ref{lem:leader_election:t1} and \ref{lem:leader_election:t2} proves Theorem~\ref{th:leader_election}.

\section{Compass Alignment}
\label{app:alignment}

This appendix states the full proofs of Lemmas~\ref{lem:alignment1} and~\ref{lem:alignment2}.

\begin{proof}[Proof of Lemma~\ref{lem:alignment1}]
    We prove the lemma by induction.
    The statement holds trivially for $\mathcal{R}_0$ since all amoebots are candidates at the beginning.
    
    Suppose the statement holds after the $t$-th iteration.
    Consider region $R \in \mathcal{R}_t$ and its candidates $C_R$.
    By induction hypothesis, $C_R \neq \emptyset$.
    Note that candidates which have tossed $\mathit{HEADS}$ do not revoke their candidacy.
    Therefore, suppose that all candidates have tossed $\mathit{TAILS}$ and hence also region $R$.
    The candidacies of all candidates in $C_R$ are only revoked if region $R$ fuses with another region $R'$.
    Region $R'$ either has failed its coin toss or has tossed $\mathit{HEADS}$.
    In both cases there exists a candidate $u \in C_{R'}$ such that $u.c = \mathit{HEADS}$.
    Candidate $u$ is contained in the fused region and does not revoke its candidacy.
    Thus, each region in $\mathcal{R}_{t+1}$ contains at least one candidate.
\end{proof}

\begin{proof}[Proof of Lemma~\ref{lem:alignment2}]
    Let $C_t$ denote the candidates over all regions after the $t$-th iteration. For each candidate $u \in C_t$, we define
    \begin{equation*}
        Y_u =
        \begin{cases}
            1, & \text{if } u \in C_{t+1} \\
            0, & \text{otherwise}.
        \end{cases}
    \end{equation*}
    Let $R_u$ be the region that contains $u$.
    First, consider the case that the coin toss of region $R_u$ has not been successful.
    The candidacy of $u$ is revoked if $u$ has tossed $\mathit{TAILS}$.
    Therefore, we have
    \begin{equation*}
        \Pr[Y_u = 0 \mid \text{coin toss of region $R_u$ is not successful}] = \frac{1}{2}.
    \end{equation*}
    Next, consider the case that the coin toss of region $R_u$ has been successful.
    The candidacy of $u$ is revoked if $u$ has tossed $\mathit{TAILS}$ and if at least one neighbor of $R_u$ has not tossed $\mathit{TAILS}$.
    The former is given by $\Pr[u.c = \mathit{TAILS}] = \frac{1}{2}$.
    The latter is given by $\Pr[\operatorname N'(R_u) \neq \emptyset] \geq \frac{1}{2}$.
    Note that $\operatorname N(R) \neq \emptyset$.
    Otherwise, $|C| = 1$ already holds.
    Therefore, we have
    \begin{equation*}
        \Pr[Y_u = 0 \mid \text{coin toss of region $R_u$ is successful}] \geq \frac{1}{4}.
    \end{equation*}
    Overall, we obtain $\Pr[Y_u = 0] \geq \frac{1}{4}$ and $\Pr[Y_u = 1] \leq \frac{3}{4}$.
    Then, $\E[Y_u] \leq \frac{3}{4}$.
    
    Let $X_t = |C_t| - 1$ denote the number of candidates after the $t$-th iteration (minus one).
    The value for iteration $t+1$ can be expressed by $X_{t+1} = \sum_{u \in C_t}Y_u - 1$.
    Then,
    \begin{align*}
        \E[X_{t+1} \mid X_t] &= \sum_{u \in C_t} \E[Y_u] - 1 \\
        &\leq \sum_{u \in C_t} \frac{3}{4} - 1 \\
        &= \frac{3}{4} |C_t| - 1 \\
        &= \frac{3}{4} (X_t + 1) - 1 \\
        &= \frac{3}{4} X_t - \frac{1}{4} \\
        &\leq \frac{3}{4} X_t
    \intertext{and}
        \E[X_{t+2} \mid X_t] &= \sum_X \E[X_{t+2} \mid X_{t+1}=X] \cdot \Pr[X_{t+1} = X \mid X_t] \\
        &\leq \sum_X \frac{3}{4} X \cdot \Pr[X_{t+1} = X \mid X_t] \\
        &= \frac{3}{4} \sum_X X \cdot \Pr[X_{t+1} = X \mid X_t] \\
        &= \frac{3}{4} \E[X_{t+1} \mid X_t] \\
        &\leq \left(\frac{3}{4}\right)^2 X_t.
    \end{align*}
    By continuing the argument, we obtain for any constant $c > 1$
    \begin{equation*}
        \E[X_{t+c \log n} \mid X_t] \leq \left(\frac{3}{4}\right)^{c\log n} X_t \leq \left(\frac{3}{4}\right)^{c\log n} n \leq n^{-c'}
    \end{equation*}
    for some $c' = \Theta(c)$.
    For the second inequality we are utilizing the fact that $X_t \leq n$.
    
    By applying the Markov inequality, we obtain $\Pr[X_{t + c \log n} \geq 1] \leq n^{-c'}$.
    Note that a single iteration requires $O(1)$ rounds.
    Thus, $|\bigcup_{R \in \mathcal R} C_R| = 1$ holds after $O(\log n)$ rounds, w.h.p.
\end{proof}

\section{Parallelograms with Specific Side Ratio}
\label{app:parallelogram}

In this appendix, we study the recognition of parallelograms with specific side ratio.
We assume that the corners of the parallelogram coincide with the vertices in $G_{eqt} = (V, E)$, and that the sides of the parallelogram are in parallel to the axes.
An amoebot structure $S$ forms a parallelogram $\mathcal P$ iff $\mathcal P \cap V = S$.
First, we elaborate the difference to the shapes composed of triangles in Section~\ref{subsec:recognition}.
At first glance, one could think that our definition of parallelograms is only a subset of the shapes composed of triangles since parallelograms can be constructed with triangles.
However, we only require that the corners of the parallelogram coincide with the vertices in $G_{eqt}$ instead of all corners of the triangles (see Figure~\ref{fig:shape_def}).
Thus, our definition of parallelograms is not a subset of the shapes composed of triangles.

\begin{figure}[htp]
    \centering
    \begin{subfigure}[b]{0.25\textwidth}
        \centering
        \begin{tikzpicture}
            \tikzmath{
                \scale = .8;
                \t = \scale/2;
                \h = sqrt(\scale^2 - (\scale/2)^2);
                {
                    \filldraw[yellow]
                    (0*\h, 0*\t) --
                    (0*\h, 2*\t) --
                    (3*\h, 5*\t) --
                    (3*\h, 3*\t) --
                    cycle;
                    \draw
                    (0*\h, 0*\t) --
                    (0*\h, 2*\t) --
                    (3*\h, 5*\t) --
                    (3*\h, 3*\t) --
                    cycle;
                };
                for \j in {0,1,...,3}{
                    for \i in {0,2}{
                        {
                            \filldraw (\j*\h, \i*\t+\j*\t) circle (0.04);
                        };
                    };
                    {
                        \draw[dotted] (\j*\h, 0*\t+\j*\t) -- (\j*\h, 2*\t+\j*\t);
                    };
                };
                for \j in {0,1,...,2}{
                    for \i in {0,2}{
                        {
                            \draw[dotted] (\j*\h, \i*\t+\j*\t) -- (\j*\h+\h, \i*\t+\j*\t+\t);
                        };
                    };
                    for \i in {2}{
                        {
                            \draw[dotted] (\j*\h, \i*\t+\j*\t) -- (\j*\h+\h, \i*\t+\j*\t-\t);
                        };
                    };
                };
            }
        \end{tikzpicture}
        \caption{\centering}
    \end{subfigure}
    \hfill
    \begin{subfigure}[b]{0.25\textwidth}
        \centering
        \begin{tikzpicture}
            \tikzmath{
                \scale = .8;
                \t = \scale/2;
                \h = sqrt(\scale^2 - (\scale/2)^2);
                {
                    \filldraw[yellow]
                    (0*\h, 0*\t) --
                    (0*\h, 2*\t) --
                    (3*\h, 5*\t) --
                    (3*\h, 3*\t) --
                    cycle;
                };
                for \j in {0,1,...,3}{
                    for \i in {0,2}{
                        {
                            \filldraw (\j*\h, \i*\t+\j*\t) circle (0.04);
                        };
                    };
                    {
                        \draw (\j*\h, 0*\t+\j*\t) -- (\j*\h, 2*\t+\j*\t);
                    };
                };
                for \j in {0,1,...,2}{
                    for \i in {0,2}{
                        {
                            \draw (\j*\h, \i*\t+\j*\t) -- (\j*\h+\h, \i*\t+\j*\t+\t);
                        };
                    };
                    for \i in {2}{
                        {
                            \draw (\j*\h, \i*\t+\j*\t) -- (\j*\h+\h, \i*\t+\j*\t-\t);
                        };
                    };
                };
            }
        \end{tikzpicture}
        \caption{\centering}
    \end{subfigure}
    \hfill
    \begin{subfigure}[b]{0.35\textwidth}
        \centering
        \begin{tikzpicture}
            \tikzmath{
                \scale = .8;
                \t = \scale/2;
                \h = sqrt(\scale^2 - (\scale/2)^2);
                {
                    \filldraw[yellow]
                    (0*\h, 0*\t) --
                    (0*\h, 4*\t) --
                    (5*\h, 9*\t) --
                    (5*\h, 5*\t) --
                    cycle;
                    \draw
                    (0*\h, 0*\t) --
                    (0*\h, 4*\t) --
                    (5*\h, 9*\t) --
                    (5*\h, 5*\t) --
                    cycle;
                };
                for \j in {0,1,...,5}{
                    for \i in {0,2,4}{
                        {
                            \filldraw (\j*\h, \i*\t+\j*\t) circle (0.04);
                        };
                    };
                    {
                        \draw[dotted] (\j*\h, 0*\t+\j*\t) -- (\j*\h, 4*\t+\j*\t);
                    };
                };
                for \j in {0,1,...,4}{
                    for \i in {0,2,4}{
                        {
                            \draw[dotted] (\j*\h, \i*\t+\j*\t) -- (\j*\h+\h, \i*\t+\j*\t+\t);
                        };
                    };
                    for \i in {2,4}{
                        {
                            \draw[dotted] (\j*\h, \i*\t+\j*\t) -- (\j*\h+\h, \i*\t+\j*\t-\t);
                        };
                    };
                };
            }
        \end{tikzpicture}
        \caption{\centering}
    \end{subfigure}
    \caption{
        Consider a parallelogram with height $h$ and length $l = 2 \cdot h + 1$.
        (a)~shows the parallelogram with height $h = 1$.
        (b) shows the same parallelogram as a shape composed of triangles.
        (c) shows the parallelogram with height $h = 2$.
        However, we cannot scale the shape of (b) to this parallelogram.
    }
    \label{fig:shape_def}
\end{figure}
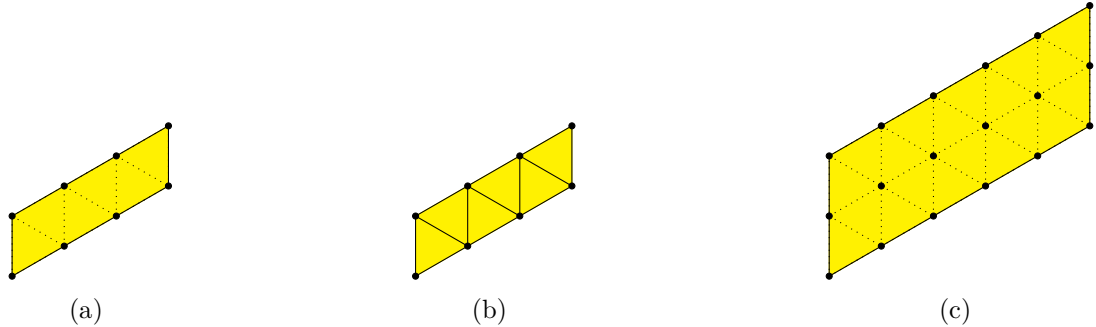

Our algorithms make use of the global circuit and circuits along the axes.
These are explained in Sections~\ref{sec:leader} and~\ref{subsec:recognition}.
As an intermediate step, we consider the detection of a parallelogram with arbitrary side ratio.
Many basic geometric shapes can be easily detected by an amoebot structure.

\begin{lemma}
\label{lem:parallelogram}
    An amoebot structure can detect a parallelogram with arbitrary side ratio.
    A constant number of rounds is required.
\end{lemma}

\begin{proof}
    We utilize an arbitrary number of pins.
    A single amoebot can locally determine that it is a (degenerate) parallelogram.
    Suppose there are more than one amoebot.
    We gather information about the neighborhoods of the amoebots by assigning a round to each type.
    Each amoebot activates the global circuit in the corresponding round.
    We distinguish between the following types of neighborhoods: (i)~a single neighbor, (ii) 2 contiguous neighbors, (iii) 3 contiguous neighbors, (iv) 4 contiguous neighbors, (v) 6 neighbors, (vi) 2 opposing neighbors, and (vii) all remaining neighborhoods.
    
    Type (vii) nodes rule out a parallelogram immediately.
    If the structure has type (i) nodes, then all nodes have to be of type (i) or (vi).
    In this case the amoebot structure is a degenerate parallelogram.
    If the structure has no type (i) nodes, then all nodes have to be of type (ii), (iii), (iv) or (v) and there has to be a node of both types (ii) and (iii).
    In this case the type (ii) nodes activate their circuits along the axes.
    If the type (iii) nodes receive a beep on two circuits along the axes, then the structure is a parallelogram.
\end{proof}

\begin{remark}
    Lines, triangles, hexagons and trapezoids can be detected in a similar fashion.
\end{remark}

We are able to adopt results for shape recognition in a related model \cite{DBLP:conf/mfcs/GmyrHKKRS18}:
A single robot with a constant number of pebbles is placed on a finite and connected subgraph of the infinite regular triangular grid graph $G_{eqt} = (V,E)$ induced by some subset $T \subseteq V$.
It can move on the subgraph and carry a set of pebbles which can be placed on and removed from nodes.
A node can contain at most one pebble.
The robot can sense the state of adjacent nodes.

Algorithms for this model can be easily simulated within the amoebot model.
An amoebot is placed on each node of the subgraph.
The current position of the robot, the number of the pebbles it carries and the marked nodes can be stored locally within the amoebots.
The state of the adjacent nodes is given by messages from neighboring amoebots.
The computational power of the robot and an amoebot are equivalent.
Note that an initial position for the robot has to be determined.
For that, we perform a leader election (see Section~\ref{sec:leader}).
This takes $\Theta(\log n)$ rounds, w.h.p.

W.l.o.g., we adopt the point of view of~\cite{DBLP:conf/mfcs/GmyrHKKRS18} and assume that the parallelogram is aligned to the axes in northern and north-eastern direction.
However, we define the length of a side by the number of edges instead of the number of nodes.
Let $h$ denote the number of edges in northern direction (height) and $l$ the number of edges in the north-eastern direction (length) where $h \leq l$ holds.

We move the robot to a corner of type (iii) (see Lemma~\ref{lem:parallelogram}).
W.l.o.g., we assume that the robot is moved to the north-western corner.
For this purpose, we utilize circuits along the axes.
If the robot is not positioned on the boundary, we let the amoebot that simulates the robot split one of the circuits along the axes and activate one of the resulting parts.
The new position of the robot is given by the amoebot at the boundary that receives the beep.
If the robot is positioned on the boundary but not on a corner, we let the amoebot that simulates the robot split the circuits along the boundary and activate one of the resulting parts.
The new position of the robot is given by the amoebot on the corner that receives the beep.
If the robot is positioned on a corner of type (ii), we let the amoebot that simulates the robot activate one of its circuits along the boundary.
The new position of the robot is given by the amoebot at the other end of the circuit.
This is a corner of type (iii).
Altogether, this takes $O(1)$ rounds.

\begin{theorem}
\label{th:linear}
    The amoebot structure can detect a parallelogram with linear side ratio, i.e., $l = ah + b$ where $a,b \in \N$ are constants.
    It requires $O(a+b)=O(1)$ rounds after determining an initial postion for the robot.
    Altogether, it requires $\Theta(\log n)$ rounds, w.h.p.
\end{theorem}

\begin{proof}
    We utilize an arbitrary number of pins.
    Suppose that the amoebot structure is indeed a non-degenerate parallelogram and that it has already verified this fact by Lemma~\ref{lem:parallelogram}.
    Otherwise, a negative result can be returned.
    Determine an initial position for the robot by leader election and move it to the north-western corner.

    We simulate and adjust the algorithm for parallelograms with linear side ratio stated by~\cite{DBLP:conf/mfcs/GmyrHKKRS18}:
    The robot moves $ah$ nodes to the north-eastern direction by performing a zig-zag pattern.
    Thereby, the following loop is executed $a$ times:
    (i) move south-east until the southern side is reached, and (ii) move north until the northern side is reached.
    Afterwards additional $b$ steps in north-eastern direction are made.
    The robot returns a negative answer if it is unable to perform a movement at any point of time or if it has not arrived at the north-eastern corner of the structure at the end of the algorithm.
    Figure~\ref{fig:linear} illustruates the algorithm.

    An amoebot structure with circuits along the axes is able to simply transfer the robot to the other side of the parallelogram instead of moving it node by node.
    This reduces the running time to $O(a+b)$.
    
    Note that the amoebots do neither have to agree on a common chirality nor on a common orientation.
    While performing the zig-zag pattern, the direction of the simulated robot is mirrored on the boundary.
    The additional $b$ steps along the boundary are performed into the direction that only requires a rotation by $60^\circ$.
\end{proof}

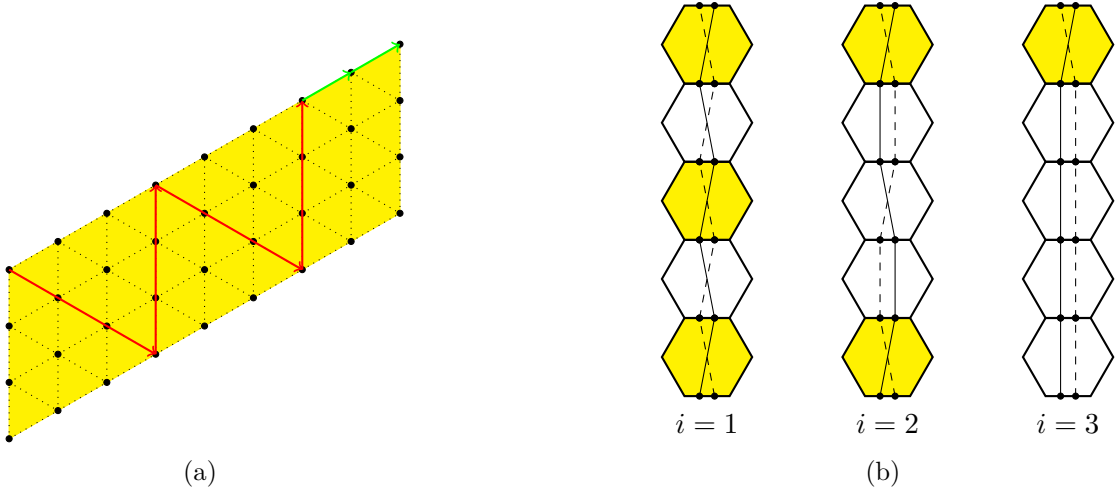
\begin{figure}[htp]
    \centering
    \begin{subfigure}[b]{0.45\textwidth}
        \centering
        \begin{tikzpicture}
            \tikzmath{
                \scale = .75;
                \t = \scale/2;
                \h = sqrt(\scale^2 - (\scale/2)^2);
                {
                    \filldraw[yellow]
                    (0*\h, 0*\t) --
                    (0*\h, 6*\t) --
                    (8*\h, 14*\t) --
                    (8*\h, 8*\t) --
                    cycle;
                };
                for \j in {0,1,...,8}{
                    for \i in {0,2,4,6}{
                        {
                            \filldraw (\j*\h, \i*\t+\j*\t) circle (0.04);
                        };
                    };
                    {
                        \draw[dotted] (\j*\h, 0*\t+\j*\t) -- (\j*\h, 6*\t+\j*\t);
                    };
                };
                for \j in {0,1,...,7}{
                    for \i in {0,2,4,6}{
                        {
                            \draw[dotted] (\j*\h, \i*\t+\j*\t) -- (\j*\h+\h, \i*\t+\j*\t+\t);
                        };
                    };
                    for \i in {2,4,6}{
                        {
                            \draw[dotted] (\j*\h, \i*\t+\j*\t) -- (\j*\h+\h, \i*\t+\j*\t-\t);
                        };
                    };
                };
                {
                    \draw[thick, ->, color = red] (0*\h, 6*\t) -- (3*\h, 3*\t);
                    \draw[thick, ->, color = red] (3*\h, 3*\t) -- (3*\h, 9*\t);
                    \draw[thick, ->, color = red] (3*\h, 9*\t) -- (6*\h, 6*\t);
                    \draw[thick, ->, color = red] (6*\h, 6*\t) -- (6*\h, 12*\t);
                    \draw[thick, ->, color = green] (6*\h, 12*\t) -- (7*\h, 13*\t);
                    \draw[thick, ->, color = green] (7*\h, 13*\t) -- (8*\h, 14*\t);
                };
            }
        \end{tikzpicture}
        \caption{\centering}
        \label{fig:linear}
    \end{subfigure}
    \hfill
    \begin{subfigure}[b]{0.45\textwidth}
        \centering
        \begin{tikzpicture}
            \def\s{.6};  % scale
            
        % subfigure 1
            
            % active amoebots
            \def\i{0};
            \def\j{0};
            \fill[yellow]
            ({\s/2+\i*\s*3/2+\j*\s*3/2},{heightEqt(\s)*\i-heightEqt(\s)*\j}) --
            ({\s)/2*3+\i*\s*3/2+\j*\s*3/2},{heightEqt(\s)*\i-heightEqt(\s)*\j}) --
            ({\s*2+\i*\s*3/2+\j*\s*3/2},{heightEqt(\s)+heightEqt(\s)*\i-heightEqt(\s)*\j}) --
            ({\s)/2*3+\i*\s*3/2+\j*\s*3/2},{heightEqt(\s)*2+heightEqt(\s)*\i-heightEqt(\s)*\j}) --
            ({\s/2+\i*\s*3/2+\j*\s*3/2},{heightEqt(\s)*2+heightEqt(\s)*\i-heightEqt(\s)*\j}) --
            ({\i*\s*3/2+\j*\s*3/2},{heightEqt(\s)+heightEqt(\s)*\i-heightEqt(\s)*\j}) -- cycle;
            
            \def\i{-2};
            \def\j{2};
            \fill[yellow]
            ({\s/2+\i*\s*3/2+\j*\s*3/2},{heightEqt(\s)*\i-heightEqt(\s)*\j}) --
            ({\s)/2*3+\i*\s*3/2+\j*\s*3/2},{heightEqt(\s)*\i-heightEqt(\s)*\j}) --
            ({\s*2+\i*\s*3/2+\j*\s*3/2},{heightEqt(\s)+heightEqt(\s)*\i-heightEqt(\s)*\j}) --
            ({\s)/2*3+\i*\s*3/2+\j*\s*3/2},{heightEqt(\s)*2+heightEqt(\s)*\i-heightEqt(\s)*\j}) --
            ({\s/2+\i*\s*3/2+\j*\s*3/2},{heightEqt(\s)*2+heightEqt(\s)*\i-heightEqt(\s)*\j}) --
            ({\i*\s*3/2+\j*\s*3/2},{heightEqt(\s)+heightEqt(\s)*\i-heightEqt(\s)*\j}) -- cycle;
            
            \def\i{-4};
            \def\j{4};
            \fill[yellow]
            ({\s/2+\i*\s*3/2+\j*\s*3/2},{heightEqt(\s)*\i-heightEqt(\s)*\j}) --
            ({\s)/2*3+\i*\s*3/2+\j*\s*3/2},{heightEqt(\s)*\i-heightEqt(\s)*\j}) --
            ({\s*2+\i*\s*3/2+\j*\s*3/2},{heightEqt(\s)+heightEqt(\s)*\i-heightEqt(\s)*\j}) --
            ({\s)/2*3+\i*\s*3/2+\j*\s*3/2},{heightEqt(\s)*2+heightEqt(\s)*\i-heightEqt(\s)*\j}) --
            ({\s/2+\i*\s*3/2+\j*\s*3/2},{heightEqt(\s)*2+heightEqt(\s)*\i-heightEqt(\s)*\j}) --
            ({\i*\s*3/2+\j*\s*3/2},{heightEqt(\s)+heightEqt(\s)*\i-heightEqt(\s)*\j}) -- cycle;

            \foreach \j in {0,1,...,4} {
                \def\i{-\j};
                
                % amoebot
                \draw[black, thick]
                ({\s/2+\i*\s*3/2+\j*\s*3/2},{heightEqt(\s)*\i-heightEqt(\s)*\j}) --
                ({\s)/2*3+\i*\s*3/2+\j*\s*3/2},{heightEqt(\s)*\i-heightEqt(\s)*\j}) --
                ({\s*2+\i*\s*3/2+\j*\s*3/2},{heightEqt(\s)+heightEqt(\s)*\i-heightEqt(\s)*\j}) --
                ({\s)/2*3+\i*\s*3/2+\j*\s*3/2},{heightEqt(\s)*2+heightEqt(\s)*\i-heightEqt(\s)*\j}) --
                ({\s/2+\i*\s*3/2+\j*\s*3/2},{heightEqt(\s)*2+heightEqt(\s)*\i-heightEqt(\s)*\j}) --
                ({\i*\s*3/2+\j*\s*3/2},{heightEqt(\s)+heightEqt(\s)*\i-heightEqt(\s)*\j}) -- cycle;
            
                % pins
                \filldraw ({\s/2+\s/3+\i*\s*3/2+\j*\s*3/2},{heightEqt(\s)*\i-heightEqt(\s)*\j}) circle (0.04);
                \filldraw ({\s/2+\s/3*2+\i*\s*3/2+\j*\s*3/2},{heightEqt(\s)*\i-heightEqt(\s)*\j}) circle (0.04);
                
                \filldraw ({\s/2+\s/3+\i*\s*3/2+\j*\s*3/2},{heightEqt(\s)*2+heightEqt(\s)*\i-heightEqt(\s)*\j}) circle (0.04);
                \filldraw ({\s/2+\s/3*2+\i*\s*3/2+\j*\s*3/2},{heightEqt(\s)*2+heightEqt(\s)*\i-heightEqt(\s)*\j}) circle (0.04);
            }
            
            % connections
            \def\i{0};
            \def\j{0};
            \draw
            ({\s/2+\s/3+\i*\s*3/2+\j*\s*3/2},{heightEqt(\s)*\i-heightEqt(\s)*\j})
            -- ({\s/2+\s/3*2+\i*\s*3/2+\j*\s*3/2},{heightEqt(\s)*2+heightEqt(\s)*\i-heightEqt(\s)*\j});
            \draw[dashed]
            ({\s/2+\s/3*2+\i*\s*3/2+\j*\s*3/2},{heightEqt(\s)*\i-heightEqt(\s)*\j})
            -- ({\s/2+\s/3+\i*\s*3/2+\j*\s*3/2},{heightEqt(\s)*2+heightEqt(\s)*\i-heightEqt(\s)*\j});
            
            \def\i{-1};
            \def\j{1};
            \draw[dashed]
            ({\s/2+\s/3+\i*\s*3/2+\j*\s*3/2},{heightEqt(\s)*\i-heightEqt(\s)*\j})
            -- ({\s/2+\s/3*2+\i*\s*3/2+\j*\s*3/2},{heightEqt(\s)*2+heightEqt(\s)*\i-heightEqt(\s)*\j});
            \draw
            ({\s/2+\s/3*2+\i*\s*3/2+\j*\s*3/2},{heightEqt(\s)*\i-heightEqt(\s)*\j})
            -- ({\s/2+\s/3+\i*\s*3/2+\j*\s*3/2},{heightEqt(\s)*2+heightEqt(\s)*\i-heightEqt(\s)*\j});
            
            \def\i{-2};
            \def\j{2};
            \draw
            ({\s/2+\s/3+\i*\s*3/2+\j*\s*3/2},{heightEqt(\s)*\i-heightEqt(\s)*\j})
            -- ({\s/2+\s/3*2+\i*\s*3/2+\j*\s*3/2},{heightEqt(\s)*2+heightEqt(\s)*\i-heightEqt(\s)*\j});
            \draw[dashed]
            ({\s/2+\s/3*2+\i*\s*3/2+\j*\s*3/2},{heightEqt(\s)*\i-heightEqt(\s)*\j})
            -- ({\s/2+\s/3+\i*\s*3/2+\j*\s*3/2},{heightEqt(\s)*2+heightEqt(\s)*\i-heightEqt(\s)*\j});
            
            \def\i{-3};
            \def\j{3};
            \draw[dashed]
            ({\s/2+\s/3+\i*\s*3/2+\j*\s*3/2},{heightEqt(\s)*\i-heightEqt(\s)*\j})
            -- ({\s/2+\s/3*2+\i*\s*3/2+\j*\s*3/2},{heightEqt(\s)*2+heightEqt(\s)*\i-heightEqt(\s)*\j});
            \draw
            ({\s/2+\s/3*2+\i*\s*3/2+\j*\s*3/2},{heightEqt(\s)*\i-heightEqt(\s)*\j})
            -- ({\s/2+\s/3+\i*\s*3/2+\j*\s*3/2},{heightEqt(\s)*2+heightEqt(\s)*\i-heightEqt(\s)*\j});
            
            \def\i{-4};
            \def\j{4};
            \draw
            ({\s/2+\s/3+\i*\s*3/2+\j*\s*3/2},{heightEqt(\s)*\i-heightEqt(\s)*\j})
            -- ({\s/2+\s/3*2+\i*\s*3/2+\j*\s*3/2},{heightEqt(\s)*2+heightEqt(\s)*\i-heightEqt(\s)*\j});
            \draw[dashed]
            ({\s/2+\s/3*2+\i*\s*3/2+\j*\s*3/2},{heightEqt(\s)*\i-heightEqt(\s)*\j})
            -- ({\s/2+\s/3+\i*\s*3/2+\j*\s*3/2},{heightEqt(\s)*2+heightEqt(\s)*\i-heightEqt(\s)*\j});
            
            \def\i{-4};
            \def\j{4};
            \draw[black, below] ({\s/2+\s/2+\i*\s*3/2+\j*\s*3/2},{heightEqt(\s)*\i-heightEqt(\s)*\j-.1}) node {$i = 1$};
                
        % subfigure 2
                
            % active amoebots
            \def\i{0};
            \def\j{0};
            \fill[yellow]
            ({\s/2+\i*\s*3/2+\j*\s*3/2+4*\s},{heightEqt(\s)*\i-heightEqt(\s)*\j}) --
            ({\s)/2*3+\i*\s*3/2+\j*\s*3/2+4*\s},{heightEqt(\s)*\i-heightEqt(\s)*\j}) --
            ({\s*2+\i*\s*3/2+\j*\s*3/2+4*\s},{heightEqt(\s)+heightEqt(\s)*\i-heightEqt(\s)*\j}) --
            ({\s)/2*3+\i*\s*3/2+\j*\s*3/2+4*\s},{heightEqt(\s)*2+heightEqt(\s)*\i-heightEqt(\s)*\j}) --
            ({\s/2+\i*\s*3/2+\j*\s*3/2+4*\s},{heightEqt(\s)*2+heightEqt(\s)*\i-heightEqt(\s)*\j}) --
            ({\i*\s*3/2+\j*\s*3/2+4*\s},{heightEqt(\s)+heightEqt(\s)*\i-heightEqt(\s)*\j}) -- cycle;
            
            \def\i{-4};
            \def\j{4};
            \fill[yellow]
            ({\s/2+\i*\s*3/2+\j*\s*3/2+4*\s},{heightEqt(\s)*\i-heightEqt(\s)*\j}) --
            ({\s)/2*3+\i*\s*3/2+\j*\s*3/2+4*\s},{heightEqt(\s)*\i-heightEqt(\s)*\j}) --
            ({\s*2+\i*\s*3/2+\j*\s*3/2+4*\s},{heightEqt(\s)+heightEqt(\s)*\i-heightEqt(\s)*\j}) --
            ({\s)/2*3+\i*\s*3/2+\j*\s*3/2+4*\s},{heightEqt(\s)*2+heightEqt(\s)*\i-heightEqt(\s)*\j}) --
            ({\s/2+\i*\s*3/2+\j*\s*3/2+4*\s},{heightEqt(\s)*2+heightEqt(\s)*\i-heightEqt(\s)*\j}) --
            ({\i*\s*3/2+\j*\s*3/2+4*\s},{heightEqt(\s)+heightEqt(\s)*\i-heightEqt(\s)*\j}) -- cycle;
            
            \foreach \j in {0,1,...,4} {
                \def\i{-\j};
                
                % amoebot
                \draw[black, thick]
                ({\s/2+\i*\s*3/2+\j*\s*3/2+4*\s},{heightEqt(\s)*\i-heightEqt(\s)*\j}) --
                ({\s)/2*3+\i*\s*3/2+\j*\s*3/2+4*\s},{heightEqt(\s)*\i-heightEqt(\s)*\j}) --
                ({\s*2+\i*\s*3/2+\j*\s*3/2+4*\s},{heightEqt(\s)+heightEqt(\s)*\i-heightEqt(\s)*\j}) --
                ({\s)/2*3+\i*\s*3/2+\j*\s*3/2+4*\s},{heightEqt(\s)*2+heightEqt(\s)*\i-heightEqt(\s)*\j}) --
                ({\s/2+\i*\s*3/2+\j*\s*3/2+4*\s},{heightEqt(\s)*2+heightEqt(\s)*\i-heightEqt(\s)*\j}) --
                ({\i*\s*3/2+\j*\s*3/2+4*\s},{heightEqt(\s)+heightEqt(\s)*\i-heightEqt(\s)*\j}) -- cycle;
            
                % pins
                \filldraw ({\s/2+\s/3+\i*\s*3/2+\j*\s*3/2+4*\s},{heightEqt(\s)*\i-heightEqt(\s)*\j}) circle (0.04);
                \filldraw ({\s/2+\s/3*2+\i*\s*3/2+\j*\s*3/2+4*\s},{heightEqt(\s)*\i-heightEqt(\s)*\j}) circle (0.04);
                
                \filldraw ({\s/2+\s/3+\i*\s*3/2+\j*\s*3/2+4*\s},{heightEqt(\s)*2+heightEqt(\s)*\i-heightEqt(\s)*\j}) circle (0.04);
                \filldraw ({\s/2+\s/3*2+\i*\s*3/2+\j*\s*3/2+4*\s},{heightEqt(\s)*2+heightEqt(\s)*\i-heightEqt(\s)*\j}) circle (0.04);
            }
            
            % connections
            \def\i{0};
            \def\j{0};
            \draw
            ({\s/2+\s/3+\i*\s*3/2+\j*\s*3/2+4*\s},{heightEqt(\s)*\i-heightEqt(\s)*\j})
            -- ({\s/2+\s/3*2+\i*\s*3/2+\j*\s*3/2+4*\s},{heightEqt(\s)*2+heightEqt(\s)*\i-heightEqt(\s)*\j});
            \draw[dashed]
            ({\s/2+\s/3*2+\i*\s*3/2+\j*\s*3/2+4*\s},{heightEqt(\s)*\i-heightEqt(\s)*\j})
            -- ({\s/2+\s/3+\i*\s*3/2+\j*\s*3/2+4*\s},{heightEqt(\s)*2+heightEqt(\s)*\i-heightEqt(\s)*\j});
            
            \def\i{-1};
            \def\j{1};
            \draw
            ({\s/2+\s/3+\i*\s*3/2+\j*\s*3/2+4*\s},{heightEqt(\s)*\i-heightEqt(\s)*\j})
            -- ({\s/2+\s/3+\i*\s*3/2+\j*\s*3/2+4*\s},{heightEqt(\s)*2+heightEqt(\s)*\i-heightEqt(\s)*\j});
            \draw[dashed]
            ({\s/2+\s/3*2+\i*\s*3/2+\j*\s*3/2+4*\s},{heightEqt(\s)*\i-heightEqt(\s)*\j})
            -- ({\s/2+\s/3*2+\i*\s*3/2+\j*\s*3/2+4*\s},{heightEqt(\s)*2+heightEqt(\s)*\i-heightEqt(\s)*\j});
            
            \def\i{-2};
            \def\j{2};
            \draw[dashed]
            ({\s/2+\s/3+\i*\s*3/2+\j*\s*3/2+4*\s},{heightEqt(\s)*\i-heightEqt(\s)*\j})
            -- ({\s/2+\s/3*2+\i*\s*3/2+\j*\s*3/2+4*\s},{heightEqt(\s)*2+heightEqt(\s)*\i-heightEqt(\s)*\j});
            \draw
            ({\s/2+\s/3*2+\i*\s*3/2+\j*\s*3/2+4*\s},{heightEqt(\s)*\i-heightEqt(\s)*\j})
            -- ({\s/2+\s/3+\i*\s*3/2+\j*\s*3/2+4*\s},{heightEqt(\s)*2+heightEqt(\s)*\i-heightEqt(\s)*\j});
            
            \def\i{-3};
            \def\j{3};
            \draw[dashed]
            ({\s/2+\s/3+\i*\s*3/2+\j*\s*3/2+4*\s},{heightEqt(\s)*\i-heightEqt(\s)*\j})
            -- ({\s/2+\s/3+\i*\s*3/2+\j*\s*3/2+4*\s},{heightEqt(\s)*2+heightEqt(\s)*\i-heightEqt(\s)*\j});
            \draw
            ({\s/2+\s/3*2+\i*\s*3/2+\j*\s*3/2+4*\s},{heightEqt(\s)*\i-heightEqt(\s)*\j})
            -- ({\s/2+\s/3*2+\i*\s*3/2+\j*\s*3/2+4*\s},{heightEqt(\s)*2+heightEqt(\s)*\i-heightEqt(\s)*\j});
            
            \def\i{-4};
            \def\j{4};
            \draw
            ({\s/2+\s/3+\i*\s*3/2+\j*\s*3/2+4*\s},{heightEqt(\s)*\i-heightEqt(\s)*\j})
            -- ({\s/2+\s/3*2+\i*\s*3/2+\j*\s*3/2+4*\s},{heightEqt(\s)*2+heightEqt(\s)*\i-heightEqt(\s)*\j});
            \draw[dashed]
            ({\s/2+\s/3*2+\i*\s*3/2+\j*\s*3/2+4*\s},{heightEqt(\s)*\i-heightEqt(\s)*\j})
            -- ({\s/2+\s/3+\i*\s*3/2+\j*\s*3/2+4*\s},{heightEqt(\s)*2+heightEqt(\s)*\i-heightEqt(\s)*\j});
            
            \def\i{-4};
            \def\j{4};
            \draw[black, below] ({\s/2+\s/2+\i*\s*3/2+\j*\s*3/2+4*\s},{heightEqt(\s)*\i-heightEqt(\s)*\j-.1}) node {$i = 2$};
            
        % subfigure 3
                
            % active amoebots
            \def\i{0};
            \def\j{0};
            \fill[yellow]
            ({\s/2+\i*\s*3/2+\j*\s*3/2+8*\s},{heightEqt(\s)*\i-heightEqt(\s)*\j}) --
            ({\s)/2*3+\i*\s*3/2+\j*\s*3/2+8*\s},{heightEqt(\s)*\i-heightEqt(\s)*\j}) --
            ({\s*2+\i*\s*3/2+\j*\s*3/2+8*\s},{heightEqt(\s)+heightEqt(\s)*\i-heightEqt(\s)*\j}) --
            ({\s)/2*3+\i*\s*3/2+\j*\s*3/2+8*\s},{heightEqt(\s)*2+heightEqt(\s)*\i-heightEqt(\s)*\j}) --
            ({\s/2+\i*\s*3/2+\j*\s*3/2+8*\s},{heightEqt(\s)*2+heightEqt(\s)*\i-heightEqt(\s)*\j}) --
            ({\i*\s*3/2+\j*\s*3/2+8*\s},{heightEqt(\s)+heightEqt(\s)*\i-heightEqt(\s)*\j}) -- cycle;
            
            \foreach \j in {0,1,...,4} {
                \def\i{-\j};
                
                % amoebot
                \draw[black, thick]
                ({\s/2+\i*\s*3/2+\j*\s*3/2+8*\s},{heightEqt(\s)*\i-heightEqt(\s)*\j}) --
                ({\s)/2*3+\i*\s*3/2+\j*\s*3/2+8*\s},{heightEqt(\s)*\i-heightEqt(\s)*\j}) --
                ({\s*2+\i*\s*3/2+\j*\s*3/2+8*\s},{heightEqt(\s)+heightEqt(\s)*\i-heightEqt(\s)*\j}) --
                ({\s)/2*3+\i*\s*3/2+\j*\s*3/2+8*\s},{heightEqt(\s)*2+heightEqt(\s)*\i-heightEqt(\s)*\j}) --
                ({\s/2+\i*\s*3/2+\j*\s*3/2+8*\s},{heightEqt(\s)*2+heightEqt(\s)*\i-heightEqt(\s)*\j}) --
                ({\i*\s*3/2+\j*\s*3/2+8*\s},{heightEqt(\s)+heightEqt(\s)*\i-heightEqt(\s)*\j}) -- cycle;
            
                % pins
                \filldraw ({\s/2+\s/3+\i*\s*3/2+\j*\s*3/2+8*\s},{heightEqt(\s)*\i-heightEqt(\s)*\j}) circle (0.04);
                \filldraw ({\s/2+\s/3*2+\i*\s*3/2+\j*\s*3/2+8*\s},{heightEqt(\s)*\i-heightEqt(\s)*\j}) circle (0.04);
                
                \filldraw ({\s/2+\s/3+\i*\s*3/2+\j*\s*3/2+8*\s},{heightEqt(\s)*2+heightEqt(\s)*\i-heightEqt(\s)*\j}) circle (0.04);
                \filldraw ({\s/2+\s/3*2+\i*\s*3/2+\j*\s*3/2+8*\s},{heightEqt(\s)*2+heightEqt(\s)*\i-heightEqt(\s)*\j}) circle (0.04);
            }
            
            % connections
            \def\i{0};
            \def\j{0};
            \draw
            ({\s/2+\s/3+\i*\s*3/2+\j*\s*3/2+8*\s},{heightEqt(\s)*\i-heightEqt(\s)*\j})
            -- ({\s/2+\s/3*2+\i*\s*3/2+\j*\s*3/2+8*\s},{heightEqt(\s)*2+heightEqt(\s)*\i-heightEqt(\s)*\j});
            \draw[dashed]
            ({\s/2+\s/3*2+\i*\s*3/2+\j*\s*3/2+8*\s},{heightEqt(\s)*\i-heightEqt(\s)*\j})
            -- ({\s/2+\s/3+\i*\s*3/2+\j*\s*3/2+8*\s},{heightEqt(\s)*2+heightEqt(\s)*\i-heightEqt(\s)*\j});
            
            \def\i{-1};
            \def\j{1};
            \draw
            ({\s/2+\s/3+\i*\s*3/2+\j*\s*3/2+8*\s},{heightEqt(\s)*\i-heightEqt(\s)*\j})
            -- ({\s/2+\s/3+\i*\s*3/2+\j*\s*3/2+8*\s},{heightEqt(\s)*2+heightEqt(\s)*\i-heightEqt(\s)*\j});
            \draw[dashed]
            ({\s/2+\s/3*2+\i*\s*3/2+\j*\s*3/2+8*\s},{heightEqt(\s)*\i-heightEqt(\s)*\j})
            -- ({\s/2+\s/3*2+\i*\s*3/2+\j*\s*3/2+8*\s},{heightEqt(\s)*2+heightEqt(\s)*\i-heightEqt(\s)*\j});
            
            \def\i{-2};
            \def\j{2};
            \draw
            ({\s/2+\s/3+\i*\s*3/2+\j*\s*3/2+8*\s},{heightEqt(\s)*\i-heightEqt(\s)*\j})
            -- ({\s/2+\s/3+\i*\s*3/2+\j*\s*3/2+8*\s},{heightEqt(\s)*2+heightEqt(\s)*\i-heightEqt(\s)*\j});
            \draw[dashed]
            ({\s/2+\s/3*2+\i*\s*3/2+\j*\s*3/2+8*\s},{heightEqt(\s)*\i-heightEqt(\s)*\j})
            -- ({\s/2+\s/3*2+\i*\s*3/2+\j*\s*3/2+8*\s},{heightEqt(\s)*2+heightEqt(\s)*\i-heightEqt(\s)*\j});
            
            \def\i{-3};
            \def\j{3};
            \draw
            ({\s/2+\s/3+\i*\s*3/2+\j*\s*3/2+8*\s},{heightEqt(\s)*\i-heightEqt(\s)*\j})
            -- ({\s/2+\s/3+\i*\s*3/2+\j*\s*3/2+8*\s},{heightEqt(\s)*2+heightEqt(\s)*\i-heightEqt(\s)*\j});
            \draw[dashed]
            ({\s/2+\s/3*2+\i*\s*3/2+\j*\s*3/2+8*\s},{heightEqt(\s)*\i-heightEqt(\s)*\j})
            -- ({\s/2+\s/3*2+\i*\s*3/2+\j*\s*3/2+8*\s},{heightEqt(\s)*2+heightEqt(\s)*\i-heightEqt(\s)*\j});
            
            \def\i{-4};
            \def\j{4};
            \draw
            ({\s/2+\s/3+\i*\s*3/2+\j*\s*3/2+8*\s},{heightEqt(\s)*\i-heightEqt(\s)*\j})
            -- ({\s/2+\s/3+\i*\s*3/2+\j*\s*3/2+8*\s},{heightEqt(\s)*2+heightEqt(\s)*\i-heightEqt(\s)*\j});
            \draw[dashed]
            ({\s/2+\s/3*2+\i*\s*3/2+\j*\s*3/2+8*\s},{heightEqt(\s)*\i-heightEqt(\s)*\j})
            -- ({\s/2+\s/3*2+\i*\s*3/2+\j*\s*3/2+8*\s},{heightEqt(\s)*2+heightEqt(\s)*\i-heightEqt(\s)*\j});
            
            \def\i{-4};
            \def\j{4};
            \draw[black, below] ({\s/2+\s/2+\i*\s*3/2+\j*\s*3/2+8*\s},{heightEqt(\s)*\i-heightEqt(\s)*\j-.1}) node {$i = 3$};
        \end{tikzpicture}
        \caption{\centering}
        \label{fig:add_tree}
    \end{subfigure}
    \caption{
        (a) Let $a = 2$ and $b = 2$.
        The parallelogram with height $h = 3$ and length $l = 2 \cdot 3 + 2 = 8$ is aligned to the axes in northern and north-eastern direction.
        The robot starts at the north-western corner.
        The red arrows indicate the zig-zag pattern.
        The green arrows indicate the additional $b$ steps.
        (b) shows the binary add tree structure.
        W.l.o.g., the first pin of each bond is depicted on the left.
        The solid line indicates the active circuit.
        The dashed line indicates the inactive circuit.
        The yellow amoebots indicate set $L_i$.
    }
\end{figure}

\begin{remark}
    The same adaption for $b \in \Z$ can be applied as stated by~\cite{DBLP:conf/mfcs/GmyrHKKRS18}.
    The additional $|b|$ steps are performed after the first loop if $|b| < h$ holds.
    The function values for small $h = O(|b|)$ can be encoded within the finite state machines of the amoebots and be checked priorly.
\end{remark}

\begin{theorem}
\label{th:polynomial}
    The amoebot structure can detect a parallelogram with polynomial side ratio, i.e., $l = p(h)$ where $p(\cdot)$ is a polynomial of constant degree $d$.
    It requires $\Theta(\log n)$ rounds, w.h.p.
\end{theorem}

\begin{proof}
    Suppose that the amoebot structure is indeed a non-degenerate parallelogram and that it has already verified this fact by Lemma~\ref{lem:parallelogram}.
    Otherwise, a negative result can be returned.
    Determine an initial position for the robot by leader election and move it to the north-western corner.

    We make use of another algorithm stated by~\cite{DBLP:conf/mfcs/GmyrHKKRS18} that is outlined in the following to a level of detail sufficient to understand the correctness of our circuit based algorithm. 

    Let $(x)_i := x\cdot(x-1)\cdots(x-i+1)$ be the falling factorial of $x$, $\lcm_i(x) := \lcm(x,\dots,x-1+1)$ where $\lcm$ is the least common multiple and $g_i(x) := (x)_i/\lcm_i(x)$.
    Note that $\lcm_i(x) | (x)_i$ and $g_i(x)$ is periodic with period $\lcm(1,\dots,i-1)$~\cite{hong2008periodicity}.
    Thus, all possible values of $g_i(\cdot)$ for all $i \in \{ 0,\dots,d\}$ can be encoded within the finite state machine of the robot.
    By traversing the parallelogram once from north to south, it can determine values $g_i(h)$ for all $i \in \{0,\dots,d\}$.
    Thereafter, it moves back to the north-western corner.
    
    Then, the robot transforms the input polynomial into the form $p(h) = a_d \cdot (h)_d + a_{d-1} \cdot (h)_{d-1} + \dots + a_0$.
    Sequentially from $d$ to 0, it moves the pebble by $|a_i \cdot (h)_i|$ along the northern side.
    Note that $a_i \cdot (h)_i = a_i \cdot g_i(h) \cdot \lcm_i(h)$.
    The algorithm returns a positive result if the robot is located at the north-eastern corner at the end of the algorithm.

    During this process, the robot stays in the interval $[0, 2l]$ if $h$ is sufficiently large such that $|a_i \cdot (h)_i + \dots + a_0| \leq p(h)$ for all $i \in \{0,\dots,d\}$.
    The robot mirrors its movements for the second half of the interval so that for any value $l+i$, the pebble is positioned at node $l-i$.
    Whenever the robot tries to exceed the interval, a negative result can be returned.
    If $h$ is not sufficiently large, i.e., $h = O(\max_i(|a_i|))$, then the robot can compute $p(h)$ and thereby, check $l = p(h)$ directly.
    
    Next, we show how to accelerate the algorithm with the help of circuits.
    We assume that the amoebot structure agrees on a common chirality and on a common orientation.
    Otherwise, the chirality agreement and compass alignment algorithms are performed (see Sections~\ref{subsec:chirality} and~\ref{subsec:compass}).
    This requires $O(\log n)$ rounds, w.h.p.
    Contrary to the orignial algorithm,
    a virtual copy of the parallelogram is utilized for the interval $[l,2l]$ so that whenever the simulated robot trespasses the eastern side it reenters on the western side (see Figure~\ref{fig:polynomial}).
    
    Due to the periodicity of $g_i(\cdot)$, it is enough to determine $h'_i = h \mod \lcm(1,\dots,i-1)$ for all $i \in \{ 0,\dots,d\}$.
    These values can be computed in time $O(\log n)$ by establishing a \emph{binary add tree structure} iteratively.
    
    We describe the general idea of the technique.
    In a preparatory step, we apply the pin labeling primitive (see Section~\ref{sec:preliminaries}).
    This takes $\Theta(\log n)$ rounds, w.h.p.
    Consider any connected chain $C$ of $n$ amoebots where, initially, each amoebot $u$ holds an value $x_u \in \{ 0, \dots, k-1\}$ for some constant $k$.
    The goal is to compute the sum of all values modulo $k$, i.e., $x = \sum_{u\in C} x_u \mod k$.
    Each amoebot knows its predecessor and its successor in the chain.
    Let $L_0$ contain every amoebot of the chain.
    For each $i \geq 1$, let $L_i$ denote the set of amoebots holding partial sums of $x$ after the $i$-th iteration.
    
    Iteration $i \geq 1$ is performed as follows.
    Each amoebot $u \in L_{i-1}$ connects its pins crossed,
    i.e., the first pin to the predecessor to the second pin to the successor and vice versa.
    Each amoebot $u \not\in L_{i-1}$ connects its pins in parallel (see Figure~\ref{fig:add_tree}).
    Then, the first amoebot of the chain activates the circuit connected to the first pin to its successor.
    Each amoebot $u \in L_{i-1}$ that has received a beep on the circuit connected to its first pin to its successor is added to $L_i$.
    Obviously, $L_i$ contains every second amoebot of $L_{i-1}$.
    
    Thereafter, each amoebot $u \in L_{i-1}$ removes all of its connections.
    Consider $k$ subsequent rounds.
    Each amoebot $u \in L_{i-1} \setminus L_i$ beeps on the round $r = x_u + 1$ on the circuit connected to one of the pins of the predecessor.
    This signal is received by the preceding amoebot $u \in L_i$.
    Each amoebot $u \in L_i$ adds the respective value to its own.
    Clearly, the first amoebot of the chain holds the total sum after $O(\log n)$ rounds.
    
    Now, consider the western side as the chain (see Figure~\ref{fig:add_tree}).
    For each amoebot, we define the northern neighbor as the predecessor an the southern neighbor as the successor.
    Initially, each amoebot except the first one holds the value 1.
    The first amoebot is initialized with 0.
    We apply the binary add tree technique with $k = \lcm(1, \dots, i-1)$ for all $i \in \{ 0, \dots, d\}$ to compute $h'_i$ at the north-western corner of the parallelogram, which can compute the corresponding $g_i(h)$ values.
    This takes $O(d \cdot \log h) = O(\log n)$ rounds.
    
    It remains to show how to simulate the movement of the pebble by $|a_i \cdot g_i(h) \cdot \lcm_i(h)|$ nodes with the help of circuits.
    The multiples of $h-i$ for all $i \in \{0,\dots,d\}$ are identified during a preprocessing phase that, for each $i \in \{0,\dots,d\}$, constructs the following circuit:
    (i)~amoebots on the northern side connect the south-eastern pin to the southern pin,
    (ii)~amoebots with a distance greater than $i$ from the southern side connect the north-western pin to the south-eastern pin and the northern pin to the southern pin, and
    (iii) amoebots with distance $i$ from the southern side connect the north-western pin to the northern pin (see Figure~\ref{fig:polynomial}).
    Note that the information about the distance from the southern side is only required up to distance $d$.
    
    Then, the north-western corner activates its circuit to the south-east.
    Amoebots on the northern side that receive a signal are the multiples of $h-i$ since the circuit is composed of equilateral triangles of length $h-i$.
    To identify the multiples for the second half of the interval, i.e., $[l,2l]$, we have to continue the circuit on the virtual copy.
    First, the amoebot on the eastern side that received the beep activates the circuit along the north-eastern axis.
    Second, the amoebot on the western side that received the beep activates the circuit to the south-east.
    Again, amoebots on the northern side that receive a signal are the multiples of $h-i$.
    %The activation is transfered back to the western side where the corresponding circuit is activated to indentify the multiples for the second half of the interval, i.e., $[l,2l]$.
    The preprocessing takes $O(d)$ rounds.
    
    Finally, the pebble can be moved by $|a_i \cdot g_i(h) \cdot \lcm_i(h)|$ nodes as follows.
    Note that each amoebot can determine if its position relative to the western side is a multiple of $\lcm_i(h)$.
    A circuit along the northern side is split at these amoebots.
    The pebble is transfered along these circuits into the respective direction for $|a_i \cdot g_i(h)|$ times.
    Thus, $O(\sum_i |a_i \cdot g_i(\cdot)|)$ rounds are required.
\end{proof}

\begin{remark}
    We may utilize the same structure formed for the binary add tree as a \emph{binary search structure}.
    Let $k$ be constant.
    Set $L_0$ contains only the amoebots we are interested in.
    The remaining amoebots are just used to transfer the beep on the chain.
    Sets $L_i$ for all $i \in \{ 1, \dots, k\}$ are determined as described for the binary add tree structure.
    We can apply a binary search on these sets starting from set $L_k$.
    Note that due to the constant-size memory of the amoebots, $k$ has to be constant.
    This implies that $|L_0| \leq 2^k$ has to hold.
    
    The binary search structure can be utilized to improve the runtime of the presented algorithm in Theorem~\ref{th:linear}.
    Consider the zig-zag pattern in Figure~\ref{fig:linear}.
    Similar to the zig-zag patterns in Theorem~\ref{th:polynomial}, we can combine the path to a single chain.
    Let $L_0$ contain all amoebots on the northern side.
    The amoebot in the north-western corner may activate the circuits formed by the zig-zag pattern and along the north-eastern axis to determine the chain and the set $L_0$.
    Then, we use binary search to find the $a$-th element.
    The additional $b$ steps can be done analogously.
    The chain starts at the amoebot of distance $ah$ and contains all amoebots on the northern side into the respective direction.
    Note that $a$ and $b$ are constant such that $k = O(\log a + \log b) = O(1)$ iterations suffice.
    Alltogether, the runtime of Theorem~\ref{th:linear} improves to $O(\log a + \log b) = O(1)$.
    
    Furthermore, we may apply the binary search structure in the algorithm of Theorem~\ref{th:polynomial} to determine the $|a_i \cdot g_i(h)|$-th multiple.
    The chain starts at the current position of the pebble and contains all amoebots on the northern side into the respective direction.
    Set $L_0$ contains only the multiples.
    Again, note that $|a_i \cdot g_i(h)|$ is constant such that $k = O(\log |a_i \cdot g_i(h)|) = O(1)$ iterations suffice.
    Note, however, that does not lead to an improved overall runtime.
    
    %We believe that this structure can be of independent interest and can be utilized in a variety of other problems.
\end{remark}

\begin{figure}[htp]
    \centering
    \begin{subfigure}[b]{0.45\textwidth}
        \centering
        \begin{tikzpicture}
            \tikzmath{
                \scale = .75;
                \t = \scale/2;
                \h = sqrt(\scale^2 - (\scale/2)^2);
                {
                    \filldraw[yellow]
                    (0*\h, 0*\t) --
                    (0*\h, 6*\t) --
                    (9*\h, 15*\t) --
                    (9*\h, 9*\t) --
                    cycle;
                    \fill[fill={rgb,255:red,255; green,200; blue,0}]
                    (0*\h, 6*\t) --
                    (2*\h, 4*\t) --
                    (2*\h, 8*\t) --
                    (4*\h, 6*\t) --
                    (4*\h, 10*\t) --
                    (6*\h, 8*\t) --
                    (6*\h, 12*\t) --
                    (8*\h, 10*\t) --
                    (8*\h, 14*\t) --
                    (9*\h, 13*\t) --
                    (9*\h, 15*\t) --
                    cycle;
                };
                for \j in {0,1,...,9}{
                    for \i in {0,2,4,6}{
                        {
                            \filldraw (\j*\h, \i*\t+\j*\t) circle (0.04);
                        };
                    };
                    {
                        \draw[dotted] (\j*\h, 0*\t+\j*\t) -- (\j*\h, 6*\t+\j*\t);
                    };
                };
                for \j in {0,1,...,8}{
                    for \i in {0,4,6}{
                        {
                            \draw[dotted] (\j*\h, \i*\t+\j*\t) -- (\j*\h+\h, \i*\t+\j*\t+\t);
                        };
                    };
                    for \i in {2,4,6}{
                        {
                            \draw[dotted] (\j*\h, \i*\t+\j*\t) -- (\j*\h+\h, \i*\t+\j*\t-\t);
                        };
                    };
                };
                {
                    \draw[dashed] (0*\h, 2*\t) -- (9*\h, 11*\t);
                    \draw[thick, color = green]
                    (0*\h, 6*\t) --
                    (2*\h, 4*\t) --
                    (2*\h, 8*\t) --
                    (4*\h, 6*\t) --
                    (4*\h, 10*\t) --
                    (6*\h, 8*\t) --
                    (6*\h, 12*\t) --
                    (8*\h, 10*\t) --
                    (8*\h, 14*\t) --
                    (9*\h, 13*\t);
                    \draw (0*\h, 6*\t) circle (0.1);
                    \draw (2*\h, 8*\t) circle (0.1);
                    \draw (4*\h, 10*\t) circle (0.1);
                    \draw (6*\h, 12*\t) circle (0.1);
                    \draw (8*\h, 14*\t) circle (0.1);
                };
            }
        \end{tikzpicture}
        \caption*{\centering $[0,l]$}
    \end{subfigure}
    \hfill
    \begin{subfigure}[b]{0.45\textwidth}
        \centering
        \begin{tikzpicture}
            \tikzmath{
                \scale = .75;
                \t = \scale/2;
                \h = sqrt(\scale^2 - (\scale/2)^2);
                {
                    \filldraw[yellow]
                    (0*\h, 0*\t) --
                    (0*\h, 6*\t) --
                    (9*\h, 15*\t) --
                    (9*\h, 9*\t) --
                    cycle;
                    \fill[fill={rgb,255:red,255; green,200; blue,0}]
                    (0*\h, 6*\t) --
                    (0*\h, 4*\t) --
                    (1*\h, 3*\t) --
                    (1*\h, 7*\t) --
                    (3*\h, 5*\t) --
                    (3*\h, 9*\t) --
                    (5*\h, 7*\t) --
                    (5*\h, 11*\t) --
                    (7*\h, 9*\t) --
                    (7*\h, 13*\t) --
                    (9*\h, 11*\t) --
                    (9*\h, 15*\t) --
                    cycle;
                };
                for \j in {0,1,...,9}{
                    for \i in {0,2,4,6}{
                        {
                            \filldraw (\j*\h, \i*\t+\j*\t) circle (0.04);
                        };
                    };
                    {
                        \draw[dotted] (\j*\h, 0*\t+\j*\t) -- (\j*\h, 6*\t+\j*\t);
                    };
                };
                for \j in {0,1,...,8}{
                    for \i in {0,4,6}{
                        {
                            \draw[dotted] (\j*\h, \i*\t+\j*\t) -- (\j*\h+\h, \i*\t+\j*\t+\t);
                        };
                    };
                    for \i in {2,4,6}{
                        {
                            \draw[dotted] (\j*\h, \i*\t+\j*\t) -- (\j*\h+\h, \i*\t+\j*\t-\t);
                        };
                    };
                };
                {
                    \draw[dashed] (0*\h, 2*\t) -- (9*\h, 11*\t);
                    \draw[thick, color = green]
                    (0*\h, 4*\t) --
                    (1*\h, 3*\t) --
                    (1*\h, 7*\t) --
                    (3*\h, 5*\t) --
                    (3*\h, 9*\t) --
                    (5*\h, 7*\t) --
                    (5*\h, 11*\t) --
                    (7*\h, 9*\t) --
                    (7*\h, 13*\t) --
                    (9*\h, 11*\t) --
                    (9*\h, 15*\t);
                    \draw (1*\h, 7*\t) circle (0.1);
                    \draw (3*\h, 9*\t) circle (0.1);
                    \draw (5*\h, 11*\t) circle (0.1);
                    \draw (7*\h, 13*\t) circle (0.1);
                    \draw (9*\h, 15*\t) circle (0.1);
                };
            }
        \end{tikzpicture}
        \caption*{\centering $[l,2l]$}
    \end{subfigure}
    \caption{
        Let $p(h) = 3 \cdot h \cdot (h-1) - 2 \cdot h - 3$.
        A parallelogram with height $h = 3$ and length $l = p(3) = 9$ is given.
        The green line indicates the activated circuit.
        The dashed line indicates distance $i$.
        The equilateral triangles are depicted in orange.
        The multiples of $h-i = 2$ are marked by circles.
    }
    \label{fig:polynomial}
\end{figure}
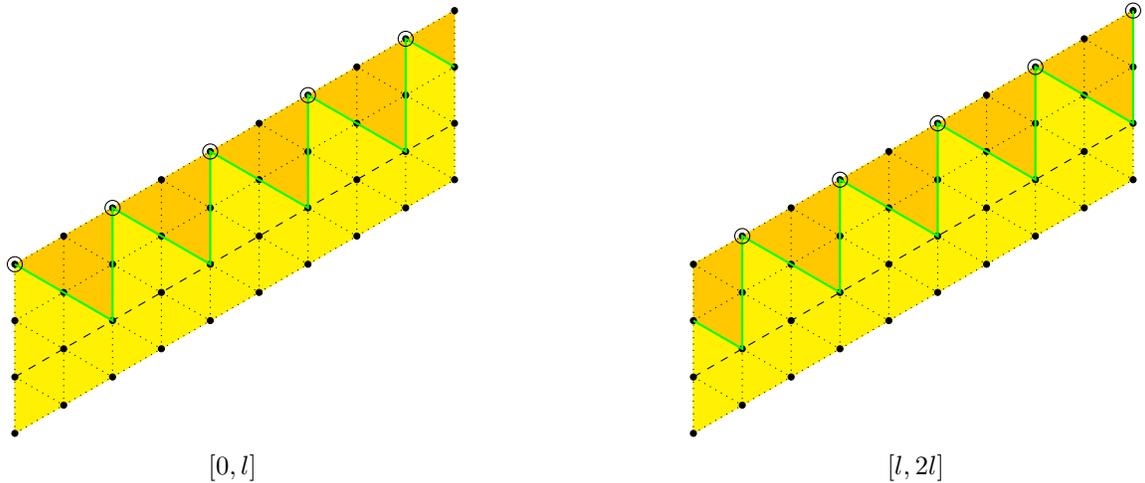

\section{Universal Shape Recognition}
\label{app:recognition}

We give a detailed analysis on the universal shape recognition algorithm described in Section~\ref{subsec:recognition}.
Consider the first phase.
Figure~\ref{fig:message1} shows an exemplary message.
The following lemma implies the correctness of the phase.

\begin{lemma}
\label{lem:usr:gathering1}
    Each amoebot knows the positions of all amoebots up to distance $i + 1$ after the $i$-th iteration.
    In particular, each amoebot knows the positions of all amoebots up to distance $3 \cdot |\mathcal S| + 1$ after the $3 \cdot |\mathcal S|$-th iteration.
\end{lemma}

\begin{proof}
    We prove the statement by induction.
    The induction base $i = 0$ holds trivially since each amoebot knows its own neighborhood, i.e., the positions of all amoebots up to distance $1$.
    Suppose that the statement holds up to the $i$-th iteration.
    Consider the $(i + 1)$-st iteration.
    By the induction hypothesis, each amoebot knows the positions of all amoebots up to distance $i + 1$.
    Each amoebot sends these positions to each neighbor.
    Each amoebot can combine the received positions to obtain the positions of all amoebots up to distance $i + 2$.
\end{proof}

Note that the amoebot structure is connected by assumption.
Each amoebot can directly compare the amoebot structure to the representations of the shape with $\sigma < 4$, which are encoded within the finite state machines of the amoebots.
Recall that $|\mathcal S|$ is constant.
Hence, the number of iterations and the size of the messages are bounded by a constant.
Thus, the first phase terminates after $O(1)$ rounds.

Next, consider the second phase.
Note that each amoebot knows its 2-neigh\-bor\-hood after the first phase.
They are therefore able to categorize themselves immediately in the second phase.
We utilize an additional round to notify the whole amoebot structure whether each amoebot was able to categorize itself by letting these activate the global circuit.
Thus, the second phase terminates after $O(1)$ rounds.
The following lemma proves the correctness of the phase.

\begin{lemma}
\label{lem:usr:boundary}
    The second phase behaves as follows.
    \begin{enumerate}
        \item[(i)] On any representation of the shape with $\sigma \geq 4$, the phase detects the boundary.
        \item[(ii)] On any representation of the shape with $\sigma < 4$, the phase behaves arbitrarily.
        \item[(iii)] On any other amoebot structure, the phase either fails to detect a boundary or detects a boundary unequal to the shape.
    \end{enumerate}
\end{lemma}

\begin{proof}
    (i) The statement holds by construction.
    The classes include all occurring 2-neigh\-bor\-hoods.
    Note that the classes are disjoint.
    (ii) The statement holds trivially.
    (iii) Suppose that the phase detects a boundary that is equal to the shape.
    Otherwise, the statement holds trivially.
    Since the amoebot structure is not a representation of the shape, either an amoebot has to be located outside the shape or an amoebot has to be missing within the shape.
    Both cases lead to additional boundary edges that contradict the equality.
\end{proof}

\begin{remark}
    We do not care about the outcome for representations of the shape with $\sigma < 4$ because these are eliminated during the first phase.
\end{remark}

Next, we prove the correctness and the runtime of the third phase.
The correctness is given by the following lemma.

\begin{lemma}
\label{lem:usr:triangulation}
    Algorithm~\ref{alg:triangulation} computes a minimal triangulation of each face.
\end{lemma}

\begin{proof}%[Proof of Lemma~\ref{lem:usr:triangulation}]
    First, note that the algorithm always terminates since there are only finite many amoebots.
    However, the algorithm works independently on each face.
    The only possibility for interactions are the circuits.
    But all connections at the boundaries are removed.
    So, consider a single face.
    
    The algorithm divides the given amoebot structure into faces.
    The vertices of the faces are given by $V$.
    The edges of the faces are given by the activated circuits.
    Note that the circuits along the boundaries are activated within the first iteration.
    
    The degree of each amoebot in $V$ is equal to the number of its neighbors since it sends a beep on all circuits along the axes.
    This implies that all corners of the faces have an angle of $60^\circ$.
    Thus, all faces are indeed triangles.
    Furthermore, all triangles have the same size.
    Suppose the contrary.
    Since the shape is connected, there have to be 2 triangles of different size with a common edge.
    At least one of the vertices of the smaller triangle has to lie on this edge.
    This vertex would activate all circuits along the axes (see Figure~\ref{fig:triangulation1}).
    This contradicts the existence of the greater triangle.
    Furthermore, the triangles coincide with the triangles of $G_{eqt}$.
    Otherwise, we can apply the same argument (see Figure~\ref{fig:triangulation2}).
    
    Finally, the number of triangles is minimal.
    Note that the boundary corners are necessary for all triangulations.
    Additional corners are only induced if not avoidable.
    Edges are added according to Observation~\ref{obs:incident_edges}.
\end{proof}

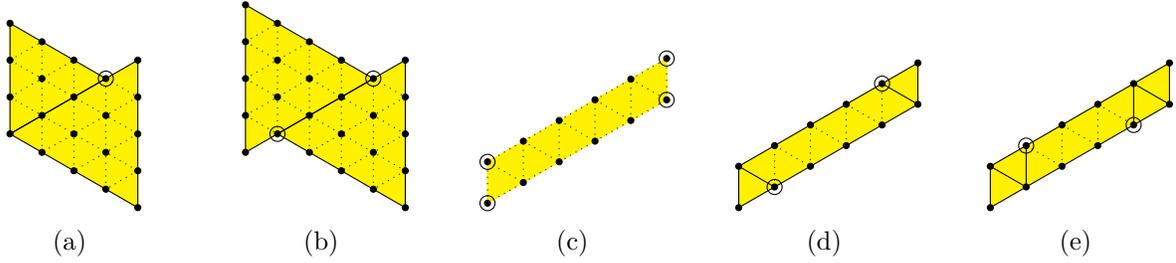
\begin{figure}[htb]
    \centering
    \begin{subfigure}[b]{0.19\columnwidth}
        \centering
        \begin{tikzpicture}
            \tikzmath{
                \scale = .3/2/11*36;
                \t = \scale/2;
                \h = sqrt(\scale^2 - (\scale/2)^2);
                % shape
                {
                    \filldraw[yellow]
                    (1*\h, 4*\t) --
                    (5*\h, 8*\t) --
                    (5*\h, 0*\t) --
                    cycle;
                    \filldraw[yellow]
                    (4*\h, 7*\t) --
                    (1*\h, 4*\t) --
                    (1*\h, 10*\t) --
                    cycle;
                };
                % particles
                for \i in {4,6,8,10}{
                    {
                        \filldraw (1*\h, \i*\t) circle (0.04);
                    };
                };
                for \i in {3,5,7,9}{
                    {
                        \filldraw (2*\h, \i*\t) circle (0.04);
                    };
                };
                for \i in {2,4,6,8}{
                    {
                        \filldraw (3*\h, \i*\t) circle (0.04);
                    };
                };
                for \i in {1,3,5,7}{
                    {
                        \filldraw (4*\h, \i*\t) circle (0.04);
                    };
                };
                for \i in {0,2,4,6,8}{
                    {
                        \filldraw (5*\h, \i*\t) circle (0.04);
                    };
                };
                {
                % corners
                    \draw (4*\h, 7*\t) circle (0.1);
                % bonds
                    \draw
                    (1*\h, 4*\t) --
                    (5*\h, 8*\t) --
                    (5*\h, 0*\t) --
                    cycle;
                    \draw
                    (4*\h, 7*\t) --
                    (1*\h, 4*\t) --
                    (1*\h, 10*\t) --
                    cycle;
                    % N - S
                    \draw[dotted] (2*\h, 3*\t) -- (2*\h, 9*\t);
                    \draw[dotted] (3*\h, 2*\t) -- (3*\h, 8*\t);
                    \draw[dotted] (4*\h, 1*\t) -- (4*\h, 7*\t);
                    % SW - NE
                    \draw[dotted] (1*\h, 6*\t) -- (3*\h, 8*\t);
                    \draw[dotted] (1*\h, 8*\t) -- (2*\h, 9*\t);
                    \draw[dotted] (5*\h, 2*\t) -- (4*\h, 1*\t);
                    \draw[dotted] (5*\h, 4*\t) -- (3*\h, 2*\t);
                    \draw[dotted] (5*\h, 6*\t) -- (2*\h, 3*\t);
                    % NW - SE
                    \draw[dotted] (4*\h, 7*\t) -- (5*\h, 6*\t);
                    \draw[dotted] (1*\h, 8*\t) -- (5*\h, 4*\t);
                    \draw[dotted] (1*\h, 6*\t) -- (5*\h, 2*\t);
                };
            }
        \end{tikzpicture}
        \caption{\centering}
        \label{fig:triangulation1}
    \end{subfigure}
    \hfill
    \begin{subfigure}[b]{0.19\columnwidth}
        \centering
        \begin{tikzpicture}
            \tikzmath{
                \scale = .3/2/11*36;
                \t = \scale/2;
                \h = sqrt(\scale^2 - (\scale/2)^2);
                % shape
                {
                    \filldraw[yellow]
                    (1*\h, 4*\t) --
                    (5*\h, 8*\t) --
                    (5*\h, 0*\t) --
                    cycle;
                    \filldraw[yellow]
                    (4*\h, 7*\t) --
                    (0*\h, 3*\t) --
                    (0*\h, 11*\t) --
                    cycle;
                };
                % particles
                for \i in {3,5,7,9,11}{
                    {
                        \filldraw (0*\h, \i*\t) circle (0.04);
                    };
                };
                for \i in {4,6,8,10}{
                    {
                        \filldraw (1*\h, \i*\t) circle (0.04);
                    };
                };
                for \i in {3,5,7,9}{
                    {
                        \filldraw (2*\h, \i*\t) circle (0.04);
                    };
                };
                for \i in {2,4,6,8}{
                    {
                        \filldraw (3*\h, \i*\t) circle (0.04);
                    };
                };
                for \i in {1,3,5,7}{
                    {
                        \filldraw (4*\h, \i*\t) circle (0.04);
                    };
                };
                for \i in {0,2,4,6,8}{
                    {
                        \filldraw (5*\h, \i*\t) circle (0.04);
                    };
                };
                {
                % corners
                    \draw (1*\h, 4*\t) circle (0.1);
                    \draw (4*\h, 7*\t) circle (0.1);
                % bonds
                    \draw
                    (1*\h, 4*\t) --
                    (5*\h, 8*\t) --
                    (5*\h, 0*\t) --
                    cycle;
                    \draw
                    (4*\h, 7*\t) --
                    (0*\h, 3*\t) --
                    (0*\h, 11*\t) --
                    cycle;
                    % N - S
                    \draw[dotted] (1*\h, 4*\t) -- (1*\h, 10*\t);
                    \draw[dotted] (2*\h, 3*\t) -- (2*\h, 9*\t);
                    \draw[dotted] (3*\h, 2*\t) -- (3*\h, 8*\t);
                    \draw[dotted] (4*\h, 1*\t) -- (4*\h, 7*\t);
                    % SW - NE
                    \draw[dotted] (0*\h, 5*\t) -- (3*\h, 8*\t);
                    \draw[dotted] (0*\h, 7*\t) -- (2*\h, 9*\t);
                    \draw[dotted] (0*\h, 9*\t) -- (1*\h, 10*\t);
                    \draw[dotted] (5*\h, 2*\t) -- (4*\h, 1*\t);
                    \draw[dotted] (5*\h, 4*\t) -- (3*\h, 2*\t);
                    \draw[dotted] (5*\h, 6*\t) -- (2*\h, 3*\t);
                    % NW - SE
                    \draw[dotted] (4*\h, 7*\t) -- (5*\h, 6*\t);
                    \draw[dotted] (0*\h, 9*\t) -- (5*\h, 4*\t);
                    \draw[dotted] (0*\h, 7*\t) -- (5*\h, 2*\t);
                    \draw[dotted] (0*\h, 5*\t) -- (1*\h, 4*\t);
                };
            }
        \end{tikzpicture}
        \caption{\centering}
        \label{fig:triangulation2}
    \end{subfigure}
    \hfill
    \begin{subfigure}[b]{0.19\columnwidth}
        \centering
        \begin{tikzpicture}
            \tikzmath{
                \scale = .55;
                \t = \scale/2;
                \h = sqrt(\scale^2 - (\scale/2)^2);
                % shape
                {
                    \filldraw[yellow]
                    (0*\h, 0*\t+2*0*\t) --
                    (0*\h, 0*\t+2*1*\t) --
                    (5*\h, 5*\t+2*1*\t) --
                    (5*\h, 5*\t+2*0*\t) --
                    cycle;
                };
                % particles
                for \j in {0,1,2,3,4,5}{
                    for \i in {0,1}{
                        {
                            \filldraw (\j*\h, \j*\t+2*\i*\t) circle (0.04);
                        };
                    };
                };
                % bonds
                for \j in {0,1,2,3,4,5}{
                    {
                        \draw[dotted] (\j*\h, \j*\t) -- (\j*\h, \j*\t+2*\t);
                    };
                };
                for \j in {0,1,2,3,4}{
                    {
                        \draw[dotted] (\j*\h, \j*\t+2*\t) -- (\j*\h+\h, \j*\t+3*\t);
                        \draw[dotted] (\j*\h, \j*\t+2*\t) -- (\j*\h+\h, \j*\t+\t);
                        \draw[dotted] (\j*\h, \j*\t) -- (\j*\h+\h, \j*\t+\t);
                    };
                };
                % corners
                {
                    \draw (0*\h, 0*\t+2*0*\t) circle (0.1);
                    \draw (0*\h, 0*\t+2*1*\t) circle (0.1);
                    \draw (5*\h, 5*\t+2*1*\t) circle (0.1);
                    \draw (5*\h, 5*\t+2*0*\t) circle (0.1);
                };
            }
        \end{tikzpicture}
        \caption{\centering}
        \label{fig:triangulation3}
    \end{subfigure}
    \hfill
    \begin{subfigure}[b]{0.19\columnwidth}
        \centering
        \begin{tikzpicture}
            \tikzmath{
                \scale = .55;
                \t = \scale/2;
                \h = sqrt(\scale^2 - (\scale/2)^2);
                % shape
                {
                    \filldraw[yellow]
                    (0*\h, 0*\t+2*0*\t) --
                    (0*\h, 0*\t+2*1*\t) --
                    (5*\h, 5*\t+2*1*\t) --
                    (5*\h, 5*\t+2*0*\t) --
                    cycle;
                };
                % particles
                for \j in {0,1,2,3,4,5}{
                    for \i in {0,1}{
                        {
                            \filldraw (\j*\h, \j*\t+2*\i*\t) circle (0.04);
                        };
                    };
                };
                % bonds
                for \j in {0,5}{
                    {
                        \draw (\j*\h, \j*\t) -- (\j*\h, \j*\t+2*\t);
                    };
                };
                for \j in {1,2,3,4}{
                    {
                        \draw[dotted] (\j*\h, \j*\t) -- (\j*\h, \j*\t+2*\t);
                    };
                };
                for \j in {0,4}{
                    {
                        \draw (\j*\h, \j*\t+2*\t) -- (\j*\h+\h, \j*\t+3*\t);
                        \draw (\j*\h, \j*\t+2*\t) -- (\j*\h+\h, \j*\t+\t);
                        \draw (\j*\h, \j*\t) -- (\j*\h+\h, \j*\t+\t);
                    };
                };
                for \j in {1,2,3}{
                    {
                        \draw (\j*\h, \j*\t+2*\t) -- (\j*\h+\h, \j*\t+3*\t);
                        \draw[dotted] (\j*\h, \j*\t+2*\t) -- (\j*\h+\h, \j*\t+\t);
                        \draw (\j*\h, \j*\t) -- (\j*\h+\h, \j*\t+\t);
                    };
                };
                % corners
                {
                    \draw (1*\h, 1*\t+2*0*\t) circle (0.1);
                    \draw (4*\h, 4*\t+2*1*\t) circle (0.1);
                };
            }
        \end{tikzpicture}
        \caption{\centering}
        \label{fig:triangulation4}
    \end{subfigure}
    \hfill
    \begin{subfigure}[b]{0.19\columnwidth}
        \centering
        \begin{tikzpicture}
            \tikzmath{
                \scale = .55;
                \t = \scale/2;
                \h = sqrt(\scale^2 - (\scale/2)^2);
                % shape
                {
                    \filldraw[yellow]
                    (0*\h, 0*\t+2*0*\t) --
                    (0*\h, 0*\t+2*1*\t) --
                    (5*\h, 5*\t+2*1*\t) --
                    (5*\h, 5*\t+2*0*\t) --
                    cycle;
                };
                % particles
                for \j in {0,1,2,3,4,5}{
                    for \i in {0,1}{
                        {
                            \filldraw (\j*\h, \j*\t+2*\i*\t) circle (0.04);
                        };
                    };
                };
                % bonds
                for \j in {0,1,4,5}{
                    {
                        \draw (\j*\h, \j*\t) -- (\j*\h, \j*\t+2*\t);
                    };
                };
                for \j in {2,3}{
                    {
                        \draw[dotted] (\j*\h, \j*\t) -- (\j*\h, \j*\t+2*\t);
                    };
                };
                for \j in {0,4}{
                    {
                        \draw (\j*\h, \j*\t+2*\t) -- (\j*\h+\h, \j*\t+3*\t);
                        \draw (\j*\h, \j*\t+2*\t) -- (\j*\h+\h, \j*\t+\t);
                        \draw (\j*\h, \j*\t) -- (\j*\h+\h, \j*\t+\t);
                    };
                };
                for \j in {1,2,3}{
                    {
                        \draw (\j*\h, \j*\t+2*\t) -- (\j*\h+\h, \j*\t+3*\t);
                        \draw[dotted] (\j*\h, \j*\t+2*\t) -- (\j*\h+\h, \j*\t+\t);
                        \draw (\j*\h, \j*\t) -- (\j*\h+\h, \j*\t+\t);
                    };
                };
                % corners
                {
                    \draw (1*\h, 1*\t+2*1*\t) circle (0.1);
                    \draw (4*\h, 4*\t+2*0*\t) circle (0.1);
                };
            }
        \end{tikzpicture}
        \caption{\centering}
        \label{fig:triangulation5}
    \end{subfigure}
    \caption{
        (a) shows the case where the triangles do not all have the same size, and (b) shows the case where the triangles do not coincide with the triangles of $G_{eqt}$.
        In both cases the algorithm identifies the circled nodes as corners.
        (c), (d), and (e) show the first iterations of Algorithm~\ref{alg:triangulation} on a parallelogram with height 1 and length $n/2 = 5$.
        Each iteration adds 2 edges.
    }
\end{figure}

The next two lemmas deal with the runtime of the phase.
Lemma~\ref{lem:usr:arbitrary} shows that the runtime of Algorithm~\ref{alg:triangulation} is linear in the worst case.
However, Lemma~\ref{lem:usr:bound} shows that we can bound the runtime of the phase by a constant since $|\mathcal S|$ is constant.
Thus, the third phase terminates after $O(1)$ rounds.

\begin{lemma}
\label{lem:usr:arbitrary}
    Algorithm~\ref{alg:triangulation} may take $\Omega(n)$ rounds on an arbitrary amoebot structure.
\end{lemma}

\begin{proof}
    Consider a parallelogram with height 1 and length $n/2$ (see Figures~\ref{fig:triangulation3} to~\ref{fig:triangulation5}).
    The minimal triangulation has $\Omega(n)$ edges.
    Algorithm~\ref{alg:triangulation} only adds two edges in each iteration.
\end{proof}

\begin{lemma}
\label{lem:usr:bound}
    The runtime of Algorithm~\ref{alg:triangulation} does only depend on the shape and not on its representation.
\end{lemma}

\begin{proof} 
    The algorithm only employs circuits.
    The delay of a circuit does not depend on its size.
\end{proof}

Now, we consider the fourth phase.
Figure~\ref{fig:message2} shows an exemplary message.
The following lemma shows the correctness of the information gathering.

\begin{lemma}
\label{lem:usr:gathering2}
    Each triangle corner knows the positions of all triangles up to distance $i + 1$ after the $i$-th iteration.
    In particular, each triangle corner knows the positions of all triangles up to distance $|\mathcal S| + 1$ after the $|\mathcal S|$-th iteration.
\end{lemma}

\begin{proof}
    The proof works analogously to Lemma~\ref{lem:usr:gathering1}.
\end{proof}

Contrary to the first phase, we cannot assume that the triangles are connected.
We therefore gather the information for each face separately.
In order to do that, amoebots representing multiple boundary corners perform the information gathering for each boundary corner separately.
The boundary corners of these amoebots have to coincide in all gathered triangulations.
Otherwise, there is more than one face.
Hence, these amoebots can detect multiple faces and notify the whole amoebot structure about that by beeping the global circuit in an additional round.
If the global circuit is not activated, there is only a single face.
In this case, the amoebots can directly compare the triangulation to the shape.
Note that different triangulation would lead to different results.
However, the following corollary of Lemma~\ref{lem:usr:triangulation} shows that the triangulation is indeed unique.

\begin{corollary}
\label{cor:usr:triangulation}
    The minimal triangulation of each face is unique.
\end{corollary}

\begin{proof}%[Proof of Corollary~\ref{cor:usr:triangulation}]
    The statement follows directly from Lemma~\ref{lem:usr:triangulation}.
    Any other minimal triangulation would have to include all corners determined by Algorithm~\ref{alg:triangulation} since it only includes unavoidable corners.
\end{proof}

The runtime analysis is analogous to the first phase with the addition of the previously mentioned round where the amoebots beep the global circuit in order to notify the whole amoebot structure about multiple faces.
Thus, the fourth phase terminates after $O(1)$ rounds.

Finally, we can prove the Theorem~\ref{th:usr}.

\begin{proof}[Proof of Theorem~\ref{th:usr}]
    The correctness of the theorem follows from the correctness of the phases.
    All phases only require a constant number of rounds.
    Note that we have to perform the chirality agreement algorithm if the amoebots do not share a common chirality.
\end{proof}

\end{document}